\documentclass[onecolumn]{article}
\pdfoutput=1
\usepackage{graphicx, hyperref}
\usepackage{amsmath}
\usepackage{amssymb}
\usepackage{slashed}
\usepackage{stmaryrd}
\usepackage{yfonts}
\usepackage[T1]{fontenc}
\usepackage[utf8]{inputenc}
\usepackage[english]{babel}
\usepackage{amsthm}
\usepackage{tikz-cd}
\usepackage{mathabx}
\usepackage{braket}
\usepackage{physics}

\usepackage{tikz}
\usetikzlibrary{shapes}
\usepackage{xcolor}

\usepackage{geometry}
\newtheorem{theorem}{Theorem}[section]
\newtheorem{lemma}{Lemma}[section]

\newtheorem{proposition}{Proposition}[section]

\theoremstyle{remark}
\newtheorem{remark}{Remark}[section]

\theoremstyle{definition}
\newtheorem{definition}{Definition}[section]

\usepackage{mathtools}

\usepackage{mathrsfs}
\numberwithin{equation}{section}

\newcommand{\beq}{\begin{eqnarray}}
\newcommand{\eeq}{\end{eqnarray}}

\hypersetup{
pdfstartview = {FitH},
}
\hypersetup{
	colorlinks=true,         
	linkcolor=brown,          
	citecolor=red,        
	urlcolor=blue            
}

\def\g{\mathfrak g}

\def\T{\mathfrak{T}}

\def\bbZ{\mathbb{Z}}
\def\R{\mathbb{R}}
\def\bbC{\mathbb{C}}

\newcommand{\cA}{{\cal A}}
\newcommand{\cB}{{\cal B}}
\newcommand{\cC}{{\cal C}}
\newcommand{\cD}{{\cal D}}

\newcommand{\cT}{{\cal T}}

\newcommand{\cH}{{\cal H}}

\newcommand{\cK}{{\cal K}}
\newcommand{\cM}{{\cal M}}
\newcommand{\cN}{{\cal N}}

\newcommand{\cR}{{\cal R}}

\def\id{\operatorname{id}}

\usepackage{atbegshi}% http://ctan.org/pkg/atbegshi
\AtBeginDocument{\AtBeginShipoutNext{\AtBeginShipoutDiscard}\addtocounter{page}{-1}}

\usepackage{authblk}
  \allowdisplaybreaks

\begin{document}

\title{Drinfel'd Double Symmetry of the   4d Kitaev Model}
\author[1]{{ \sf Hank Chen}}\thanks{hank.chen@uwaterloo.ca}
% \affiliation{Department of Applied Mathematics, University of Waterloo, Waterloo, Ontario, Canada}

% \author[1]{\sffamily Florian Girelli\thanks{florian.girelli@uwaterloo.ca}}
% \author[1]{\sffamily Panagiotis Tsimiklis\thanks{ptsimiklis@uwaterloo.ca}}

% \author[1]{{\sf Florian Girelli}}\thanks{fgirelli@uwaterloo.ca}

\affil[1]{\small Department of Applied Mathematics, University of Waterloo, 200 University Avenue West, Waterloo, Ontario, Canada, N2L 3G1}

% \author{Hank Chen}
\maketitle

\begin{abstract}
    Following the general theory of categorified quantum groups developed by the author previously, we construct the Drinfel'd double 2-bialgebra associated to a finite group $N=G_0$. For $N=\bbZ_2$, we explicitly compute the braided 2-categories of 2-representations of certain version of this Drinfel'd double 2-bialgebra, and prove that they characterize precisely the  4d toric code and its spin-$\mathbb{Z}_2$ variant. This result relates the two descriptions (categorical vs. field theoretical) of 4d gapped topological phases in existing literature and displays an instance of {\it higher Tannakian duality} for braided 2-categories. In particular, we show that particular twists of the underlying Drinfel'd double 2-bialgebra is responsible for much of the higher-structural properties that arise in 4d topological orders.
    
\end{abstract}

\tableofcontents

\section{Introduction}
It is well-known that the charge algebra of excitations in the  3d BF theory (or equivalently  3d Chern-Simons \cite{Osei:2013xra,chen:2022} theory) associated to an ordinary Lie group $G$ is described by its Jimbo-Drinfel'd $q$-deformed double $D_q(G)$ \cite{Delcamp:2016yix,Dupuis:2020ndx,Majid:1990bt,Witten:1988hc}. In particular, the representations of the Drinfel'd double is known to form a braided tensor category \cite{Majid:1994nw}, the algebraic structures of which is of tantamount importance in many areas of physics \cite{KitaevKong_2012,Dupuis:2020ndx,Severa:2005}. Of particular interest that serves as the main motivation for this paper is the theory of topological ordering in condensed matter. 

More precisely, it is well-known that the 2d toric code \cite{Kitaev1997}, can be described by an effective BF $\bbZ_2$-gauge theory \cite{Delcamp:2018kqc}, and therefore hosts a Drinfel'd double symmetry $D(\bbZ_2)$. This is a physical manifestation of the fact that the Drinfel'd centre of $\operatorname{Rep}(\bbZ_2)$ coincides with the representation category of $D(\bbZ_2)$,
\begin{equation}
    Z_1(\operatorname{Vect}[\bbZ_2])\simeq \operatorname{Rep}(D(\bbZ_2)),\nonumber
\end{equation}
the former of which describes the 2D toric code \cite{KitaevKong_2012}. In this 3-dimensional scenario, both sides of the above correspondence are phrased in terms of the linear symmetric fusion category $\operatorname{Vect}$ of $k$-vector spaces; in particular, the representation categories are understood as functor categories into $\operatorname{Vect}$, and (bi)algebras $H$ over $k$ correspond to (bi)algebra objects in $\operatorname{Vect}$.

We wish to lift the above correspondence between gauge theory and symmetry-protected topological (SPT) phases to 4-dimensions. This calls for a {\it categorification} of the underlying algebraic structure (see eg. \cite{Wen:2019,Kong:2017hcw,Johnson-Freyd:2020usu}) --- in particular one needs a higher-categorical notion of $\operatorname{Vect}$. The issue here is that there are different such notions of "2-vector spaces", two of which play a central role in this paper.
\begin{enumerate}
    \item Kapranov-Voevodsky (KV) 2-vector spaces $\operatorname{2Vect}^{KV}$ \cite{Kapranov:1994}, which is a linear 2-category consisting of $k$-linear finite semisimple 1-categories (such as $\operatorname{Vect}$), linear functors as 1-morphisms and natural transformations between these functors as 2-morphisms, and 
    \item Baez-Crans (BC) 2-vector spaces $\operatorname{2Vect}^{BC}$ \cite{Baez:2003fs}, which is a linear 2-category consisting of $k$-linear 2-term cochain complexes of vector spaces, cochain maps as 1-morphisms and cochain homotopies between such chain maps as 2-morphisms.
\end{enumerate}
The theory surrounding the KV 2-vector space has received much more attention in the literature (see eg. \cite{Johnson-Freyd:2020usu,Johnson_Freyd_2023,Delcamp:2023kew,Douglas:2018,decoppet2022morita,decoppet2023drinfeld}), and have seen successful applications to describe gapped topological phases in 4d \cite{Hamma:2004ud,Else:2017yqj,Wen:2019,Kong:2020wmn}.

On the other hand, differential graded algebraic structures --- such as $L_\infty$-algebras \cite{Chen:2012gz,Kim:2019owc,Sati:2009ic,costello_gwilliam_2016} and crossed-complexes of groups \cite{Brown,Ang:2018rls,Kapustin2017,Baez:2004,Yekuteli:2015} --- have also appeared in the literature as a way to model higher-dimensional physics, topology and geometry. These notions enjoy desirable properties, such as the fact that 2-term $L_\infty$-algebras, aka "Lie 2-algebras", serve as infinitesimal approximations of Lie 2-groups. The sort of gauge principles that are built out of the corresponding principal (2-truncated) $\infty$-bundle \cite{Wockel2008Principal2A,Nikolaus_2014} then forms the basis of 2-gauge theories studied in the literature.

\medskip

The main result of this paper is to unite the above two ways of studying higher-dimensional physics, in the simplest case of the $\bbZ_2$-SPT. A main obstacle towards this goal is that the BC 2-category $\operatorname{2Vect}^{BC}$ is {\it strict} \cite{Baez:2003fs}, meaning that it has no non-trivial coherence 1- and 2-morphisms. This is a problem, as the corresponding 2-representation theory (namely linear 2-functors into $\operatorname{2Vect}^{BC}$) is unable to carry the sort of coherence data that higher monoidal categories have, which are key components in (non-trivial) topological field theories \cite{Douglas:2018}.

A workaround to this was found by the author in \cite{Chen:2023tjf}, based on the theory of 2-term $A_\infty$-algebras (in chain complexes), aka "2-algebras", with $n$-ary products up to $n=3$ \cite{Stasheff:1963}. In particular, a general theory of {\it 2-bialgebras} was developed, and its "weak" 2-representation theory $\operatorname{2Rep}_\text{wk}(\cA)$ was studied. The main result we shall leverage is the following.
\begin{theorem}\label{basicthm}
The 2-representation 2-category $\operatorname{2Rep}_\text{wk}(\cA)$ of a weak quasitriangular 2-bialgebra $(\cA,\cT,\Delta,\cR)$ forms a braided monoidal 2-category.
\end{theorem}
\noindent We refer the reader to \cite{GURSKI20114225} for the definition of a braided monoidal 2-category. Here, by a \textbf{2-bialgebra} we mean a 2-term $A_\infty$-algebra equipped with a compatible {shifted} 2-coalgebra structure; this will be defined properly in the main text. 

We note briefly that the non-trivial homotopy data attached to the $A_\infty$-algebra structure give rise to the higher coherence data seen in braided monoidal 2-categories. Moreover, we have also shown that these 2-bialgebras serve as a \textit{quantization} of Lie 2-bialgebras known in the literature \cite{Bai_2013,Chen:2012gz,chen:2022}. In summary, in terms of the following categorical ladder \cite{Crane:1994ty,Mackaay:ek},
\begin{figure}[h]
    \centering
    \includegraphics[width=0.7\columnwidth]{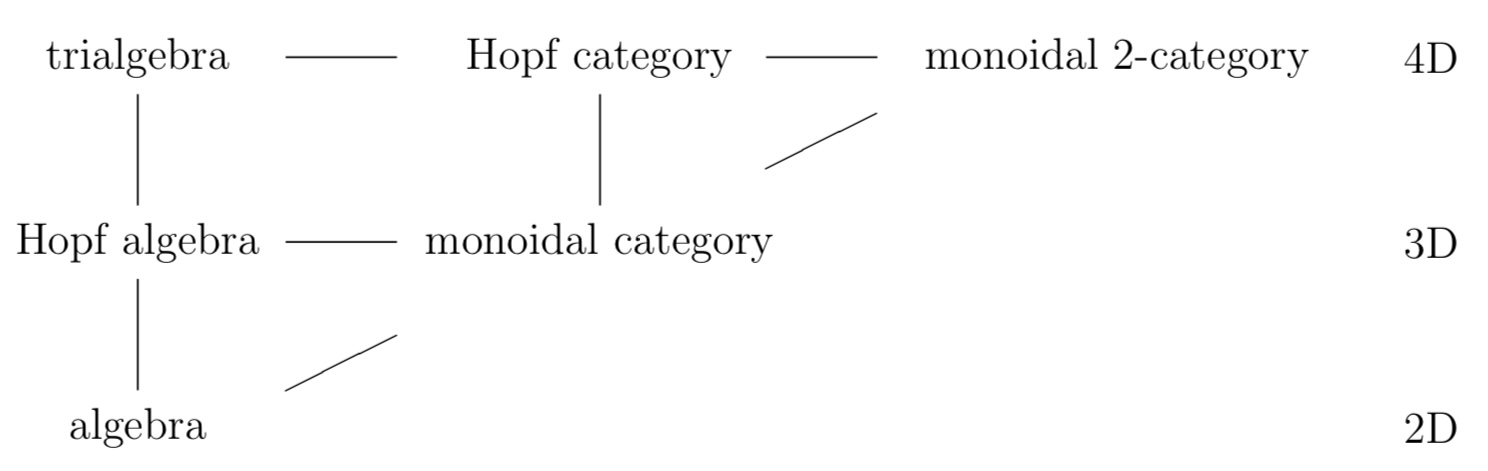}
    % \caption{Caption}
    % \label{fig:my_label}
\end{figure}

\noindent our strategy will be to construct concrete examples of {2-bialgebras} which model directly the structures of a Hopf category \cite{neuchl1997representation}, as well as certain 4d TFTs associated to them. For a study of trialgebras, see for instance \cite{Pfeiffer2007,grosse2001suggestion}.

\medskip

The goal for this paper will be accomplished in two steps. First, we specialize to a particular {\it 2-group} $B\bbZ_2$ obtained by delooping the finite group $\bbZ_2$, and form its Drinfel'd double 2-bialgebra $D(B\bbZ_2)$. We then construct topological field theories based on $D(B\bbZ_2)$ which we collectively call the {\it 4d Kitaev model}; we also briefly mention its relation to previous works \cite{Zhu:2019,chen:2022,Kapustin2017,Wen:2019}. In analogy with the 3-dimensional case, the excitations of the   4d Kitaev model are labeled by the 2-representations of $D(B\bbZ_2)$ in the sense of \cite{Chen:2023tjf}.

The second step is then to explicitly compute the monoidal and braiding structures of the 2-representation 2-category $\operatorname{2Rep}_\text{wk}(D(B\bbZ_2))$ by utilizing  the structures of the underlying Drinfel'd double 2-bialgebra \cite{Chen:2023tjf}. Provided $D(B\bbZ_2)$ is equipped with certain twists/group cocycles on $\bbZ_2$, we demonstrate that $\operatorname{2Rep}_\text{wk}(D(B\bbZ_2))$ recovers in fact two of the braided monoidal 2-categories $\mathscr{R},\mathscr{S}$ studied in \cite{Johnson-Freyd:2020,Johnson_Freyd_2023}, which model respectively the   4d toric code \cite{Hamma:2004ud,Kong:2020wmn,Wan:2019,Else:2017yqj} and its spin counterpart. 

More precisely, we first note that contributions to the Dijkgraaf-Witten 4-cocycle $\omega\in H^4(D(B\bbZ_2),k^\times)$ on the 2-group underlying $D(B\bbZ_2)$ has been computed in \cite{Kapustin2017}. Given a tuple of group 2-cocycles $\bar e \in H^2(\bbZ_2,\bbZ_2),\bar c\in H^2(\bbZ_2,k^\times)$, we can construct the following 2-group 4-cocycles 
\begin{equation}
    \omega_b = \bar e,\qquad \omega_f = \bar c[1] + \bar e; \nonumber
\end{equation}
we will detail this construction in the main text. These 2-cocycles $\bar e,\bar c$ contribute to twist the 2-algebra structure on $D(B\bbZ_2)$ in such a way that we have the following.
\begin{theorem}\label{mainthm}
We have the following \textrm{braided} equivalences 
\begin{eqnarray}
    \operatorname{2Rep}_\text{wk}(D ^{\omega_b}(B\bbZ_2)) &\simeq& \mathscr{R}\simeq Z_1(\operatorname{2Vect}^{KV}[\bbZ_2]), \nonumber\\
    \operatorname{2Rep}_\text{wk}(D^{\omega_f}(B\bbZ_2)) &\simeq&\mathscr{S}\simeq Z_1(\Sigma\operatorname{sVect}),\nonumber
\end{eqnarray}
where $Z_1$ is the Drinfel'd centre and $\Sigma$ is the condensation completion functor defined in \cite{Gaiotto:2019xmp} (note $\operatorname{2Vect}^{KV} = \Sigma\operatorname{Vect}$).
\end{theorem}
\noindent Moreover, we show that the 4-cocycles $\omega_b,\omega_f$, when pulled-back by a classifying map $X\rightarrow BD(B\bbZ_2)$ on a 4-manifold $X$, give  precisely the Dijkgraaf-Witten Lagrangian of the underlying topological NLSM \cite{Zhu:2019}. The meaning of "twisting" shall be made clear in the main text.

This can be understood as an example of {\it categorified} Tannakian reconstruction\footnote{This sentence should be stated more precisely by first picking appropriate forgetful 2-functors into $\operatorname{2Vect}^{KV}$ \cite{Johnson_Freyd_2023}, analogous to the 1-category case \cite{Woronowicz1988}.} \cite{Pfeiffer2007} in the special case of Drinfel'd centre 2-categories. This also suggests that one may model the tube algebra in 4-dimensions \cite{Bullivant:2021} as the "fusion 2-ring" of the Drinfel'd double 2-bialgebras $D(BG)$, categorifying an analogous statement in  3d \cite{Delcamp:2016yix}.

\medskip

It is important to note that here the 2-representations of $D(B\bbZ_2)$ we consider are defined in the setting of "weak Baez-Crans 2-vector spaces" as examined in \cite{Chen:2023tjf}. They form a linear 2-category $\operatorname{2Vect}^{hBC}$ whose algebra objects are precisely 2-term $A_\infty$-algebras (cf. the macrocosm principle \cite{Baez:1995xq}). The linear 2-category $\operatorname{2Vect}^{BC}$ of ordinary Baez-Crans 2-vector spaces (ie. 2-term vector space chain complexes) \cite{Baez:2003fs,Baez:2010ya}, on the other hand, have merely crossed-modules of associative algebras \cite{Wagemann+2021} as algebra objects. The 2-representation theory based on $\operatorname{2Vect}^{hBC}$ is therefore capable of carrying non-trivial homotopy data that one cannot see in $\operatorname{2Vect}^{BC}$.

On the other hand, the 2-representation theory of 2-groups commonly studied in the literature \cite{Douglas:2018,Baez:2008hz,Delcamp:2021szr,Delcamp:2023kew,Johnson_Freyd_2023,Kong:2020wmn} is based on the KV 2-category $\operatorname{2Vect}^{KV}$. These carry higher coherence data, such as the associators and the pentagonators, and they are in fact very closely related to the homotopy data carried by the 2-representation theory based on $\operatorname{2Vect}^{hBC}$ \cite{Chen:2023tjf} (see {\it Remark \ref{wk2algrepthy}}). The exact relationship between $\operatorname{2Vect}^{hBC}$ and $\operatorname{2Vect}^{KV}$, as well as module 2-categories over them, is currently under investigation by the author. 

\medskip

The gist of this paper is therefore the following: we construct the (weak) Drinfel'd double 2-bialgebra $D(B\bbZ_2)$ out of the finite group $\bbZ_2$, and compute all of the structures in the 2-representation 2-category $$\operatorname{2Rep}_\text{wk}(\cA) = [\cA,\operatorname{2Vect}^{hBC}],\qquad \cA = D(B\bbZ_2)$$ based on the theory laid out in \cite{Chen:2023tjf}. 

\subsubsection*{Outline}
The paper is organized as follows. In Section \ref{prelim}, we shall recall some definitions that will be used throughout this paper. We very briefly review the notion of 2-groups, 2-algebras and 2-gauge theory based on existing literature \cite{Chen:2012gz,Wagemann+2021,Baez:2003fs,Baez:2004,chen:2022,Kapustin2017,Zhu:2019}, then move on to recall the construction and key duality structures of a 2-quantum double as given in a previous work \cite{Chen:2023tjf}. We focus in particular on skeletal models of Drinfel'd double 2-bialgebras based on finite Abelian groups --- as well as the properties of its 2-$R$-matrix --- which plays a central role in this paper.

In Section \ref{4dkit}, we define the "monster 2-BF theory" \cite{chen:2022} associated to the Drinfel'd double 2-bialgebra $D(B\bbZ_2)$, and construct its partition function on a 4-manifold. We refer to the resulting theory as the {\bf 4d Kitaev model}. Here, we also introduce additional terms in the topological action in accordance with previous work on 2-group Dijkgraaf-Witten theory \cite{Kapustin2017,Zhu:2019}, and provide an interpretation for them as {\it twists} on the underlying Drinfel'd double 2-bialgebra symmetry.

With the main tools under our belt, we  study the excitations of the   4d Kitaev model (without twists) in Section \ref{toriccharges}. This is done by explicitly computing all of the braided and monoidal/fusion structures in the 2-representation 2-category $\operatorname{2Rep}_\text{wk}(D(B\bbZ_2))$. We show that the   4d Kitaev model without any twists {\it cannot} describe any non-degenerate gapped topological phase, as the symmetric M{\" u}ger centre $Z_2(\operatorname{2Rep}_\text{wk}(D(B\bbZ_2))\not\simeq \operatorname{2Vect}^{KV}$ is in fact not trivial.

This problem is amended in Section \ref{ferkitcharges} by incorporating the twists, whence the fusion rules and braided structures are revisited. We find that, by interpreting these twists naturally as properties of the 2-representations of $D(B\bbZ_2)$, we are able to finally prove {\bf Theorem \ref{mainthm}} directly. We provide explicit braided equivalences between the 2-categories $\mathscr{R},\mathscr{S}$ studied in \cite{Johnson-Freyd:2020,Johnson_Freyd_2023} and the 2-representation 2-category of twisted versions of the Drinfel'd double 2-bialgebra $D(B\bbZ_2)$.

\smallskip

Finally, in Appendix \ref{Zp}, we apply our computational techniques to compute the 2-category $\operatorname{2Rep}_\text{wk}(D^\omega(B\bbZ_p))$ describing the $\bbZ_p$-analogue of the 4d toric code, where $p>2$ is an odd prime. We will in particular examine the fusion and braiding properties of this 2-category, and demonstrate how the statement of remote detectability arises from the twist cocycle $\omega\in H^4(D(B\bbZ_p),k^\times)$.

\subsubsection*{Acknowledgements}
The author would like to thank Ruochen Ma, Matthew Yu, Theo Johnson-Freyd, Liang Kong, Mathew Bullimore and Clement Delcamp for valuable discussions throughout the completion of this work, and also pointing me towards relevant results in the literature. I would also like to thank the journal editor and referee for the comments and suggestions, which helped to drastically improve the quality and novelty of the article.

\section{Preliminaries}\label{prelim}
Let us begin with a brief overview of the notion of 2-groups, 2-bialgebras and 2-gauge theory \cite{Bai_2013,Chen:2012gz,Baez:2004in,Martins:2006hx,chen:2022,Chen:2023tjf}. In this article, we shall mainly consider the presentation of a 2-group $G$ as a group crossed-module, but there are many equivalent definitions that we shall mention.

\subsection{A brief overview on 2-groups and 2-gauge theory}
\begin{definition}
A \textbf{2-group} $\mathbb{G}=(G_{-1},G_0,\rhd)$ is a 2-term crossed complex, ie. a crossed-module, of groups, which consist of two groups $G_{-1},G_0$, a group homomorphism $t:G_{-1}\rightarrow G_0$ and an action $\rhd$ of $G_0$ on $G_{-1}$, such that the following conditions
\begin{equation}
    t(x\rhd y) = xt(y)x^{-1},\qquad (ty)\rhd y'=yy'y^{-1}\label{pfeif2}
\end{equation}
are satisfied for each $x\in G_0$ and $y,y'\in G_{-1}$. The first and second conditions are known respectively as the {\it equivariance} and the {\it Peiffer identity}.
\end{definition}
\noindent Readers more familiar with category theory may understand 2-groups equivalently as a connected 2-groupoid $[G_{-1}\rtimes G_0,G_0,\ast]$ \cite{Baez:2008hz,Douglas:2018}, or equivalently its loop 1-groupoid $G_{-1}\rtimes G_0 \rightrightarrows G_0$ \cite{Chen:2012gz,Porst2008Strict2A}.

\subsubsection{Classification of 2-groups}
As we shall mainly be interested in topological field theories, we only need to work with equivalence classes of 2-groups. More precisely, we consider 2-group $\mathbb{G}=G_{-1}\rightarrow G$ only up to chain homotopies.

\begin{definition}
A cochain map between 2-term group complexes, or simply a {\bf 2-group homomorphism}, is a graded map $\phi=(\phi_{-1},\phi_0):\mathbb{G}\rightarrow \mathbb{G}'$ such that
\begin{enumerate}
    \item $\phi_0:G_0\rightarrow G_0'$ and $\phi_{-1}:G_{-1}\rightarrow G_{-1}'$ are group homomorphisms,
    \item $\phi_{-1}(x\rhd y) = (\phi_0x)\rhd'(\phi_{-1}y)$ for each $x\in G_0,y\in G_{-1}$, 
    \item $\phi_0t=t'\phi_{-1}$.
\end{enumerate}
\end{definition}
\noindent We say that two 2-groups $\mathbb{G},\mathbb{G}'$ are equivalent, or quasi-isomorphic, if there exists a weakly invertible 2-group homomorphism between them. The fundamental classification result \cite{Zhu:2019,Kapustin2017,Wagemann+2021,Ang:2018rls,Baez:2004} is that
\begin{theorem}
{\bf (Gerstenhaber, attr. Mac-Lane).} 2-groups are classified up to quasi-isomorphism by a degree-3 group cohomology class $\tau\in H^3(N,M)$, called the {\bf Postnikov class} of $\mathbb{G}$, where $N=\operatorname{coker}t,M=\operatorname{ker}t$.
\end{theorem}
\noindent Note $M$ is Abelian due to the Peiffer identities. The tuple $(N,M,\tau)$ is known as the {\it Ho{\'a}ng data} \cite{Ang:2018rls,Nguyen2014CROSSEDMA}, which was proven by Xu{\^ a}n S{\' i}nh Ho{\'a}ng, a Vietnamese female mathematician supervised by Alexander Grothendieck, to classify "$\operatorname{Gr}$-categories".

Note that the part $(N,M)$ of the Ho{\'a}ng data of $\mathbb{G}$ defines a skeletal 2-group, where the $t$-map is trivial. We call this 2-group $\pi \mathbb{G}=M\xrightarrow{0}N$ the {\bf skeleton} of $\mathbb{G}$. Clearly, skeletal 2-groups $\mathbb{G}$ are their own skeletons, and are hence classified by the Ho{\'a}ng data $(G_0,G_{-1},\tau)$. 

% \begin{remark}
% An analogous classification statement also holds for Lie algebra crossed-modules, in which $\g = \g_{-1}\rightarrow \g_0$ is classified by the data $(\operatorname{coker}t=\mathfrak{n},\operatorname{ker}t=V,\kappa)$, where $\kappa\in H^3(\mathfrak{n},V)$ is a Lie algebra cohomology element \cite{Wag:2006}. A fun fact is that Lie 2-algebras are classified in the same way; the skeleton $\pi \g=\mathfrak{n}\xrightarrow{0}V$ forms a weak Lie 2-algebra for which the homotopy map $\mu_3$ is in fact a 3-cocycle representative of the cohomology class $\kappa$ \cite{Baez:2005sn,Kim:2019owc}.
% \end{remark}

\subsubsection{2-groups as categorical groups}\label{catgrp} 
In this skeletal case, we can in fact understand $\mathbb{G}=(G_0,G_{-1},\tau)$ as a {\bf categorical group} with categorical structure \cite{Cui_2017}. Let $G_{-1}$ denote an Abelian $G_0$-module, and let us consider a category $\cC=\mathbb{G}$ with objects $\operatorname{Obj}(\cC)=G_0$ and morphisms $\operatorname{Hom}_{\cC}(x,x') = G_{-1}$ when $x=x'\in G_0$, and trivial otherwise. Composition of morphisms is written additively,
\begin{equation*}
    (y,x)+(y',x) = (y+y',x),\qquad x\in G_0,\quad y,y'\in G_{-1}= \operatorname{End}_\cC(x), 
\end{equation*}
and we endow $\cC$ with the monoidal structure (the horizontal multiplication \cite{Baez:2003fs}) given by
\begin{equation}
    (y,x) (y,x') = (y+x\rhd y',xx').\nonumber
\end{equation}
The two composition laws $\otimes,+$ must satisfy the {\it interchange law} \cite{Porst2008Strict2A,Baez:2003fs,Baez:2012}
\begin{equation}
    ((y_1,x_1)+(y_1',x_1))((y_2,x_2)+(y_2',x_2)) = (y_1,x_1)(y_2,x_2)+(y_1',x_1)(y_2',x_2)\label{grpinterchange}
\end{equation}

Furthermore, we also define distinguished associator isomorphisms for the horizontal multiplication
\begin{equation}
    \tau(x,x',x''):(xx')x''\xrightarrow{\sim}x(x'x'')\in \operatorname{Hom}_\cC(xx'x'') =G_{-1},\nonumber
\end{equation}
for which the \textbf{pentagon relation} 
\begin{equation}
    \begin{tikzcd}
                                                        &                                                             & ((x_1x_2)x_3)x_4 \arrow[rrd, "{\tau(x_1x_2,x_3,x_4)}"] \arrow[lld, "{\tau(x_1,x_2,x_3)}"'] &                  &                                                        \\
(x_1(x_2x_3))x_4 \arrow[rdd, "{\tau(x_1,x_2x_3,x_4)}"'] &                                                             &                                                                                            &                  & (x_1x_2)(x_3x_4) \arrow[ldd, "{\tau(x_1,x_2,x_3x_4)}"] \\
                                                        &                                                             &                                                                                            &                  &                                                        \\
                                                        & x_1((x_2x_3)x_4) \arrow[rr, "{x_1\rhd \tau(x_2,x_3,x_4)}"'] &                                                                                            & x_1(x_2(x_3x_4)) &                                                       
\end{tikzcd}\label{pentagon}
\end{equation}
is satisfied. It is well-known that, by interpreting $\tau$ as a group 3-cochain $\tau:G_0^{3\times}\rightarrow G_{-1}$, the pentagon relation implies that it is in fact a 3-cocycle, which represents the Postnikov class $\tau\in H^3(G_0,G_{-1})$. This yields a categorical understanding of the kind of algebraic data a 2-group encodes; similar categorical realizations can also be given for 2-algebras \cite{Baez:2003fs,Wagemann+2021,Chen:2023tjf}.

If one is more geometrically-minded, one may construct the {\it classifying space} $B\mathbb{G}$ of the skeletal 2-group $\mathbb{G}$ through the Postnikov tower \cite{Zhu:2019,Kapustin2017}. This space is given as a fibration
\begin{equation}
    B^2G_{-1}\rightarrow BG\rightarrow BG_0 \rightarrow \ast, \nonumber
\end{equation}
in which $BG_0$ is the classifying space of $G_0$ and $B^2G_{-1}$ is the second delooping of $G_{-1}$, namely $\pi_2(B^2G_{-1}) = G_{-1}$ with all other homotopy groups trivial. The (representative of the) Postnikov class $\tau\in H^3(G_0,G_{-1})$ then prescribes the way in which the 2-cells of $B^2G_{-1}$ are glued onto the principal path fibration on $BG_0$ \cite{Mackaay:uo}.

\begin{remark}\label{gencatgrp}
More generally, the monoidal structure can be weakened in various ways by introducing distinguished {\it 2-morphisms} to $\mathbb{G}$, which are invertible elements $k^\times$ of the ground field $k$ \cite{Johnson-Freyd:2020}. For instance,  \eqref{grpinterchange} can be implemented by an interchanger 2-morphism $h(y_1,y_1';y_2,y_2')\in k^\times$, which is an isomorphism between 1-morphisms over $xx'\in G_0$. Together with the associator 1-morphism $\tau$, it defines an associator 2-morphism \cite{Wen:2019},
\begin{equation}
    \tilde\tau(y,y',y''): \tau(x,x',x'')+(yy')y''\rightarrow y(y'y'')+\tau(x,x',x''), \nonumber
\end{equation}
for the 1-morphisms $y,y',y''\in G_{-1}$ with respective sources $x,x',x''\in G_0$. The pentagon relation can also be weakened to a {\it pentagonator} $\nu(x_1,x_2,x_3,x_4)\in k^\times$, which is classified by degree-4 group cohomology class $H^4(G_0,k^\times)$ \cite{Zhu:2019}. 
\end{remark}

In the following, we shall interchangeably understand skeletal 2-groups both in terms of a crossed complex, or as a categorical group; in particular, we shall call the classifying Postnikov class $\tau$ also as the "associator". The same will be done for 2-algebras introduced later.

\subsubsection{2-gauge theory}\label{2gaugethry}
The notion of 2-groups is the natural setting to study a {\it 2-gauge theory}, ie. a categorification of the usual gauge theory. Recall that the fundamental structure underlying gauge theory is a principal bundle $P\rightarrow X$ with connection \cite{book-baez}. Locally, one may describe the connection in terms of a Lie algebra-valued connection 1-form $A$.%\in\Omega^1(X)\otimes \g$, with $\g$ the Lie algebra of the gauge group. 

Similarly, a principal 2-bundle $\mathcal{P}\rightarrow X$ on $X$ with connection is described by, locally, a pair of connection forms
\begin{equation}
    A\in \Omega^1(X)\otimes\g_0,\qquad \Sigma\in\Omega^2(X)\otimes\g_{-1},\nonumber 
\end{equation}
where $\g_0$ and $\g_{-1}$ are the Lie algebras of respectively $G_0$ and $G_{-1}$.
The 2-form $\Sigma$ allows us to define a {\it 2-holonomy} $\operatorname{2Hol}_\mathcal{P}(S) = S\exp\int_S \Sigma$ across a 2-dimensional surface $S\subset X$, where $S\exp$ is the surface-ordered exponential defined in \cite{Yekuteli:2015}.  See also \cite{Wockel2008Principal2A}.

The covariant geometric quantities are given by the fake-flatness and the 2-curvature \cite{Mikovic:2016xmo,chen:2022}
\begin{equation}
    \mathcal{F} = F-t\Sigma,\qquad \cA = d_A\Sigma - \kappa(A),\nonumber
\end{equation}
where $\kappa(A)$ is the evaluation of a 3-cocycle representative of the Postnikov class $\kappa\in H^3(\mathfrak{n},V)$ on the 1-connection $A$. For the details regarding the gauge structure, see \cite{Martins:2006hx,Mikovic:2011si}. The same construction can be applied to define flat 2-connections valued in the skeleton $\pi\g$ of $\g$.

\begin{remark}
Note that here we have used a Lie algebra crossed-module to model the 2-gauge theory. This is because the gauge structure associated to a weak Lie 2-algebra is problematic: the 2-gauge algebra does not close unless fake-flatness $\mathcal{F}=0$ is always satisfied \cite{Kim:2019owc}. There is no such problem in the crossed-module formulation, as long as we introduce the first descendant $\zeta(A,\lambda)$ \cite{Kapustin2017} of the Postnikov class $\kappa$ into the 1-gauge transformation \cite{Kim:2019owc}.
\end{remark}

\medskip

In the above, we have outlined the continuum description of a 2-gauge theory on $X$ using Lie 2-algebras, but a similar construction can be applied to yield a discrete 2-gauge theory \cite{Bullivant:2017sjz,Kapustin2017}. In particular, we define a discrete 2-connection $(A,\Sigma)$ as a pair of cochains
\begin{equation}
    A \in C^1(X,G_0),\quad  %\rightarrow
    \Sigma\in C^2(X,G_{-1}) \nonumber
\end{equation}
satisfying the flatness conditions
\begin{equation}
    (d_AA)_{(012)} = t\Sigma_{(012)},\qquad (d_A\Sigma)_{(0123)} = \kappa(A_\text{tree})\nonumber
\end{equation}
on a %discretized simpliciation 
simplicial discretization of $X$, where we denote an oriented $k$-simplex by an ordered tuple $(0\cdots k)$ and $d$ is now the differential on group cochains. Here, $A_\text{tree}$ denotes the evaluation of $A$ on the edges $(01),(02),(13)$ of $(0123)$ \cite{Kapustin2017}, as depicted in red in Fig. \ref{fig:tree}.

\begin{figure}[h]
    \centering
    \includegraphics[width=0.65\columnwidth]{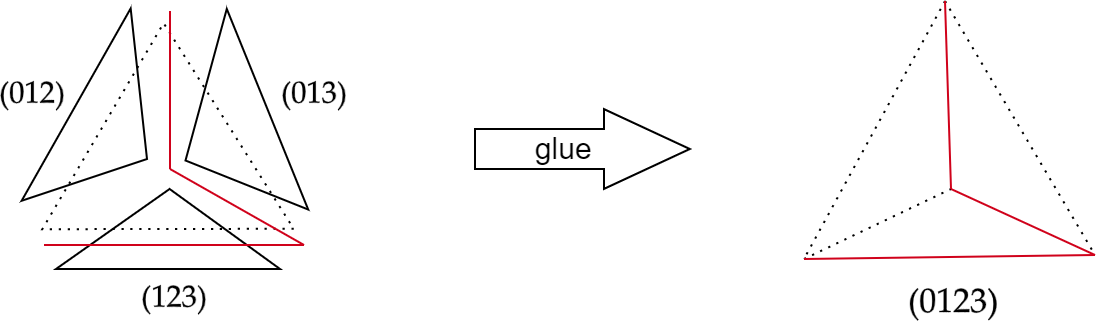}
    \label{fig:tree}
\end{figure}

\paragraph{Topological 2-gauge theories and higher-SPT phases.}
In general, a 2-group homomorphism $\mathbb{G}\rightarrow \mathbb{G}'$ induces a 2-bundle homomorphism $\mathcal{P}\rightarrow \mathcal{P}'$ which preserves the covariant curvature quantities $(\mathcal{F},\cA)$. As such, by the Gerstenhaber classification theorem above, 2-gauge bundles are classified up to 2-bundle maps by the Ho{\'a}ng data $(\pi \mathbb{G},\tau) = (N,M,\tau)$. 

The topological information encoded in a 2-gauge theory with structure 2-group $\mathbb{G}$,  are thus captured up to homotopy by a representative 2-bundle $\pi\mathcal{P}$, whose structure 2-group is the skeleton $\pi \mathbb{G}$ of $\mathbb{G}$. Indeed, it is known that flat 2-connections are in one-to-one correspondence with homotopy classes of classifying maps $f: X\rightarrow B(\pi \mathbb{G})$, where $B(\pi \mathbb{G})$ is the classifying space of the skeleton $\pi \mathbb{G}=M\rightarrow N$ (via the Postnikov tower construction) \cite{Kapustin2017,Zhu:2019}. 

These classifying maps define 2-group homomorphisms $\Pi_2X\rightarrow \pi \mathbb{G}$ from the homotopy 2-type $\Pi_2X$ of $X$ to $\pi \mathbb{G}$; such maps have been used to construct  3d topological defects \cite{Ang:2018rls}. The associated 2-gauge theory has also been proposed to describe higher-symmetry protected topological (SPT) phases of matter in (3+1)D \cite{Kapustin2017}. It is typically understood and accepted that the excitations in such 2-gauge theories {should} form a braided monoidal 2-category that characterizes 4-dimensional topological phases. In this paper, we explicitly demonstrate this for the example of $\bbZ_2$. 

\smallskip

In the following, we shall recall the Drinfel'd double 2-bialgebra construction of \cite{Chen:2023tjf}, but specialized to {\it skeletal} 2-group algebras $k\mathbb{G}$. We will also describe some basic monoidal properties of its 2-representation 2-category $\operatorname{2Rep}_\text{wk}(D(\mathbb{G}))$, where $\tau$ is the Postnikov class of $\mathbb{G}$. We shall use this structure in order to construct a (3+1)D topological field theory (TFT) and study its excitations as 2-representations of the Drinfel'd double $D(\mathbb{G})$ of $\mathbb{G}$. All 2-groups in the following shall be understood as skeletal.

\subsection{A lightning review on 2-bialgebras and Drinfel'd double 2-bialgebras}
Let $k$ denote the ground field in the following. Recall a bialgebra $(H,\cdot,\Delta)$ is a vector space $H$ equipped with an associative algebra structure and a coassociative coalgebra structure $\Delta:H\rightarrow H\otimes H$ that is mutually compatible, in the sense that
\begin{equation}
    \Delta(h h') = h_{(1)}  h'_{(1)}\otimes h_{(2)} h'_{(2)},\qquad \Delta(h)= h_{(1)}\otimes h_{(2)} \nonumber
\end{equation}
for all $h\in H$. Note we have used the conventional Sweedler notation for the coproduct $\Delta$.

\medskip 

Let $(\cA_0,\Delta_0'),(\cA_{-1},\Delta_{-1})$ denote a pair of associative bialgebras, which are equipped with the following linear maps $$\cdot:\cA_0\otimes\cA_{-1}\oplus\cA_{-1}\otimes\cA_0\rightarrow\cA_{-1},\qquad \Delta_0=\Delta_0^l\oplus\Delta_0^r:\cA_0\rightarrow \cA_{-1}\otimes\cA_0\oplus \cA_0\otimes\cA_{-1}.$$ We say that $\cA_{-1}$ is a $\cA_0$-bimodule iff 
\begin{equation}
        y(x\cdot y') = (y\cdot x)y',\qquad \forall y,y'\in\cA_{-1},~x\in\cA_0,\nonumber
\end{equation}
and $\cA_0$ is a $\cA_{-1}$-cobimodule iff
\begin{equation*}
    (1\otimes\Delta_0^l)\circ \Delta_0' = (\Delta_0^r\otimes 1)\circ \Delta_0'
\end{equation*}
% \begin{eqnarray}
%      (1\otimes \Delta_0')\circ \Delta_0^l &=& (\Delta_0^l\otimes 1)\circ\Delta_0',\nonumber\\
%      ,\nonumber\\
%      . \nonumber 
% \end{eqnarray}
The following is from \cite{Chen:2023tjf}.
\begin{definition}\label{assoc2bialg}
An \textbf{associative 2-bialgebra} $(\cA,\cdot,\Delta)$ is a bialgebra homomorphism $t:\cA_{-1}\rightarrow\cA_0$ such that
\begin{enumerate}
    \item $\cA_{-1}$ is a $\cA_0$-bimodule and $\cA_0$ is a $\cA_{-1}$-cobimodule,
    \item $t$ is two-sided \textbf{$\cA_0$-equivariant} and \textbf{$\cA_{-1}$-coequivariant},
    \begin{equation}
        \begin{cases}t(x\cdot y) = x t(y) \\ t(y\cdot x)=t(y)x\end{cases},\qquad D_t^+\circ\Delta_{-1}= \Delta_0\circ t\label{algpeif1}
    \end{equation}
    for all $y\in\cA_{-1},x\in\cA_0$, where $D_t^\pm = t\otimes 1 \pm 1\otimes t$, and
    \item the \textbf{Peiffer/coPeiffer identities} are satisfied,
    \begin{equation}
        t(y)\cdot y' = yy' = y\cdot t(y'),\qquad D_t^-\circ \Delta_0=0,\label{algpeif2}
    \end{equation}
    where $y,y'\in\cA_{-1}$, and 
    \item the {\bf 2-bialgbera axioms} are satisfied,
    \begin{eqnarray}
    \Delta_{-1}(x\cdot y)=  \bar x_{(1)}\cdot y_{(1)}\otimes  \bar x_{(2)}\cdot y_{(2)},&\qquad& \Delta_{-1}(y\cdot x)= y_{(1)}\cdot \bar x_{(1)}\otimes y_{(2)}\cdot \bar x_{(2)},\nonumber\\
    \Delta_0^l(xx')=x_{(1)}^lx_{(1)}'^l\otimes x_{(2)}^lx_{(2)}'^l,&\qquad& \Delta_0^r(xx') = x_{(1)}^rx_{(1)}'^r\otimes x_{(2)}^rx_{(2)}'^r,\label{2algcoprod}
    \end{eqnarray}
    where $x,x'\in \cA_0,y\in \cA_{-1}$ and we have used the conventional Sweedler notation
    \begin{equation}
        \Delta_0(x) = x_{(1)}^l\otimes x_{(2)}^l + x_{(1)}^r\otimes x_{(2)}^r,\qquad \Delta_0'(x) = \bar x_{(1)}\otimes \bar x_{(2)}.\nonumber
    \end{equation}
\end{enumerate}
\end{definition}
\noindent In other words, an associative 2-bialgebra $(\cA,\cdot\Delta)$ is simultaneously an associative 2-algebra $(\cA,\cdot)$ --- or equivalently a crossed-module of associative algebras \cite{Wagemann+2021} --- and a coassociative 2-coalgebra $(\cA,\Delta)$, which are mutually compatible according to the above 2-bialgebra axioms. 

\begin{definition}
A 2-bialgebra $\cA$ is {\it unital} if there exists a unit $\eta=(\eta_{-1},\eta_0):k\rightarrow \cA$ and a counit $\epsilon = (\epsilon_{-1},\epsilon_0):\cA\rightarrow k$ such that
\begin{gather}
  \eta_{-1} y = y \eta_{-1} = y,  \quad \eta_0  x = x \eta_0 = x,\quad
   \eta_0\cdot y = y\cdot \eta_0 = y, \nonumber\\
    \operatorname{id}=(\operatorname{id}\otimes \epsilon_{-1})\circ\Delta_{-1},\quad \operatorname{id} = (\epsilon_{-1}\otimes \operatorname{id})\circ\Delta_{-1},\nonumber\\
    % \operatorname{id}&=&(\operatorname{id}\otimes \epsilon_0)\circ\Delta_0^l = (\epsilon_{-1}\otimes\operatorname{id})\circ\Delta_0^l,\nonumber\\
    % \operatorname{id}&=&(\operatorname{id}\otimes\epsilon_{-1})\circ\Delta_0^r = (\epsilon_0\otimes\operatorname{id})\circ\Delta_0^r.\nonumber
    \operatorname{id}=(\epsilon_{-1}\otimes \operatorname{id})\circ\Delta_0^l, \quad \operatorname{id} =(\operatorname{id}\otimes \epsilon_{-1})\circ\Delta_0^r.\nonumber
\end{gather}
Moreover, $t$ should respect the units and counits such that $\eta_{-1} = t\eta_0$ and $\epsilon_0 = \epsilon_{-1}t$. 
\end{definition}

% as well as the fact that the (co)bimodule structure preserves the (co)unit:
% \begin{gather}
%     x\cdot\eta_{-1} = \eta_{-1}\cdot x = \epsilon_0(x)\eta_{-1}\nonumber \\
%     \eta_{-1}\epsilon_0=(\operatorname{id}\otimes \epsilon_0)\circ\Delta_0^l=(\epsilon_0\otimes \operatorname{id})\circ\Delta_0^r.\nonumber
% \end{gather}
We note that for non-trivial $t\neq 0$, the component $\Delta_0'$ of the coproduct is constrained to satisfy $$\Delta_0' = (t\otimes 1)\circ\Delta_0^l= (1\otimes t)\circ\Delta_0^r.$$ Otherwise (ie. for {\it skeletal} 2-algebras), all of the components of the coproduct are independent.

\paragraph{Classification of 2-algebras.} Similar to the case of 2-groups, an associative 2-algebra homomorphism $f=(f_{-1},f_0):\cA\rightarrow \cA'$ is a cochain map between 2-algebras, or equivalently a graded pair of algebra homomorphisms that respect the underlying bimodule structure, such that
\begin{enumerate}
    \item $f_0:\cA_0\rightarrow \cA_0'$ and $f_{_1}:\cA_{-1}\rightarrow \cA_{-1}'$ are algebra homomorphisms,
    \item $f_{-1}(x\cdot y) = (f_0x)\cdot'(f_{-1}y)$ and $f_{-1}(y\cdot x) = (f_{-1}y)\cdot' (f_0x)$ for each $x\in\cA_0,y\in\cA_{-1}$, and
    \item $f_0t=t'f_{-1}$.
\end{enumerate}
We say that two 2-algebras are equivalent, $\cA\sim\cA'$, if there exists a weakly invertible (ie. invertible up to cochain homotopy) 2-algebra homomorphism between them. The following is known in \cite{Wagemann+2021}.
\begin{theorem}\label{2algclass}
{\bf (Gerstenhaber, attr. Wagemann).} Associative 2-algebras are are classified up to equivalence by a third {\bf Hochschild cohomology} class $\cT\in HH^3(\cN,V)$, where $\cN = \operatorname{coker}t$ and $\cM=\operatorname{ker}t$. 
\end{theorem}
\noindent The Peiffer identity implies that $\cM\subset Z(\cA_{-1})$ is in the {\it nucleus} of $\cA_{-1}$.

\subsubsection{2-group bialgebras}\label{2grpbialg}
The main example of 2-bialgebras of interest in this paper comes from finite 2-groups $\mathbb{G}$, in which we take the group algebra functor on each graded component $kG_{-1},kG_0$ \cite{Pfeiffer2007}. 

One may try to extend all 2-group structures $k$-linearly, but the main difficulty is that $G_{-1}$ acts on $G_0$ by group automorphism, not algebra automorphism, as required by Peiffer identity. As a result, the $kG_0$-bimodule structure on $kG_{-1}$ induced this way, namely
\begin{equation}
    \begin{cases}
    x\cdot y = x\rhd y\\ y\cdot x = x^{-1}\rhd y
    \end{cases},\qquad x\in kG_0,~y\in kG_{-1},\nonumber
\end{equation}
does not give a bona fide 2-algebra.

There are several ways to amend this. One is to quotient $kG_{-1}$ out by an ideal generated by $yy' - yy'y^{-1}$; this is the method given in \cite{Wagemann+2021}, and is much too trivial for our purposes. We are looking for an association $\mathbb{G}\rightarrow \cA$ which descends to a well-defined map on the equivalence classes.

% An alternative way is to suppose the action $x\rhd\in\operatorname{Inn}G_{-1}$ is inner for all $x\in G_0$, and endow a $kG_0$-bimodule structure $\cdot$ on $kG_{-1}$ inherited thus:
% \begin{equation}
%     x\rhd y = x\cdot y\cdot x^{-1},\nonumber
% \end{equation}
% whence the 2-algebra axioms follow from the 2-group axioms. In either case, there is a correspondence $\tau\mapsto \cT$ between the classifying group 3-cocycle $\tau\in H^3(N,V)$ of the 2-group $G$ and the Hochschild 3-cocycle $\cT\in H^3(\mathcal{N},V)$ of the 2-group algebra $kG$.

We can describe this construction explicitly in the skeletal case. Let $N$ be a finite group and let $M$ denote an Abelian $N$-module. Given the skeletal 2-group $\mathbb{G}=M\xrightarrow{0}N$, we extend all structures $k$-linearly {\it except} for the $t$-map. We take, instead, the skeletal 2-algebra $\cA=k\mathbb{G}$ with $t=0:kM\rightarrow kN$. If the Postnikov class $\tau\in H^3(N,M)$ of $\mathbb{G}$ is non-trivial, then theassociated Hochschild class $\cA=k\mathbb{G}$ is also non-trivial.
\begin{remark}\label{2alg-2grp}
In certain cases, the above construction is in fact a bijection. To see this, we can leverage a natural (ring) isomorphism of the Hochschild cohomology $HH^\ast(kN,kN)$ with the group cohomology $H^\ast(N,kN)$ \cite{Siegel:1999}. If $N$ is Abelian, then
\begin{equation}
    HH^\ast(kN,kN) \cong kN\otimes_k H^\ast(k N,k).\nonumber
\end{equation}
Assuming that $N,M$ are isomorphic, then in degree-3 we have an isomorphism of Abelian groups
\begin{equation}
    HH^3(kN,kM) \cong kM\otimes_k H^3(k N,k)\cong H^3(N,M)\otimes k.\label{class}
\end{equation}
Given the element $\tau\in H^3(N,M)$ classifying the 2-group $G = M\xrightarrow{0} N$, we take the corresponding element $\cT\in HH^3(kN,kM)$ for the associated 2-algebra $\cA=kG$.
\end{remark}

\paragraph{Skeletal Drinfel'd double 2-bialgebra of a finite cyclic Abelian group.} We take now $N$ to be a finite cyclic Abelian group, and let $M$ denote an Abelian $N$-module. Suppose further that $N\cong M$, and let $p:kN\xrightarrow{\sim}kM$ denote the induced linear isomorphism. The simplest graded coproduct $\Delta$ we can endow on the 2-algebra $k\mathbb{G} = kM\xrightarrow{0}kN$ is the grouplike coproduct:
\begin{eqnarray}
    \Delta_{-1}(y) = y\otimes y,\qquad \Delta_0'(x) = x\otimes x,\nonumber\\
    \Delta_0(x) =  p(x)\otimes x + x \otimes p(x),\label{grouplike}
\end{eqnarray}
where $x\in kN,y\in kM$. By definition, $\Delta$ is coassociative and admits the usual antipode $S_0^0(x) = x^{-1}, S_0^1(y) = y^{-1}$,
% such tha
% \begin{eqnarray}
%     \cdot\circ(\operatorname{id}\otimes  S_0^1)\circ \Delta_{-1} &=& \cdot\circ( S_0^1\otimes \operatorname{id})\circ\Delta_{-1}=\eta_{-1}\circ\epsilon_{-1},\nonumber\\
%     \cdot\circ(\operatorname{id}\otimes  S_0^0)\circ \Delta_0^l &=& \cdot\circ(S_0^1\otimes \operatorname{id})\circ\Delta_0^l=\eta_{-1}\circ\epsilon_0,\nonumber\\
%     \cdot\circ(\operatorname{id}\otimes  S_0^1)\circ \Delta_0^r &=& \cdot\circ(S_0^0\otimes \operatorname{id})\circ\Delta_0^r=\eta_{-1}\circ\epsilon_0,\label{antpode}
% \end{eqnarray}
% as well as 
together with the unit/counit
\begin{equation}
    \begin{cases}\eta_0=1\in N \\ \eta_{-1}=1\in M\end{cases},\qquad \begin{cases}\epsilon_0(x)=\delta_{x\eta_0} \in k \\ \epsilon_{-1}(y)=\delta_{y\eta_{-1}}\in k\end{cases}.\nonumber
\end{equation}
Moreover, this coproduct can be very easily shown to satisfy the 2-bialgebra axioms,
\begin{eqnarray}
    \Delta_{-1}(x\cdot y) = x\cdot y\otimes x\cdot y,\qquad \Delta_{-1}(y\cdot x) = y\cdot x\otimes y\cdot x,\nonumber\\
    \Delta_0(xx') = p(x)p(x')\otimes  xx' +  xx'\otimes  p(x)p(x'),\nonumber
\end{eqnarray}
where we have used the fact that $p$ is a group homomorphism $p(xx') = p(x)p(x')$. 

This defines $(k\mathbb{G},\cdot,\Delta,S)$ as a unital 2-Hopf algebra (see Appendix A of \cite{Chen:2023tjf}), but we require more structure: we require it to define a {\it Drinfel'd double}. The Drinfel'd double, in this context of 2-term $A_\infty$-algebras, is briefly a graded bicrossed product $\cA\bar{\bowtie}\cH$ of two mutually paired 2-bialgebras $\cA,\cH$ \cite{Chen:2023tjf}, and is therefore in particular naturally \textbf{self-dual} in the following sense.

We define the dual of the complex $\cA=\cA_{-1}\xrightarrow{t}\cA_0$ as the shifted cochain complex $\cA^* = \cA_0^*\xrightarrow{t^*} \cA^*_{-1}$, which has degrees concentrated in $-1,0$. The natural evaluation pairing then defines a symmetric bilinear form of degree-1 $$\langle-,-\rangle: \cA_0^*\otimes \cA_0 \oplus \cA_{-1}^*\otimes\cA_{-1}\rightarrow k,\qquad \langle (g,f),(y,x)\rangle = g(x) + f(y),$$ where $x\in\cA_0,y\in\cA_{-1}$ and $g\in \cA_0^*,f\in\cA_{-1}^*$. Now if $(\cA,\cdot)$ is a 2-algebra, then we may use this bilinear form to define several linear maps on $\cA^*$:
\begin{eqnarray}
    \Delta^*_{-1}: \cA_0^*\rightarrow \cA_0^*\otimes\cA_0^*,&\qquad & \langle \Delta_{-1}^*(g),x_1\otimes x_2\rangle = \langle g,x_1x_2\rangle,\nonumber\\
    (\Delta^*_0)^l:\cA^*_{-1}\rightarrow \cA_0^*\otimes \cA_{-1}^*,&\qquad& \langle (\Delta_0^*)^l(f),x\otimes y\rangle = \langle f,x\cdot y\rangle,\nonumber\\
    (\Delta_0^*)^r:\cA^*_{-1}\rightarrow \cA_{-1}^*\otimes \cA_0^*,&\qquad& \langle (\Delta_0^*)^r(f),y\otimes x\rangle = \langle f,y\cdot x\rangle.\label{dualcoprod}
\end{eqnarray}
The self-duality result for 2-bialgebras is the following \cite{Chen:2023tjf}.
\begin{proposition}\label{dual2bialg}
$(\cA,\cdot,\Delta)$ is a (unital) 2-bialgebra iff its dual $(\cA^*,\cdot^*,\Delta^*)$ is also a (unital) 2-bialgebra. Moreover, Drinfel'd double 2-bialgebras $(\cD,\cdot,\Delta)$ are self-dual --- namely it is isomorphic as a 2-bialgebra to $(\cD^*,\cdot^*,\Delta^*)$.
\end{proposition}
\noindent This statement is intimated related to the recent result \cite{decoppet2023drinfeld} that Drinfel'd centre 2-categories are Morita self-dual, as we expect the 2-representation 2-category associated to Drinfel'd double 2-bialgebras are Drinfel'd centres.

To endow our particular skeletal 2-bialgebra $(k\mathbb{G},\cdot,\Delta)$ with the structure of a Drinfel'd double, we take $M= \widehat{N}$ --- the Pontrjagyn dual of $N$ --- such that
\begin{equation}
    k\mathbb{G} \equiv D(BM) = kBM\bar{\bowtie} kN,\qquad \begin{cases} BM = M\xrightarrow{0}\ast \\ N_\ast = \ast\xrightarrow{0}N\end{cases}.\nonumber
\end{equation}
The map $p: x\mapsto p(x)=\hat x$ is the Pontrjagyn duality isomorphism (recall both $N,M$ are cyclic Abelian groups), and $kM$ is a $N$-bimodule canonically through the dual left- and right- actions
\begin{equation}
    (x\cdot y)(x') = y(x^{-1}x'),\qquad (y\cdot x)(x') = y(x'x),\label{dualaction}
\end{equation}
where $x,x'\in N$ and $y\in\widehat{N}=M$. This two-sided action dualizes to the grouplike coproduct $\Delta_0$ in  \eqref{grouplike}, as required by self-duality and \eqref{dualcoprod}. We call $(D(BM),\cdot,\Delta,S)$ the (skeletal) Drinfel'd double of $M$.

\medskip

Another important property of a Drinfel'd double $D(BM)$ is its \textbf{factorizability}, which means that it fits into a cospan of 2-bialgebra injections 
\begin{equation}
    kBM \hookrightarrow D(BM) \hookleftarrow kN. \nonumber
\end{equation}
If these injections $kBM\otimes kN\rightarrow D(BM)$ induce an equivalence $D(BM)\sim kBM\bar{\bowtie}kN$, then this forbids $D(BM)$ from carrying a non-trivial Postnikov class. This is because the Postnikov class of the bicrossed product $kBM\bar{\bowtie}kN$ \cite{Chen:2023tjf} is constructed out of those of each of its sectors $kM,kN$, which have trivial Hochschild cohomology groups
\begin{equation}
    HH^3(\ast,kM) = 0,\qquad HH^3(kN,\ast) = 0.\nonumber
\end{equation}
Therefore, in order to allow $D(BM)$ to carry a non-trivial Postnikov class, these 2-bialgebra injections cannot extend to an equivalence. 

In any case, we shall allow $D(BM)$ to carry a non-trivial Hochschild 3-cocycle $\cT\in HH^3(kN,kM)$. As $M=\widehat{N}\cong N$ in this case, {\it Remark \ref{2alg-2grp}} states that $\cT$ comes from the Postnikov class $\tau\in H^3(N,M)$ of its underlying 2-group $\mathbb{D} = (M,N,\rhd)$, hence we shall abuse notation and denote the Hochschild 3-cocycle $\cT$ by $\tau$ as well.

\subsubsection{Weakening the 2-bialgebra structure}\label{weak2bialgstructure}
Going back to {\bf Definition \ref{assoc2bialg}}, let $(\cA,\cdot,\Delta)$ denote a 2-bialgebra. There is a corresponding notion of weak 2-algebras and weak 2-bialgebras, which is accomplished by violating associativity/coassociativity in a homotopical manner. More precisely, we introduce the following maps
\begin{equation}
    \cT:\cA_0^{3\otimes}\rightarrow \cA_{-1},\qquad \Delta_1:\cA_0\rightarrow\cA_{-1}^{3\otimes},\nonumber
\end{equation}
which measures the failure of the associativity in $\cA$. We call these maps $\cT,\Delta_1$ the {\it homotopy associator/coassociator} of the 2-bialgebra $(\cA,\cdot,\Delta)$, respectively, and they are required to satisfy a particular algebraic condition captured by certain Hochschild cocycle conditions. For details, see \cite{Chen:2023tjf}.

\medskip

In a mathematically more elegant way, a "weak 2-algebra" $(\cA,\cdot,\cT)$ is simply an $A_\infty$-algebra $(\cA,m_n)$ with non-trivial $n$-nary product $m_n\in \operatorname{Hom}^{2-n}(A^2,A)$ for $n\leq 3$ (see for instance \cite{Stasheff:1963}). The map $m_1 = t:\cA_{-1}\rightarrow \cA_0$and the binary product $\cdot=m_2$ makes $A$ into a graded differential algebra, and $m_3=\cT$ plays the role of the homotopy associator --- when $m_3=0$, we get precisely the associative formulation above (and also in \cite{Wagemann+2021}). Maps between weak 2-algebras are therefore $A_\infty$-algebra homomorphisms.
\begin{definition}\label{2alghomdef}
    A \textbf{2-algebra homomorphism} $(f,\varphi): (\cA,\cdot,\cT)\rightarrow (\cA',\cdot',\cT')$ between two weak 2-algebras consists of a chain map $f=(f_1,f_0): \cA\rightarrow \cA'$, such that the following diagram commutes,
        \begin{equation}
            \begin{tikzcd}
\cA_{-1} \arrow[r, "t"] \arrow[d, "f_1"] & \cA_0 \arrow[d, "f_0"] \\
\cA_{-1}' \arrow[r, "t'"]                & \cA_0'                
\end{tikzcd},\label{commsq}
        \end{equation}
        and a bilinear map $\varphi:\cA_0^{2\otimes}\rightarrow \cA_{-1}'$ such that
        \begin{enumerate}
    \item for each $x,x_1,x_2\in\cA_0$ and $y\in\cA_{-1}$, we have
    \begin{gather}
    t'(\varphi(x_1,x_2))  = f_0(x_1x_2)  - f_0(x_1)f_0(x_2),\nonumber\\
    \varphi(ty,x) = f_1(y\cdot x) - f_1(y)\cdot' f_0(x),\qquad \varphi(x,ty) = f_1(x\cdot y) - f_0(x)\cdot' f_1(y),\label{weak2alg}
\end{gather}
and 
\item for each $x_1,x_2,x_3\in \cA_0$, we have 
\begin{eqnarray}
    \cT'(f_0(x_1),f_0(x_2), f_0(x_3)) - f_1(\cT(x_1,x_2,x_3)) &=& f_0(x_1)\cdot' \varphi(x_2,x_3) -\varphi(x_1x_2,x_3) \nonumber\\
    &\qquad& +~ \varphi(x_1,x_2,x_3) - f_1(x_1,x_2)\cdot' f_0(x_3)).\label{weak2hom}
\end{eqnarray}
\end{enumerate}
\end{definition}
\noindent On the other hand, a weak 2-coalgebra $(\cA,\Delta,\Delta_1)$ is a {\it shifted} $A_\infty$-coalgebra $(\cA,D_n[1])$, whose $n$-nary coproduct $D_n\in \operatorname{Hom}^{2-n}(\cA[1],\cA^2[1])$ for $n=2,3$ describes the coproduct $D_2=\Delta$ and the coassociator homotopy $D_3=\Delta_1$.

\medskip

Provided certain compatibility conditions between these maps hold, we form a \textit{weak 2-bialgebra} (or, from now on, simply a 2-bialgebra) and the duality result {\bf Proposition \ref{dual2bialg}} extends to the weak case \cite{Chen:2023tjf}. The coassociator $\Delta_1$ is of particular importance, as it defines the natural associator morphisms $a$ on $\operatorname{2Rep}_\text{wk}(\cA)$. The corresponding pentagon relations are satisfied provided the following compatibility condition\footnote{Here we are using a shorthand for tensor extensions of the coproduct maps, in which 
\begin{eqnarray}
    \Delta_{-1}\circ\Delta_{-1} &\equiv& (\Delta_{-1}\otimes 1)\circ\Delta_{-1} - (1\otimes \Delta_{-1})\circ\Delta_{-1},\nonumber\\
    (\Delta_{-1} + \Delta_0)\circ\Delta_0 &\equiv& \left[(\Delta_{-1}\otimes 1)\circ \Delta_0^l - (1\otimes\Delta_0^l)\circ\Delta_0^l\right] \nonumber\\
    &\qquad&+~ \left[(1\otimes\Delta_{-1})\circ\Delta_0^r -(\Delta_0^r\otimes 1)\circ\Delta_0^r\right].\nonumber
\end{eqnarray}}
\begin{equation}
    \Delta_{-1}\circ\Delta_1 = \Delta_1\circ\Delta_0\label{2coass}
\end{equation}
holds \cite{Chen:2023tjf}.

Given Abelian cyclic groups $N,M$, what is important for us here in this paper is the Drinfel'd double $D(BM) =  kBM\bar\bowtie kN $, in which $M = \widehat{N}$. This particular 2-bialgebra has two important features: it is {\it skeletal} and {\it self-dual}. The fact that it is skeletal means that 
% $\cT,\Delta_1$ do not directly participate in weakening the associativity/coassociativity of the underlying 2-bialgebra structure (cf. the 2-group case mentioned in Section \ref{catgrp}), and 
all adjoint equivalences (ie. the "mixed associator" 2-morphisms such as $a_{VWk}: a_{VWU}\Rightarrow a_{VWU'}$ for a 2-intertwiner $k:U\rightarrow U'$) are isomorphisms, and hence the only possibly non-trivial data endowed by $\Delta_1$ are the associator 1-morphisms
\begin{equation}
    a_{VWU}: (V\otimes W)\otimes U\rightarrow V\otimes (U\otimes W) \nonumber
\end{equation}
on 2-representations $V,W,U\in\operatorname{2Rep}_\text{wk}(D(BM))$. 

Secondly, the fact that $D(BM)$ is self-dual means that the coassociator $\Delta_1$ is dual to the Hochschild 3-cocycle $\cT$,
\begin{equation}
    \langle x,\tau(x_1,x_2,x_3)\rangle = \langle \Delta_1(x),x_1\otimes x_2\otimes x_3\rangle,\qquad x_1,\dots,x_4\in \big( D(BM) \big)_0 =kN.\nonumber
\end{equation}
Under this duality, the condition  \eqref{2coass} is  equivalent to the Hochschild 3-cocycle condition for $\cT$. By {\it Remark \ref{2alg-2grp}}, such Hochschild 3-cocycles are determined by the classifying Postnikov class of the 2-group underlying the Drinfel'd double 2-bialgebra $D(BM)$.

In summary, the Postnikov class $\tau \in H^3(N,M)$ (i) determines the Hochschild class $[\cT]$ of $D(BM)$, which in turn (ii) dualizes to the coassociator $\Delta_1$, and which in turn (iii) gives rise to the associator 1-morphism $a$ on $\operatorname{2Rep}_\text{wk}(D(BM))$. In other words, an element in $H^3(N,M)$ determines the associator data $a$ on $\operatorname{2Rep}_\text{wk}(D(BM))$, consistent with what one expects from the literature \cite{Zhu:2019,Wen:2019,Johnson-Freyd:2020}.

\subsection{2-representations of the Drinfel'd double}
The key result in \cite{Chen:2023tjf} states that the 2-representations of a Drinfel'd double $D(BM)$ forms a braided monoidal 2-category as defined in \cite{GURSKI20114225,KongTianZhou:2020}. We shall recall how the braiding and fusion structures of the 2-representations are induced by the 2-Hopf algebra structure of $D(BM)$.

Recall that a Baez-Crans 2-vector space is equivalently a 2-term chain complex $V = V_{-1}\xrightarrow{\partial}V_0$ of vector spaces \cite{Baez:2010ya} over $k$ --- namely a nuclear 2-algebra \cite{Wagemann+2021} or an Abelian Lie 2-algebra \cite{Bai_2013,Angulo:2018} --- and they form a linear 2-category $\operatorname{2Vect}^{BC}$ \cite{Baez:2003fs}. In its {\it weakened} form, namely in the context of $\operatorname{2Vect}^{hBC}$ as opposed to $\operatorname{2Vect}^{BC}$, the endormophisms $\operatorname{End}(V)$ of $V\in\operatorname{2Vect}^{hBC}$ is modelled by a weak 2-algebra as described in Section \ref{weak2bialgstructure}, 
\begin{equation*}
    \operatorname{End}(V) =\operatorname{End}(V)_{-1}\xrightarrow{\delta}\operatorname{End}(V)_0,
\end{equation*}
where $\operatorname{End}(V)_{-1}$ denotes the space of chain homotopies between, and $\operatorname{End}(V)_0$ are \textit{homotopy} chain maps on $V$. The map $\delta: \mu\mapsto (\mu\partial,\partial\mu)$ pre-/post-composes a chain homotopy $\mu$ with the chain map $\partial$ of $V$. 

Crucially, the chain maps in $\operatorname{End}(V)_0$ compose only up to chain homotopy; in other words, we have a homotopy associator $\T: \operatorname{End}(V)_0^{3\otimes}\rightarrow \operatorname{End}(V)_{-1}$ which witnesses the associativity in $\operatorname{End}(V)_0$, precisely as one would expect in a 2-bialgebra/2-term $A_\infty$-bialgebra. The presence of $\T$ leads to a proliferation of very rich structures supported by the 2-representation 2-category $\operatorname{2Rep}_\text{wk}(D(BM))$ of (weak) 2-representations of $D(BM)$ \cite{Chen:2023tjf}, as we now describe.

% is equivalently a  There is a natural associative 2-algebra $\operatorname{End}(V)$ of linear transformations on $V$, given by \cite{Angulo:2018}
% \begin{eqnarray}
%      \operatorname{End}(V)_0 &=& \{(f,g)\in\operatorname{End}(V_{-1})\times \operatorname{End}(V_0) \mid \partial m = n\partial\},\nonumber\\
%      \operatorname{End}(V)_{-1} &=&\{\phi\in\operatorname{Hom}(V_0,V_{-1})\mid (\phi\partial,\partial \phi) \in \operatorname{End}(V_{-1})\times\operatorname{End}(V_0)\},\nonumber
% \end{eqnarray}
% equipped with the 2-algebra structure
% \begin{equation}
%      \delta: \phi\mapsto (\phi\partial,\partial \phi),\qquad (f,g)\cdot \phi = f\phi,\qquad \phi\cdot(f,g)=\phi g.\nonumber 
% \end{equation}
% The associativity of matrix multiplication implies that $\operatorname{End}(V)_{-1}$ is clearly a $\operatorname{End}(V)_0$-bimodule, Moreover, $\delta$ is two-sided equivariant
% \begin{eqnarray}
%      \delta((f,g)\cdot \phi) &=& (f\phi\partial,\partial f\phi) = (f\phi\partial,g\partial \phi)= (f,g)\delta(\phi),\nonumber\\ 
%      \delta(\phi\cdot(f,g)) &=& (\phi g\partial,\partial \phi g) = (\phi\partial f,\partial \phi g) = \delta(\phi)(f,g),\nonumber
% \end{eqnarray}
% and we also have the Peiffer identity (note $\phi,\phi'\in\operatorname{End}(V)_{-1}$)
% \begin{equation*}
%       \phi\cdot \phi'= \phi\partial \phi'= \delta(\phi)\cdot \phi' = \phi\cdot\delta(\phi'),\nonumber
% \end{equation*}
% and hence $\operatorname{End}(V)$ is a 2-algebra. Note that none of the matrices here are required to be invertible.

\begin{definition}\label{2repcat}
A {\bf weak 2-representation} $(V,\rho)$ of $D(BM)$ on $V\in\operatorname{2Vect}^{BC}_\text{wk}$ is a weak 2-algebra homomorphism $(\rho,\varrho):D(BM)\rightarrow\operatorname{End}(V)$ as defined in {\bf Definition \ref{2alghomdef}}. The linear 2-category $\operatorname{2Rep}_\text{wk}(D(BM))$ of (weak) 2-representations of $D(BM)$ has the following morphisms.
\begin{enumerate}
    % \item The objects in $\operatorname{2Rep}_\text{wk}(D(BM))$ are {\bf weak 2-representations} together with a Hochschild 2-cocycle $\varrho:kN\rightarrow \operatorname{End}(V)_{-1}$ trivializing $\rho_1(\tau) - \T\circ \rho_0$. Namely
    % \begin{eqnarray}
    %      \T(\rho_0(x_1),\rho_0(x_2),\rho_0(x_3)) - \rho_{-1}(\tau(x_1,x_2,x_3)) &=& \rho_0(x_1) \varrho(x_2,x_3)-\varrho(x_1x_2,x_3) \nonumber\\
    %  &\qquad&+~ \varrho(x_1,x_2x_3) - \varrho(x_1,x_2) \rho_0(x_3) \nonumber
    % \end{eqnarray}
    % for each $x_1,\dots,x_3\in kN$.
    \item The 1-morphisms are \textbf{weak 2-intertwiners} $(I,i):\rho\rightarrow \rho'$, consisting of a 2-vector space homomorphism $i=(i_1,i_0):V\rightarrow W$ together with a Hochschild 1-cocycle $I_{\cdot,i}:kN\rightarrow \operatorname{End}(V)_{-1}$ fitting into the following commutative diagrams,
    \begin{equation}
    \begin{tikzcd}
V \arrow[rr, "i"] \arrow[dd, "\rho_0(x)"'] &                         & W \arrow[dd, "\rho'_0(x)"] \\
                                           & {\xRightarrow{I_{x,i}}} &                            \\
V \arrow[rr, "i"]                          &                         & W                         
\end{tikzcd},\qquad\qquad \begin{tikzcd}
V_0 \arrow[r, "\rho_1"] \arrow[d, "i_0"] & V_{-1} \arrow[d, "i_{1}"] \\
W_0 \arrow[r, "\rho_1'"]                 & W_{-1}                    
\end{tikzcd},\label{2int}
\end{equation}
such that $I_{\cdot,i}$ trivializes $\varrho - \varrho'$. Namely 
\begin{equation}
    \varrho'(x,x')\cdot \operatorname{id}_i - \operatorname{id}_i\cdot \varrho(x,x') = \operatorname{id}_{\rho_0(x)}\cdot I_{x',i} - I_{xx',i} + I_{x,i} \cdot\operatorname{id}_{\rho_0(x')}\nonumber
\end{equation}
for each $x_1,x_2\in kN$, where $\operatorname{id}_i:i\Rightarrow i$ denotes the identity 2-morphism on the intertwiner $i$.

\item The 2-morphisms/modifications are {\bf equivariant cochain homotopies} $\mu:i\Rightarrow i'$ that trivializes $I_{\cdot,i} - I_{\cdot,i'}$, namely
\begin{equation}
    I_{x,i} - I_{x,i'} = \operatorname{id}_{\rho_0(x)}\cdot \mu - \mu\nonumber
\end{equation}
for all $x\in kN$. More explicitly, $\mu$ is a cochain homotopy
\begin{equation}
    \begin{tikzcd}
V_{-1} \arrow[r, "\partial"] \arrow[d, "i_1-i_1'"'] & V_0 \arrow[ld, "\mu"'] \arrow[d, "i_0-i_0'"] \\
W_{-1} \arrow[r, "\partial'"]                      & W_0                                        
\end{tikzcd}\label{modif}
\end{equation}
that intertwines between $\rho_1(y),\rho_1'(y)$ as cochain homotopies for each $y\in M$.
\end{enumerate}
\end{definition}

As the Baez-Crans 2-vector spaces by definition do not have homotopy associators $\T$ for its endomorphisms, one sees that 2-representations based on $\operatorname{2Vect}^{BC}$ would fail to carry any non-trivial chain homotopy $\varrho$, and hence also the chain homotopy $I$ for its 2-intertwiners. These pieces of homotopy data are critical for recovering the higher-coherence data present in 4d gapped topological phases \cite{Kong:2020,Kong:2020wmn,KongTianZhou:2020,Johnson-Freyd:2020,Johnson-Freyd:2020usu,Else:2017yqj,Wen:2019}.

\subsubsection{Tensor product and direct sums}\label{tensor2repthy} Recall vector space cochain complexes $V,W$ come equipped with natural notions of direct sum
\begin{equation}
    V\oplus W = V_{-1}\oplus W_{-1}\xrightarrow{\partial\oplus\partial'}V_0\oplus W_0,\nonumber
\end{equation}
as well as tensor product
\begin{equation}
    V\otimes W = V_{-1}\otimes W_{-1}\xrightarrow{D^+}V_{-1}\otimes W_0 \oplus V_0\otimes W_{-1}\xrightarrow{D^-} V_0\otimes W_0,\label{grtensor}
\end{equation}
where $D^\pm = \pm1\otimes\partial'+\partial\otimes 1$ is the tensor extension of the differentials $\partial:V_{-1}\rightarrow V$ and $\partial':W_{-1}\rightarrow W_0$.

For 2-representations, the direct sum is defined by extending the definition to a direct sum of 2-algebra homomorphisms
\begin{equation}
    \rho\oplus \rho': \cA\oplus\cA \rightarrow \operatorname{End}(V)\oplus\operatorname{End}(W),\nonumber
\end{equation}
while the tensor product is accomplished by precomposing with the coproduct. By interpreting chain homotopies $\rho_1(y)$ as an "action" of $y\in M$ on the 2-intertwiners $i,j$ of $\operatorname{2Rep}_\text{wk}(D(BM))$, we can put formally
\begin{gather}
    \rho_{V\otimes W} = (\rho_0\otimes \rho_0)\circ\Delta_0',\qquad \rho_{i\otimes j} = (\rho_1\otimes\rho_1)\circ\Delta_{-1},\nonumber\\
    \rho_{V\otimes j} = (\rho_0\otimes\rho_1) \circ \Delta_0,\qquad \rho_{i\otimes W} = (\rho_1\otimes\rho_0)\circ \Delta_0.\label{tensor2rep}
\end{gather}
The coequivariance and the coPeiffer conditions of the graded coproducts $\Delta_{-1},\Delta_0,\Delta_0'$ then implies the naturality of this tensor product, as well as the decomposition \cite{Chen:2023tjf}
\begin{equation}
    \rho_{i\otimes j} = \rho_{V'\otimes j}\rho_{i\otimes W} = \rho_{i\otimes W'}\rho_{V\otimes j}.\label{tensordecomp}
\end{equation}
Now crucially, the weak component $\varrho$ attached to these decompositions \eqref{tensordecomp} may differ; their difference gives rise to a possibly non-trivial \textit{interchanger 2-isomorphism}
\begin{equation}
    \phi_{ij}: (\id_{V'}\otimes j) \circ (i\otimes\id_{W}) \xRightarrow{\sim} (i\otimes \id_{W'})\circ (\id_V\otimes j) \label{interchange}
\end{equation}
for the 2-intertwiners $i,j$ in $\operatorname{2Rep}_\text{wk}(\cA)$. This fact will play an important role in Section \ref{fermifusion}.

% In the following, it would be useful to consider components of $\rho$ given by cochain homotopies, such as $\rho_1(y)$ or $\rho_{i\otimes W}(x)$, etc., as an "action" on the 2-intertwiners.

\medskip

Notice since the grouplike coproduct  \eqref{grouplike} is cocommutative, the induced tensor product  \eqref{tensor2rep} is commutative, as we shall see in Sections \ref{toriccharges} and \ref{ferkitcharges}. The zero 2-representation is the zero complex $0\xrightarrow{0} 0$, while the tensor unit 2-representation is the trivial complex $\mathbf{1}=k\xrightarrow{1} k$ carrying the action $\rho(-)\cdot \mathbf{1} = \epsilon(-)\cdot\mathbf{1}$ by a scalar multiplication through the counit $\epsilon:D(BM)\rightarrow k$. We briefly note that the tensor product of 2-morphisms is induced by composition $\mu\otimes\mu' = \mu\cdot\mu'$. The reader is referred to \cite{Chen:2023tjf} for an explicit description of the above structures and formulas.

% Note the quantity $\rho_{i\otimes W}$ lies in $\operatorname{End}(V)_0\otimes\operatorname{End}(V)_{-1}\oplus\operatorname{End}(V)_{-1}\otimes\operatorname{End}(V)_0$, and is in fact a 2-morphism. As such, we shall also require modifications $\mu$ in $\operatorname{2Rep}_\text{wk}(D(BM))$ to commute with $(\rho_{i\otimes W})(x),(\rho_{V\otimes j})(x)$ for each $x\in kN$. 

\paragraph{Dual/contragredient 2-representations.} Let $V \in \operatorname{2Rep}_\text{wk}(D(BM))$ be a 2-representation. As the Drinfel'd double 2-bialgebra $D(BM)$ is equipped with the (grouplike) antipode $S$, we can in fact define the {\bf dual 2-representation} $\bar V$ by
\begin{equation}
    \rho_{\bar V} = \rho_V\circ S.\nonumber
\end{equation}
Since $S$ is bijective in this case, all 2-representations and 2-intertwiners of $D(BM)$ has a dual and hence $\operatorname{2Rep}_\text{wk}(D(BM))$ is a {\it fully-dualizable}. This is a necessary step for {\bf Theorem \ref{mainthm}}, since tensor 2-categories do indeed have duals \cite{etingof2016tensor,KongTianZhou:2020,Johnson_Freyd_2023}, but it shall not play an important role in this paper. 

% Now in general, the antipode comes with an additional component $S_1:kN\rightarrow kBM$ that enters  \eqref{antpode} by 

\subsubsection{Braiding} In general, Drinfel'd double 2-bialgebras such as $D(BM)$ come equipped naturally with a {\bf 2-$R$-matrix} 
\begin{equation}
        \cR = \cR^+ + \cR^-,\qquad \begin{cases}\cR^+ \in kBM \otimes kN \\ \cR^-\in kN\otimes kBM\end{cases},\nonumber
\end{equation}
which satisfy certain algebraic properties \cite{Chen:2023tjf}. There is also defined a component $R \in kN\otimes kN$, which serves as a universal $R$-matrix for $kN$.

We shall make use of the conventional Sweedler notation for the 2-$R$-matrices,
\begin{equation}
    \cR^\pm = \cR_{(1)}^\pm \otimes \cR_{(2)}^\pm,\qquad R = R_{(1)}\otimes R_{(2)}. \nonumber
\end{equation}
If $t\neq 0$, this quantity $R$ is subject to the constraint 
\begin{equation}
    R = \cR^-_{(1)}\otimes t\cR^-_{(2)} (\equiv R^-) = t\cR^+_{(1)}\otimes \cR^+_{(2)}(\equiv R^+),\nonumber
\end{equation}
while if $t=0$, $R$ and $\cR$ are independent.

\smallskip

Now take two 2-representations $V,W\in\operatorname{2Rep}_\text{wk}(D(BM))$; we define the braiding map between $V,W$ by
\begin{equation}
    b_{V,W}(V\otimes W) =\text{flip}\circ \rho_0(R)(V\otimes W).\label{2repbraid1}
\end{equation}
We call $B_{V,W} =b_{V,W}\circ b_{W,V}$ the {\it full-braiding}, and $b_V=b_{V,V}$ the {\it self-braiding} of $V$. Provided $V,W$ are invertible and can be self-braided, we have \cite{Majid:1996kd,Johnson-Freyd:2020}
\begin{equation}
    b_{V\otimes W} \cong (V\otimes b_{V,W}\otimes W) \circ (b_V \otimes b_W)\circ (V\otimes b_{W,V} \otimes W).\label{ribboneq}
\end{equation}

Now similar to the tensor product structure, we can form the \textit{mixed braiding} between 1-morphisms $i$ and objects $W$, which is formally given by
\begin{equation}
    b_{i,W} =\text{flip}  \circ (\rho_1\otimes\rho_0)(\cR^+),\qquad b_{W,i} =\text{flip}  \circ (\rho_0\otimes\rho_1)(\cR^-).\label{2repbraid2}
\end{equation}
See \cite{Chen:2023tjf} for more explicit formulas.
% More explicitly, $b_{i,W} = b_{i,W}^1 + b_{i,W}^2$ determines two maps 
% \begin{equation}
%     b_{i,W}^1: V_0\otimes W_0\rightarrow (W_{-1}\otimes U_0)\oplus (W_0\otimes U_{-1}),\qquad b_{i,W}^2: (V_{-1}\otimes W_0)\oplus (V_0\otimes W_{-1}) \rightarrow W_{-1}\otimes U_{-1},\nonumber
% \end{equation}
% the latter of which carries the non-trivial sign $(-1)^{\text{deg}} = -1$. 
There, we have also proved the following naturality property
\begin{equation}
    b_{i,W}: i\circ b_{V,W}\Rightarrow b_{U,W}\circ i,\qquad b_{W,i}: i\circ b_{W,V}\Rightarrow b_{W,U}\circ i,\nonumber
\end{equation}
where $i:V\rightarrow U$ is a 2-intertwiner in $\operatorname{2Rep}(D(BM))$.

\paragraph{The hexagonator.} By weakening the Drinfel'd double 2-bialgebra via Section \ref{weak2bialgstructure}, we also introduce a {\it hexagonator} 2-morphism $\Omega_{V|\bullet\bullet}$ that implements the hexagon relation for each $V\in \operatorname{2Rep}_\text{wk}(D(BM))$. More precisely, given three 2-representations $V,W,U$, we have the following diagram
\begin{equation}
\begin{tikzcd}
                                                                                  & V(W U) \arrow[rr, "b_{V,WU}"] &  & (W U) V \arrow[rd, "a_{WUV}"] &                                  \\
(V W) U \arrow[ru, "a_{VWU}"] \arrow[rd, "b_{V,W}"'] & {} \arrow[rr, "\Omega_{V|WU}", Rightarrow]                             &  & {}                                                          & W (U V) \\
                                                                                  & (W V) U \arrow[rr, "a_{WVU}"']    &  & W (V U) \arrow[ru, "b_{V,U}"']      &                                 
\end{tikzcd},\nonumber
\end{equation}
as well as its adjoint $\Omega^\dagger_{V|WU}$. This 2-morphism results from applying the component $\varrho$ of the weak 2-representation $\rho= (\rho_1,\rho_0,\varrho)$ to one of the defining relations of a 2-$R$-matrix.

It is a non-trivial theorem (and in fact the main result) in \cite{Chen:2023tjf} that the various structures --- namely the monoidal/associator and the braiding/hexagonator data --- induced by the 2-representation theory of weak 2-bialgebras are mutually cohesive: $\operatorname{2Rep}_\text{wk}(D(BM))$ forms a braided monoidal 2-category in the sense of \cite{GURSKI20114225}. Since this hexagonator datum does not play a central role in the following, we shall direct the reader to the aforementioned works for a more complete description.

\section{The 4d Kitaev model}\label{4dkit}
We are now finally ready to study the   4d Kitaev model. In this section, we specialize the above Drinfel'd double 2-bialgebra $D(BM) = D(B\bbZ_2)$ to the case $N = \bbZ_2, M = \widehat{\bbZ_2} \cong \bbZ_2$, and construct (3+1)D 2-BF theory based on $D(BM)$. We study its extended $\bbZ_2$-charged excitations by studying $\operatorname{2Rep}_\text{wk}(D(B\bbZ_2))$, and seek to prove the main result {\bf Theorem \ref{mainthm}}. 

Along the way, we shall make concrete the connection between our 2-BF theory and the higher-gauge topological nonlinear $\sigma$-models (NLSMs) that have already appeared in the literature \cite{Zhu:2019,Wen:2019}. We shall take the ground field $k=\bbC$ throughout the following. 

\subsection{Monster 2-BF theory on the Drinfel'd double 2-bialgebra $D(BM)$}
Recall the Drinfel'd double 2-bialgebra $D(B\bbZ_2)$ has the structure of a skeletal 2-algebra
\begin{equation}
    D(B\bbZ_2)=k\widehat{\bbZ_2}\xrightarrow{0}k\bbZ_2,\nonumber
\end{equation}
whose Hochschild class is determined by the choice of a Postnikov class
\begin{equation}
    \tau\in H^3(\bbZ_2,\widehat{\bbZ_2}) \cong \bbZ_2 \nonumber
\end{equation}
of the underlying 2-group. Let $x\in k\bbZ_2$ and $y\in k\widehat{\bbZ_2}$ be understood as the non-trivial generators.

Let $k\bbZ_2=kN$ in degree-0 act on $k\widehat{\bbZ_2}= kM$ on the left by  \eqref{dualaction} as group algebras. There are two such algebra automorphisms: the trivial or the sign representation. We denote the Drinfel'd double 2-bialgebra by $D(B\bbZ_2)^\text{trv}$ in the former case, while by $D(B\bbZ_2)^\text{sgn}$ in the latter case. This then induces a non-trivial grouplike component $\Delta_0(x)= \hat x\otimes x+x\otimes\hat x$ of the coproduct $\Delta$ on $D(B\bbZ_2)$ (recall $\hat x = p(x)$ where $p$ is the Pontrjagyn duality).

Now consider the discrete combined $D(B\bbZ_2)$-connection $({\bf A},\boldsymbol\Sigma) = (A+\Sigma,C+B)$ on a 4-manifold $X$ \cite{chen:2022}. These connection forms are given by cochains
\begin{equation}
    A\in C^1(X,\bbZ_2),\qquad B\in C^2(X,\widehat{\bbZ_2}),\nonumber
\end{equation}
with the components $\Sigma=0,C=0$ trivial. Depending on the automorphism $\operatorname{Aut}(k\bbZ_2)$ encoded in the Drinfel'd double 2-bialgebra $D(B\bbZ_2)$, the 1- and 2-curvatures of the field theory are given by
\begin{equation}
    F=\begin{cases}dA &; \text{in $D(B\bbZ_2)^\text{triv}$} \\ dA + \frac{1}{2}A\cup A &;\text{in $D(B\bbZ_2)^\text{sgn}$}\end{cases},\qquad d_AB =\begin{cases}dB &; \text{in $D(B\bbZ_2)^\text{triv}$} \\ dB + A\cup B &;\text{in $D(B\bbZ_2)^\text{sgn}$}\end{cases},\nonumber
\end{equation}
where the cup products are implemented through the automorphism $\operatorname{Aut}(k\bbZ_2)$ or its dual. The corresponding \textit{monster 2-BF theory} \cite{chen:2022} is given by the topological action
\begin{equation}
    S[A,B] = \frac{1}{2}\int_X \langle B\cup F\rangle + \langle \tau(A)\cup A\rangle,\label{toric}
\end{equation}
where we recall that $\tau\in H^3(\bbZ_2,\widehat{\bbZ_2})$ is the underlying Postnikov class of $D(B\bbZ_2)$.

Note that the discrete 1-form gauge fields must be flat, $F=dA=0$, and terms like $A^2=0\mod 2$ vanish, hence the classical equations of motion (EOMs) are  given by
\begin{equation}
    F=dA=0,\qquad d_AB=\tau(A).\label{toriceom}
\end{equation}
These EOMs, together with the coefficient of $\frac{1}{2}$ in front of the topological action \eqref{toric}, tell us that the cochains $A,B$ are $\bbZ_2$-valued. We will introduce in the following a non-trivial cohomological term that "mimics" $\frac{1}{2}A^2$. However, it is important to note that these cohomological terms constitute {\it twists} on the Drinfel'd double 2-bialgebra and are not dynamical; they do not alter the EOM  \eqref{toriceom}.

\smallskip

We define the partition function corresponding to  \eqref{toric} on a 4-manifold $X$ as a formal path integral
\begin{equation}
    Z_\text{Kit}(X) = \int dAdB e^{i2\pi S[A,B]},\label{kit}
\end{equation}
which should be appropriately normalized such that $Z_\text{Kit}(S^4)=1$ \cite{Zhu:2019}. We call $Z_\text{Kit}$ the {\bf   4d Kitaev model}. It should be understood as a collective of two such theories\footnote{There is a slight misnomer here, where $Z_\text{Kit}^0$ should really be called the "invisible" toric code, as it fails to satisfy the principle of {\it remote detecatbility} \cite{Kong:2020,KongTianZhou:2020,Johnson-Freyd:2020}; see {\it Remark \ref{invisibletoric}} later.},
\begin{equation}
    \text{(Invisible) toric code}:~Z_\text{Kit}^0,\qquad\text{Spin-Kiatev}:~ Z_\text{Kit}^s,\nonumber
\end{equation}
arising respectively from $D(B\bbZ_2)^\text{trv}$ and $D(B\bbZ_2)^\text{sgn}$. We shall refer to either of these Drinfel'd double 2-bialgebras collectively as $D(B\bbZ_2)$ in the following. The central idea is then that $Z_\text{Kit}$ has a Drinfel'd double 2-bialgebra symmetry.

\subsubsection{$Z_\text{Kit}$ as a topological nonlinear $\sigma$-model} There had been proposals to construct (3+1)D topological phases with a higher-gauge field theory \cite{Kapustin2017}. Specifically, \cite{Zhu:2019} constructs a topological non-linear $\sigma$-model (NLSM) which corresponds to a higher-Dijkgraaf-Witten theory based on a 2-group, and claims that all (3+1)D topological phases can be described this way. 

The NLSM they construct is characterized by the following data: (i) a (skeletal) 2-group $\mathbb{G}=\bbZ_2\rightarrow G_b$, where $G_b$ is a finite group labeling "stringlike bosonic charges", and $\bbZ_2$ is either fermion parity $\bbZ_2^f$ or a magnetic $\pi$-flux $\bbZ_2^m$, (ii) the first Postnikov class $\tau\in H^3(G_b,\bbZ_2^f)$ of $\mathbb{G}$ and (iii) a Dijkgraaf-Witten class $\omega\in H^4(\mathbb{G},\R/\bbZ)$ \cite{Kapustin2017,Zhu:2019}. We write the Ho{\'a}ng data \cite{Ang:2018rls} of $\mathbb{G}$ as $(G_b,\bbZ_2^f,\tau)$.

Our construction of the Kitaev model  \eqref{kit} fits nicely into this framework, with the 2-group $\widehat{\bbZ_2}\xrightarrow{0}\bbZ_2$ given by the Ho{\'a}ng data $$(G_b=\bbZ_2,\bbZ_2^f\cong \widehat{\bbZ_2},\tau).$$ To construct the Dijkgraaf-Witten cocycle, we begin with the group cohomology ring $H^\ast(\bbZ_2,\bbZ_2)\cong \bbZ_2[u]$ with a generator $u\in H^1(\bbZ_2,\bbZ_2)$ in degree-1 \cite{book-groupcohomology}. Considering $\bbZ_2$ as a trivial $\bbZ_2$-module, the sign representation $\text{sgn}\in \operatorname{Aut}(k\bbZ_2)\cong \operatorname{Hom}(\bbZ_2,\bbZ_2)\cong H^1(\bbZ_2,\bbZ_2)$ then serves as a representative of the generator $u$. 

Now consider $D(B\bbZ_2)^\text{sgn}$. The cup product for the term $\frac{1}{2}A\cup A$ in the curvature $F$ is implemented by the sign representation $\text{sgn}=u\in H^1(\bbZ_2,\bbZ_2)$, from which
\begin{equation}
    \frac{1}{2}A\cup A = \bar e(A),\qquad \bar e=\frac{1}{2}u\cup u\in H^2(\bbZ_2,\bbZ_2).\label{ext}
\end{equation}
The factor of $1/2$ is very important as, without it, $u\cup u =0\mod 2$ is trivial in $\bbZ_2$-cohomology \cite{book-groupcohomology}. Dualizing the value of $\bar e$ to a class in $H^2(\bbZ_2,\widehat{\bbZ_2})$, it lifts the action $\rhd$ of $k\bbZ_2$ on $k\widehat{\bbZ_2}$ to a central extension. 

The term $\langle B\cup \bar e(A)\rangle$ that appears in  \eqref{toric} gives precisely the Dijkgraaf-Witten cocycle $\omega\in Z^4(G,\R/\bbZ)$. Indeed, going on-shell of the EOM  \eqref{toriceom} reduces the spin-Kitaev partition function to
\begin{equation}
    Z_\text{Kit}^s(X) \sim \sum_{\substack{dA=0 \\ dB=\tau}}e^{i\pi\int_X \langle B\cup \bar e(A)\rangle}.\nonumber
\end{equation}
This gives exactly the NLSM constructed in \cite{Zhu:2019} with $\omega(B,A) = B\cup \bar e(A)$, provided the {\bf anomaly-free condition}
\begin{equation}
    \tau \cup \bar e =0\label{anomfree}
\end{equation}
is satisfied. This condition ensures that that the Dijkgraaf-Witten integrand $\omega(A,B) = \langle B\cup \bar e(A)\rangle$ is a cocycle $d\omega(A,B) = 0$ in light of the EOM $dB=\tau(A)$.

\subsubsection{Classification of   4d topological phases with a single pointlike $\bbZ_2$-charge} The above describes the construction of a   4d Dijkgraaf-Witten topological field theory. As we have mentioned, these were proposed to describe \cite{Zhu:2019,Kapustin2017,Bullivant:2016clk,Wan:2014woa}, in a very general sense,   4d gapped topological phases. Another approach towards this follows the program of "higher categorical symmetries" \cite{Wen:2019,Kong:2014qka,Kong:2017hcw,Kong:2020,Johnson-Freyd:2020usu,Delcamp:2023kew,Hamma:2004ud,KitaevKong_2012}. In particular, the   4d toric code has been extensively studied in the literature \cite{Else:2017yqj,Kong:2020wmn,Johnson-Freyd:2020} from this perspective, so we understand its corresponding braided fusion 2-category quite well.

By hypothesis, gapped topological phases are characterized by non-degenerate\footnote{Namely the sylleptic/$E_2$-centre $Z_2$ is trivial.} braided fusion 2-categories, based on the physical principle of {\bf remote detectability} \cite{Kong:2014qka,Kong:2017hcw,Kong:2020,Johnson-Freyd:2020usu,KitaevKong_2012}. In particular, those with a single pointlike $\bbZ_2$-charge have been classified in \cite{Johnson-Freyd:2020,Johnson_Freyd_2023}. These phases are
\begin{enumerate}
    \item the   4d toric code $\mathscr{R} \simeq Z_1(\Sigma\operatorname{Vect}[\bbZ_2])$,
    \item the   4d spin-$\bbZ_2$ gauge theory $\mathscr{S} \simeq Z_1(\Sigma\operatorname{sVect})$, 
    \item the $w_2w_3$ gravitational anomaly $\mathscr{T}$,
\end{enumerate}
where $Z_1$ denotes the Drinfel'd centre and $\Sigma$ denotes the condensation completion functor \cite{Gaiotto:2019xmp}. Here, $\operatorname{Vect}[\bbZ_2]$ denotes the category of $\bbZ_2$-graded vector spaces, and $\operatorname{sVect}$ is the category of supervector spaces.

\begin{remark}\label{2vects}
Note $\Sigma\operatorname{Vect}[\bbZ_2] \simeq \operatorname{2Vect}^{KV}[\bbZ_2]$ is equivalent to the 2-category of $\bbZ_2$-graded 2-vector spaces as defined in the sense of Kapranov-Voevodsky \cite{Baez:2012,Douglas:2018,Johnson-Freyd:2020}. It is important to keep in mind that this notion of a 2-vector space is distinct from what we are using in this paper, which is in the sense of Baez-Crans \cite{Bai_2013,Baez:2003fs,Angulo:2018}. For this reason, we shall always denote $\operatorname{2Vect}^{KV}[\bbZ_2]$ in terms of the condensation completion $\Sigma\operatorname{Vect}[\bbZ_2]$. However, we have proven in \cite{Chen:2023tjf} that the 2-representation theory developed from {\bf Definition \ref{2repcat}} coincides with that studied in the literature \cite{Baez:2012,Douglas:2018,Delcamp:2021szr,Delcamp:2023kew} for finite skeletal 2-groups.
\end{remark}

In this paper, we shall mostly focus on the gapped phases $\mathscr{R},\mathscr{S}$, and leave the study of the gravitational anomaly $\mathscr{T}$ to a later work; the reason for this shall be given at the end of Section \ref{ferkitcharges}. We will find explicit realizations of these phases as 2-representation 2-categories of certain versions of the 2-quantum double $D(B\bbZ_2)$. To do so, we study the excitations in the associated NLSM  \eqref{toric}. 

\subsection{Anomaly-freeness of the   4d spin-Kitaev model}\label{anomalyfree}
Recall from the above that the   4d Kitaev model $Z_\text{Kit}$ is well-defined provided the non-trivial Postnikov class $\tau$ and extension class $\bar e$ of the underlying 2-group satisfies the anomaly free condition  \eqref{anomfree}.

Let us here study, from the point of view of the 2-representation 2-category $\operatorname{2Rep}_\text{wk}(D(B\bbZ_2))$, why the anomaly-free condition  \eqref{anomfree} is necessary. Recall from Section \ref{weak2bialgstructure} that the self-duality of $D(B\bbZ_2)$ as a Drinfel'd double 2-bialgebra means that the Postnikov class $\tau$ dualizes to a coassociator $\Delta_1:{\bbZ_2}\rightarrow \widehat{\bbZ_2}^{3\otimes}$ defining the
associator 1-morphism $a_{VWU}$ for the objects $V,W,U\in \operatorname{2Rep}_\text{wk}(D(B\bbZ_2))$.

The key point is that, in general, the pentagon relation for $a$ follows from the condition  \eqref{2coass}, which in turn follows from the 3-cocycle condition for $\tau$. This notion generalizes to the case where $D(B\bbZ_2)$ is twisted by $\bar e\in H^2(\bbZ_2,\widehat{\bbZ_2})$; that is, the product in $D(B\bbZ_2)_0 = \bbZ_2$ is modified such that
\begin{equation*}
    x\times x' = \bar e(x,x') (xx'),\qquad x,x'\in\bbZ_2.
\end{equation*}
We shall denote the corresponding 2-representation 2-category by
\begin{equation}
    \operatorname{2Rep}_\text{wk}(D^{\bar e}(B\bbZ)) = \operatorname{2Rep}_m^\tau(D(B\bbZ_2)^\mathrm{sgn}).\nonumber
\end{equation}
This notation shall be explained later in Section \ref{ferkitcharges}. For now, we prove the following.

\begin{lemma}\label{pentagonfree}
The anomaly-free condition  \eqref{anomfree} implies that the associator $a$ of $\operatorname{2Rep}_m^\tau(D(B\bbZ_2)^\mathrm{sgn})$ satisfies the pentagon relations.
\end{lemma}
\begin{proof}
In order to see the anomaly-free condition  \eqref{anomfree} manifest on the 2-representations, we begin with the observation that the component $\Delta_0$ of the coproduct on $D(B\bbZ_2)$ satisfies
\begin{equation}
    \Delta_0(x^2) = \bar e(x,x)\hat x^2\otimes x^2= \bar e(x,x)\otimes 1,\label{boscoprod}
\end{equation}
by  \eqref{2algcoprod}. This means that $\Delta_0$ is an algebra map on $k\bbZ_4$, not $k\bbZ_2$.

Because of  \eqref{boscoprod}, evaluating the condition  \eqref{2coass} on $1=x^2\in k\bbZ_2$ gives
\begin{equation}
    1^{4\otimes}= \Delta_{-1}\circ\Delta_1(1) = (\Delta_1\otimes 1)\circ \Delta_0(x^2) = \bar e(x,x)\otimes \Delta_1(1),\nonumber
\end{equation}
which violates the pentagon relation unless the right-hand side is also trivial $1^{4\otimes}$. Pairing this equation with arbitrary $x_1,\dots,x_3\in k\bbZ_2$ gives
\begin{equation}
    1=\langle \bar e(x,x)\otimes \Delta_1(1),1\otimes x_1\otimes x_2\otimes x_3\rangle = \langle \bar e(x,x)\otimes 1,1\otimes \tau(x_1,x_2,x_3)\rangle = \bar e(x,x)\tau(x_1,x_2,x_3).\nonumber
\end{equation}
This is nothing but $\bar e\cup \tau =0$. 
\end{proof}
\noindent Notice on the other hand that if $\tau=0$ is trivial, then so is $\Delta_1$ and the coassociativity condition simply implies the group cocycle condition for $\bar e$.

\paragraph{Weakening the anomaly-free condition.} There is a way to weaken the anomaly-free condition, by imposing  \eqref{anomfree} only in {\it cohomology} $\tau\cup\bar e =0\in  H^5(\bbZ_2,k^\times)$ \cite{Kapustin2017,Cui_2017}. This means that the   4d Kitaev model gains an additional term $\nu(A)$ that trivializes the coboundary of the Dijkgraaf-Witten 4-cocycle, $$d(\omega_b-\nu) = 0.$$ Algebraically, this 4-cocycle $\nu\in H^4(\bbZ_2,k^\times)$ is known to play the role of a "pentagonator" 2-morphism in the underlying 2-group \cite{Wen:2019,Zhu:2019}, implementing the pentagon relation  \eqref{pentagon}; see {\it Remark \ref{gencatgrp}}. 

This group 4-cocycle $\nu$ is intimately related to the Hochschild 3-cocycle $\mathfrak{T}$ attached to the {\it weak} endomorphism 2-algebra $\operatorname{End}(V)$ (see {\bf Definition \ref{2repcat}}) of a 2-representation $V$. From the 2-representation theory perspective, it was shown \cite{Chen:2023tjf} that the module pentagonator $\pi$ \cite{Delcamp:2023kew} on $\operatorname{2Rep}_\text{wk}(D(B\bbZ_2))$ (in the untwisted case) is given by
\begin{equation*}
    \pi_{x_1x_2x_3|V} = \mathfrak{T}(\rho_0(x_1),\rho_0(x_2),\rho_0(x_3))(V),\qquad x_1,\dots,x_3\in k\bbZ_2,
\end{equation*}
where $\rho=(\rho_1,\rho_0,\varrho): D(B\bbZ_2) \rightarrow \operatorname{End}(V)$ is a 2-representation on $V$. 

Given $V$ is irreducible with an associated label $x_4\in\bbZ_2$, then $\pi_{x_1x_2x_3|x_4}=\nu(x_1,x_2,x_3,x_4)$ defines a group 4-cocycle whose cocycle condition arises from the associahedron condition \cite{Delcamp:2023kew,Chen:2023tjf}. Moreover, the fact that the Hochschild cohomology of
\begin{equation}
    \rho_1\circ \tau - \mathfrak{T}\circ \rho_0^{3\otimes}\nonumber
\end{equation}
is trivial (see {\bf Definition \ref{2alghomdef}}) translates to precisely the equation $d(\omega_b-\nu)=0$. This relationship between $\tau$ and $\nu$ is intimately related to the conjecture \cite{Wen:2019} that the 4d $\nu$-twisted gauge theory on $G$ coincides with the 4d untwisted 2-gauge theory on $\mathbb{G}= (G,k^\times,\tau)$. We will say more about this in a future work, but it is worth emphasizing again here that one would not be able to see the pentagonator in the (strict) Baez-Crans 2-category $\operatorname{2Vect}^{BC}$.

% The proof of the above lemma then implies that this pentagonator $\nu$ appears, upon dualization, as an element $\Pi\in k\widehat{\bbZ_2}^{4\otimes}$ that gives rise to a cochain homotopy $((\rho_V)_1\otimes(\rho_W)_1\otimes(\rho_U)_1\otimes(\rho_X)_1)(\Pi)$ between the associator morphisms $a$ involved in the pentagon relation  \eqref{pentagon}. Equivalently, this element $\Pi$ modifies  \eqref{2coass} in a way such that
% \begin{equation}
%     \Pi \cdot (\Delta_1\circ\Delta_0^{\bar e}) = \Delta_{-1}\circ\Delta_1,\nonumber
% \end{equation}
% where the superscript of $\bar e$ reminds us of the twist by the group 2-cocycle as in  \eqref{boscoprod}. We shall not dwell on this point here, however.

% \subsection{The spin-Kitaev model $Z_{\text{Kit}}^s$}

% Here, the skeletal 2-group $G=\widehat{\bbZ_2}\xrightarrow{0}\bbZ_2$ is treated as a categorical group, in which $y,y',y''\in \widehat{\bbZ_2}$ denote morphisms over the objects $x,x',x''\in \bbZ_2$, respectively. Therefore $\tilde\tau\in H^3(\widehat{\bbZ_2},k^\times)$ determines a set of \textit{2-morphism}, whose higher-pentagon relation follows from the 2-cocycle condition for $\bar c$ and the pentagon relation for $\tau$ \cite{Wen:2019,Zhu:2019}.

% We wish to also see this 2-associator appear on the side of the 2-representation 2-category $\operatorname{2Rep}_f^\tau(D(B\bbZ_2)^\text{sgn})$.

% *************************************

% \medskip

\section{Excitations in the (invisible) toric code $Z_\text{Kit}^0$}\label{toriccharges}
Excitations are inserted into the theory $Z_\text{Kit}$ with 2-representations $\rho$ of $D(B\bbZ_2)$. Since $D(B\bbZ_2)$ is skeletal, it suffices to study 2-representations of the underlying 2-group. Let us first focus on the trivial case $D(B\bbZ_2)^\text{trv}$. 

Recall that a 2-representation $\rho:D(B\bbZ_2)^\text{trv}\rightarrow \operatorname{End}(V)$ on a 2-vector space $V=V_{-1}\xrightarrow{\partial}V_0$ consists of the following data:
\begin{enumerate}
    \item a pair of $\bbZ_2$-representations 
    \begin{equation}
        \rho_0 = \rho_0^1\oplus\rho_0^0:\bbZ_2\rightarrow \operatorname{End}(V_0)\oplus \operatorname{End}(V_{-1}),\nonumber
    \end{equation}
    such that $\partial$ is an intertwiner between $\rho_0^0$ and $\rho_0^1$, and 
    \item a map $\rho_{1}:\widehat{\bbZ_2}\rightarrow \operatorname{Hom}(V_0,V_{-1})$ such that $\rho_1(e)$ is the trivial chain homotopy on the identity element $e\in\widehat{\bbZ_2}$.
\end{enumerate}
Since the $t$-map for $D(B\bbZ_2)$ is trivial, $\rho$ must satisfy $\delta \rho_1 =(\rho_1\circ \partial ,\partial \circ \rho_1)=\rho_0t=0$, which means either $\rho_{1}=0$ or $\partial =0$. For 1-dimensional irreducible representations (irreps) $V_0,V_{-1}\cong k$ over the ground field $k$, the value of $\rho_1$ on the non-trivial generator $y\in\widehat{\bbZ_2}$ is either trivial or a scalar multiplication. We write simply $\rho_1=0$ in the former case, while in the latter case we shall normalize the scalar $\rho_1(y)$ to $1\in k^\times$, and denote this map by $\rho_1=\hat 1$. 

\begin{remark}\label{rho1}
Though $\rho_1$ need \textit{not} be an intertwiner, we require it to preserve the identity $\rho_0^{0,1}(1)=\rho_0^{0,1}(x^2)=\operatorname{id}$ in the sense that
\begin{equation}
    \rho_1(y)\circ\operatorname{id}_{V_0} = \operatorname{id}_{V_{-1}}\circ \rho_1(y),\qquad x\in\bbZ_2,~y\in\widehat{\bbZ_2}.\nonumber
\end{equation}
This condition is vacuous here, but it shall become non-trivial later when we {\it twist} $D(BM)$. Strictly speaking, $\rho_1$ can be trivial as well if $\partial =0$, but this distinction makes no difference for $D(B\bbZ_2)^\text{trv}$.
\end{remark}

Now given $\rho_0^0,\rho_0^1$ are irreducible, Schur's lemma implies that $\partial$ is either trivial or an isomorphism. Hence given $\partial\neq 0$, then $\rho_0^0,\rho_0^1$ are either both the trivial representation $1$, or both the sign representation $\text{sgn}$. We therefore have four inequivalent irreducible 2-representations
\begin{eqnarray}
   \textbf{Electric}:&\qquad& \mathbf{1} = (1\oplus 1,\partial=1,\rho_1=0),\qquad \mathbf{c} = (1\oplus\text{sgn},\partial=0,\rho_1=\hat 1),\nonumber\\
   \textbf{Magnetic}:&\qquad&\mathbf{1}^* = (\text{sgn}\oplus\text{sgn},\partial=1,\rho_1=0),\qquad \mathbf{c}^* = (\text{sgn}\oplus  1,\partial=0,\rho_1=\hat 1),\label{2repdd}
\end{eqnarray}
which constitute the simple objects\footnote{Given any arbitrary 2-representation $\rho\in \operatorname{2Rep}_\text{wk}(D(B\bbZ_2)$, each graded component of the vector space complex $V=V_{-1}\xrightarrow{\partial} V_0$ carries a $\bbZ_2$-representation, which decomposes individually into direct sums of irreducible representations $1,\text{sgn}$. As $\partial$ must be a $\bbZ_2$-intertwiner, it also decomposes accordingly as a direct sum on each irreducible summand, whence $V$ is a direct sum of the objects listed in \eqref{2repdd}.} in $\operatorname{2Rep}_\text{wk}(D(B\bbZ_2))$. We call the first row the {\bf electric} sector and the second row the {\bf magnetic} sector; this partition will be clear in the following. Note that $\mathbf{c}$ is {\it not} equivalent to $\mathbf{c}^*$, because the map $\partial$ remembers its domain and codomain.

\subsection{Fusion structure}\label{invtoricfusion} We now investigate the monoidal structure of the 2-category $\operatorname{2Rep}(D(B\bbZ_2)^\text{trv})$. Since the coproduct $\Delta$ on $D(B\bbZ_2)$ is grouplike, the tensor product of 2-representations $\rho,\rho'$ is just the usual \textit{graded} tensor product $\rho\otimes\rho'$. Graded here means  \eqref{grtensor}, ie. equipped with the differential $\partial$; we demonstrate this through computations below. 

Let us examine the 2-representations as listed in  \eqref{2repdd}. In the electric sector, we use the Morita equivalence $\text{sgn}^{2\otimes}\simeq 1^{2\otimes}\cong 1$ to have
\begin{equation}
    \mathbf{c}\otimes\mathbf{c} = (1\oplus\text{sgn}) \otimes (1\oplus\text{sgn}) \cong 1\oplus \text{sgn}\oplus \text{sgn}^{2\otimes}\oplus\text{sgn}\simeq \mathbf{c}\oplus\mathbf{c},\label{fus0}
\end{equation}
which tells us that $\mathbf{c}$ is a {\bf Cheshire string} \cite{Johnson-Freyd:2020}; similarly, we compute
\begin{equation}
    \mathbf{c}^*\otimes\mathbf{c}^*\cong \text{sgn}^{2\otimes}\oplus\text{sgn}\oplus 1\oplus \text{sgn} \simeq \mathbf{c}\oplus\mathbf{c}. \nonumber
\end{equation}
Note that the order of the direct sums matter, as we have are keeping track of the (trivial) differential $\partial =0$. Indeed, we have on the other hand, 
\begin{equation}
    \mathbf{c}\otimes\mathbf{c}^* = (1\oplus\text{sgn})\otimes(\text{sgn}\oplus 1) \cong \text{sgn}\oplus 1 \oplus \text{sgn}\oplus 1 \simeq \mathbf{c}^*\oplus\mathbf{c}^* \simeq \mathbf{c}^*\otimes\mathbf{c},\label{fus0*}
\end{equation}
which is distinct from the above fusion rules.

Consider the mixed fusion $\mathbf{1}^*\otimes\mathbf{c}$. Here, we need to keep track of the non-trivial maps $\partial$,
\begin{eqnarray}
\mathbf{1}^*\otimes\mathbf{c}&=&(\text{sgn}\oplus\text{sgn})\otimes (1\oplus \text{sgn})\nonumber\\
&\cong& \begin{tikzcd}
\text{sgn} \arrow[rrrr, "\partial=1", bend left] & \oplus &  \text{sgn}^{2\otimes}(\simeq 1)\arrow[rrrr, "\partial=1", bend right] & \oplus &  \text{sgn} & \oplus &\text{sgn}^{2\otimes} (\simeq 1)
\end{tikzcd}.\nonumber
\end{eqnarray}
Since these maps $\partial$ are intertwiners (in fact the identity), its domain and codomain are the same. We keep only one copy, so that
\begin{equation}
    \mathbf{1}^*\otimes \mathbf{c} \simeq \text{sgn}\oplus1=\textbf{c}^*. \label{magcheshire}
\end{equation}
Through similar computations, we have $$\mathbf{1}\otimes\mathbf{1}\cong \mathbf{1},\qquad \mathbf{1}\otimes\mathbf{c}\simeq \mathbf{c},\qquad \mathbf{1}^*\otimes\mathbf{1}^*\simeq \mathbf{1},$$ hence $\mathbf{1},\mathbf{1}^*$ are the vacuum lines; in particular, $\mathbf{1}$ is the indecomposable identity object in $\operatorname{2Rep}(D(B\bbZ_2)^\text{trv})$.

\subsubsection{2-intertwiners; the 1-morphisms}\label{2int.1-morph} Recall from {\bf Definition \ref{2repcat}} that the 1-morphisms in $\operatorname{2Rep}(D(B\bbZ_2)^\text{trv})$ are given by $\bbZ_2$-equivariant {\it cochain maps}. Form the list  \eqref{2repdd}, we clearly have identity self-2-intertwiners, such as $i[00]= \operatorname{id}: \mathbf{1}\rightarrow \mathbf{1}$ and $i[11]=\operatorname{id}:\mathbf{c}\rightarrow\mathbf{c}$.
% \begin{equation}
%     i[00]:\mathbf{1}\rightarrow\mathbf{1},\qquad i[11]:\mathbf{c}\rightarrow \mathbf{c},\nonumber
% \end{equation}
As the source and target are the same graded $\bbZ_2$-representations for self-2-intertwiners in particular, we can find two more. These are given by a swap of grading together with a certain twist,
\begin{equation}
    i'[00]:(w,v)\mapsto (v,w),\qquad i'[11]: (w,v)\mapsto (\pm 1)\cdot (v,w),\label{swaptwiners}
\end{equation}
where $(v,w)\in V_{-1}\oplus V_0$ denotes elements in $\mathbf{1}$ or $\mathbf{c}$. Clearly, the identity $i[00],i[11]$ admit trivial actions by $\rho_1$, in contrast to the grading swaps $i'[00],i'[11]$. Hence from  \eqref{tensor2rep} and the grouplike coproduct $\Delta_{-1}$  \eqref{grouplike} we deduce the following fusion rules
\begin{equation}
    i[00]\otimes i[00] = i'[00]\otimes i'[00] = i[00],\qquad i[00]\otimes i'[00] = i'[00]\otimes i[00] = i'[00], \label{1morfus}    
\end{equation}
and similarly for $i[11],i'[11]$. The same analysis applies to the dual sector $i^*[00]\in\operatorname{End}\mathbf{1}^*,i^*[11]\in \operatorname{End}\mathbf{c}^*$.

Now consider a map $i[01]:\mathbf{1}\rightarrow\mathbf{c}$; in the absence of the homotopy $I$, the commutative diagrams  \eqref{2int} respectively enforce that
\begin{equation}
    i[01]_0 \circ 1 = 0 \circ i[01]_{1},\qquad  i[01]_{1}\circ 0=\hat1 (y)\circ i[01]_0,\nonumber
\end{equation}
where $\hat 1(y)=\rho_1(y)\in\operatorname{Hom}(V_0,V_1)$ is a non-trivial scalar multiplication. These equations admit a non-trivial solution $i[01]_0 = 0,i[01]_{1}=1$, hence there is a non-trivial 2-intertwiner
\begin{equation}
    i[01]=1\oplus 0:\mathbf{1}\rightarrow\mathbf{c};\nonumber
\end{equation}
similar arguments show that we also have a non-trivial 2-intertwiner 
\begin{equation}
    i[10]=0\oplus 1:\mathbf{c}\rightarrow\mathbf{1}. \nonumber
\end{equation}
These are the {only} possible 2-intertwiners between $\mathbf{1}$ and $\mathbf{c}$. Again, the same analysis applies to the dual sector. Since $i[01]$ and $i[10]$ have different domain and codomain, we must employ the decomposition  \eqref{tensordecomp} in order to find the tensor product between them \cite{gurski2006algebraic}. However, since the coproduct $\Delta_0=0$ is trivial in $D(B\bbZ_2)^\text{trv}$, we find their tensor product
\begin{equation}
    i[01]\otimes i[10] = i[10]\otimes i[01] \simeq 1 = i[00] \nonumber
\end{equation}
to be trivial as well% (recall $D(B\bbZ_2)$ "acts" on the trivial 2-intertwiner $i[00]=\operatorname{id}$ by the trivial cochain homotopy $\rho_1=0$)
. We shall see later in Section \ref{fermifusion} that this will be different once we introduce twists on $D(B\bbZ_2)$.

Let us now come finally to the 2-intertwiners that map between dual sectors. First, consider maps such as $\mathbf{1}\rightarrow\mathbf{1}^*$ or $\mathbf{c}\rightarrow\mathbf{c}^*$. Any such maps must intertwine between different $\bbZ_2$-representations in both degrees, and the only such map is $0$. Next, consider a map $\bar i[01]:\mathbf{1}\rightarrow \mathbf{c}^*$. The commutative diagrams  \eqref{2int} enforce
\begin{equation}
    \bar i[01]_0 \circ 0 = 1 \circ \bar i[01]_{1},\qquad  \bar i[01]_{1}\circ 0=\hat1 (1)\circ \bar i[01]_0.\nonumber
\end{equation}
The first equation says $\bar i[01]_{1} = 0$, while the second equation says $\bar i[01]_0 =0$, hence $\bar i[01]=0$. Similarly, any 2-intertwiner $\bar i[10]:\mathbf{c}^*\rightarrow \mathbf{1}$ must be trivial $\bar i[10]=0$. 

The above paragraph proves that $\operatorname{2Rep}(D(B\bbZ_2)^\text{trv})$ has two connected components made separately of the electric and magnetic objects in  \eqref{2repdd}, which have no (invertible) 1-morphisms between them. We denote the \textbf{identity component} of $  \operatorname{2Rep}(D(B\bbZ_2)^\text{trv})$, namely the connected component of the fusion identity $\mathbf{1}$, by $\Gamma$, which consist of nothing but the electric sector. Relabeling $i[00],i[11]=\mathfrak{1}$ and $i'[00],i'[11]=\mathfrak{e}$, we achieve the following structure for $\Gamma$ from  \eqref{1morfus},
\begin{equation}
    \begin{tikzcd}
\mathbf{1} \arrow[rr, "{i[01]}", bend left] \arrow["{\mathfrak{1},\mathfrak{e}}"', loop, distance=2em, in=215, out=145] &  & \mathbf{c} \arrow[ll, "{i[10]}", bend left] \arrow["{\mathfrak{1},\mathfrak{e}}"', loop, distance=2em, in=35, out=325]
\end{tikzcd},\label{gammadiag}
\end{equation}
which shall become crucial in the following.

\subsubsection{Cochain homotopies; the 2-morphisms} Recall from {\bf Definition \ref{2repcat}} that the 2-morphisms in $\operatorname{2Rep}_\text{wk}(B\bbZ_2)$ are given by $\widehat{\bbZ_2}$-equivariant cochain homotopies. Of course, the monoidal structure of the 1-morphisms (eg.  \eqref{1morfus}) induce a monoidal structure on the modifications $\mu\otimes\mu':i\otimes j\Rightarrow i'\otimes j'$, which by using the (so-far trivial) interchanger \eqref{interchange} can be expressed in terms of the composition $(\mu\otimes \id_j)\circ (\id_i\otimes\mu') \simeq \mu\circ \mu'$.

By inspection of the connected component $\Gamma$ \eqref{gammadiag}, one can argue that the only modifications possible in $\Gamma$ are self-modifications $\mu:i\Rightarrow i$. To see this, we first note that there is only one unique 1-morphism $i[01]$ (or $i[10]$) between the simple objects $\mathbf{1}$ and $\mathbf{c}$, hence we only have the trivial identity cochain homotopy $\id: i[01]\Rightarrow i[01]$. On the other hand, there are two 1-endomorphisms on $\mathbf{1}$ (or equivalently $\mathbf{c}$), denoted by $\mathfrak{1},\mathfrak{e}$. Each of these of course comes with its own trivial identity cochain homotopy, denoted by
\begin{equation}
    1:\mathfrak{1}\Rightarrow\mathfrak{1}\qquad \mu :\mathfrak{e}\Rightarrow\mathfrak{e}.\label{chainhomotopy}
\end{equation}
Here, we note that $\mu\simeq -1\cdot\id$ carries a global sign due to a grading swap in $\mathfrak{e}$ \eqref{swaptwiners}.

% \medskip

It then remains to check that there are no non-trivial cochain homotopies between $\mathfrak{1}$ and $\mathfrak{e}$. Let $\bar\mu: \mathfrak{1}\Rightarrow\mathfrak{e}$ denote such a cochain homotopy. In order for $\bar\mu$ to denote a genuine 2-morphism in $\operatorname{2Rep}_\text{wk}(D(B\bbZ_2))$, it must by definition \eqref{modif} intertwine $\rho_1$. However, $\widehat{\bbZ_2}$ "acts" trivially on $\mathfrak{1}$, while non-trivially on $\mathfrak{e}$,
$$\rho_1=\id: \mathfrak{1}\Rightarrow\mathfrak{1},\qquad \rho_1=\text{sgn}\cdot\id:\mathfrak{e}\Rightarrow\mathfrak{e},$$ and hence $\bar\mu=0$ must be trivial. This demonstartes that the only modofications in $\operatorname{2Rep}_\text{wk}(D(B\bbZ_2))$ that exist are the self-modifications $\mu:i\Rightarrow i$, as desired.

\medskip

We of course have the trivial 1- and 2-morphisms given by $0$. More importantly, we note that the non-trivial 1- and 2-morphisms that we have identified above are not unique. In particular, we have made the choice to normalize all of the 2-intertwiners and the self-modifications, whereas any scalar multiple of them would also be valid. Further, we are also able to take direct sums of the 2-intertwiners that we have identified above; basically, Section \ref{2int.1-morph} lists a minimal set of generators for the Hom-spaces of $\operatorname{2Rep}_\text{wk}(D(B\bbZ_2))$.

\paragraph{Loop 1-category: the 1-endomorphism space of the tensor unit.} Recall that, for any Abelian group $N$, $\rho_1:\widehat{N}\rightarrow \operatorname{End}(V)_{-1}$ defines a "$\widehat{N}$-action" by cochain homotopies on the endomorphisms $\operatorname{End}(V)$ of a 2-representation $V\in\operatorname{2Rep}_\text{wk}(D(BN))$. Furthermore, modifications $\mu:i\Rightarrow i'$ between $i,i'\in\operatorname{End}(V)$ by definition \eqref{modif} must necessarily intertwine this $\widehat{N}$-action. On the tensor unit $V=\mathbf{1}$, in particular, the space $\operatorname{End}(\mathbf{1})\cong k\xrightarrow{1}k$ furnishes a 1-dimensional irreducible $\widehat{N}$-module, for which the intertwining modifications between the different $\widehat{N}$-module structures are either the identity or trivial. This allows us to conclude that
\begin{equation}
    \Omega\operatorname{2Rep}_\text{wk}(D(BN)) = \operatorname{End}_{\operatorname{2Rep}_\text{wk}(D(BN))}(\mathbf{1}) \simeq \operatorname{Rep}(\widehat{N}).\label{looping}
\end{equation}  
For $N=\bbZ_2$, we recover the result that $\operatorname{End}(\mathbf{1}) \simeq \operatorname{Rep}(\bbZ_2)\simeq \operatorname{Vect}[\bbZ_2]$ has two distinct objects, $\mathfrak{1},\mathfrak{e}$, with no non-trivial modifications between them.

Notice that in the usual theory of higher representations \cite{Wen:2019,KongTianZhou:2020,Johnson_Freyd_2023,decoppet2023drinfeld}, where $\text{2Vect}^{KV}$ is 2-enriched in $\text{Vect}$ \cite{Kapranov:1994,Baez1996HigherDimensionalAI}, the above statement follows immediately from definition. However, in the context of $\text{2Vect}^{hBC}$ (which is {\it not} 2-enriched), we have to prove it by direct computation.

\begin{proposition}\label{4dtoric}
There is a $\mathrm{non}$-$\mathrm{monoidal}$ equivalence between $Z_1(\Sigma\operatorname{Vect}[\bbZ_2])$ and $\operatorname{2Rep}_\text{wk}(D(B\bbZ_2)^\mathrm{trv})$.
\end{proposition}
\begin{proof}
We use the description of the braided fusion 2-category $\mathscr{R} \simeq Z_1(\Sigma\operatorname{Vect}[\bbZ_2])$ (with trivial associator class) describing the (3+1)D toric code given in \cite{Johnson-Freyd:2020}. This category has two identical components; the identity component $\Sigma\operatorname{Vect}[\bbZ_2]$ has two simple objects, given by the trivial $\bbZ_2$-algebra $I=\bbC$ and the Cheshire string $c=\bbC[x]/\langle x^2-1\rangle$, where $\bbZ_2$ acts non-trivially on $x$. Monoidally, the two components of $\mathscr{R}$ follow a fusion rule that is graded by $\bbZ_2$ \cite{KongTianZhou:2020},
\begin{equation}
    I^2\simeq m^2 \simeq I,\qquad c^2 \simeq m'^2\simeq c\oplus c,\qquad c\otimes m = m\otimes c \simeq m',\qquad c\otimes m'\simeq m'\otimes c\simeq m'\oplus m',\nonumber
\end{equation}
where $m,m'$ denotes the simple objects in the non-trivially graded copy of $\Sigma\operatorname{Vect}[\bbZ_2]$.

To show the desired equivalence, we need to exhibit a 2-functor $\mathfrak{F}:\mathscr{R}\rightarrow \operatorname{2Rep}_\text{wk}(D(B\bbZ_2))$ which is essentially surjective and fully faithful. This means that 
\begin{enumerate}
    \item $\mathfrak{F}$ is {\it essentially surjective}, namely a surjection on the equivalence classes of objects, and
    \item $\mathfrak{F}$ is {\it fully faithful}, namely it is an equivalence of Hom-categories.
\end{enumerate}
We begin by taking 
\begin{equation}
    \mathfrak{F}(I) = \mathbf{1},\qquad  \mathfrak{F}(c) = \mathbf{c},\qquad \mathfrak{F}(m) =\mathbf{1}^*,\qquad \mathfrak{F}(m') = \mathbf{c}^*,\nonumber
\end{equation}
which is a bijection on the simple objects. Hence $\mathfrak{F}$ is essentially surjective, and furthermore preserves the identity. Now to check that $\mathscr{F}$ is fully-faithful, we must consider the Hom-categories. Since $Z_1(\Sigma\operatorname{Vect}[\bbZ_2])\simeq \Sigma\operatorname{2Vect}[\bbZ_2] \oplus \Sigma\operatorname{2Vect}[\bbZ_2]$ \cite{KongTianZhou:2020}, it suffices to show full and faithfulness on the corresponding identity components $\mathfrak{F}: \Sigma\operatorname{Vect}[\bbZ_2]\rightarrow\Gamma$. 

To begin, we note that $\Sigma\operatorname{Vect}[\bbZ_2]$ is well-known to have the following form \cite{Douglas:2018},
\begin{equation}
    \begin{tikzcd}
I\arrow[rr, "{\operatorname{Vect}}", bend left] \arrow["{\operatorname{Vect}[\bbZ_2]}"', loop, distance=2em, in=215, out=145] &  & c \arrow[ll, "{\operatorname{Vect}}", bend left] \arrow["{\operatorname{Vect}[\bbZ_2]}"', loop, distance=2em, in=35, out=325]
\end{tikzcd},\nonumber
\end{equation}
with each of the Hom-categories labeled. We let $v_1\cong k$ (resp. $v_2\cong k$) denote respectively the simple object in the linear Hom-category $\operatorname{Vect}= \operatorname{Hom}_{\Sigma\operatorname{Vect}[\bbZ_2]}(I,c)$ (resp. $\operatorname{Vect}= \operatorname{Hom}_{\Sigma\operatorname{Vect}[\bbZ_2]}(I,c)$), which can be understood as the 1-dimensional vector space over $k$. Similarly, we let $\mathring{1},e$ denote the two simple objects of the linear Hom-category $\operatorname{Vect}[\bbZ_2] =  \operatorname{End}_{\Sigma\operatorname{Vect}[\bbZ_2]}(I)\simeq  \operatorname{End}_{\Sigma\operatorname{Vect}[\bbZ_2]}(c)$; the direct sum $\mathring{1}\oplus e$ corresponds to a 1-dimensional super- (ie. $\bbZ_2$-graded) vector space. 

By comparing with \eqref{gammadiag}, we define the following component functors of the 2-functor $\mathfrak{F}$ by
\begin{gather*}
    \mathfrak{1} = \begin{cases} i[00] = (\mathfrak{F})_{I\to I}(\mathring{1}) \\ i[11] = (\mathfrak{F})_{c\to c}(\mathring{1})\end{cases},\qquad \mathfrak{e} = \begin{cases}i'[00] = (\mathfrak{F})_{I\to I}(e) \\ i'[11] = (\mathfrak{F})_{c\to c}(e)\end{cases}\\
    i[01]= (\mathfrak{F})_{I\to c}(v_1),\qquad i[10] = (\mathfrak{F})_{c\to I}(v_2),
\end{gather*}
which we note are all unit-preserving and essentially surjective. It then suffices to check that these component functors are fully faithful. By leveraging the linearity of the Hom-categories under consideration, this is equivalent to checking that each of the component functors send (additive) generating 2-morphisms to generating 2-morphisms.

This is indeed the case. Let $j_1\in\operatorname{End}_{\operatorname{Vect}}(v_1)\cong k$ denote the non-trivial generating 2-morphism over the indecomposable 1-morphism $v_1\in \operatorname{Vect}= \operatorname{Hom}_{\Sigma\operatorname{Vect}[\bbZ_2]}(I,c)$. Then by construction $\mathfrak{F}_{I\to c}(j_1:v_1\Rightarrow v_1) = \operatorname{id}_{i[01]}:i[01]\Rightarrow i[01]$ is the identity self-modification on $i[01]=\mathfrak{F}_{Ic}(v_1)$, which is the generating object in the Hom-category $\operatorname{Hom}_{\operatorname{2Rep}_\text{wk}(D(B\bbZ_2))}(I,c)$ as required. 

Similarly, as $\mathfrak{F}_{I\to I},\mathfrak{F}_{c\to c}$ are additive, they send the generating 2-morphism $j_{\bbZ_2}: \mathring{1}\oplus e\Rightarrow \mathring{1}\oplus e$ over the indecomposable $\mathring{1}\oplus e\in \operatorname{Vect}[\bbZ_2]$ to the (graded) generating chain homotopy $1+\mu$ \eqref{chainhomotopy} over $\mathfrak{1}\oplus \mathfrak{e} = \mathfrak{F}_{I\to I}(\mathring{1}\oplus e)$. This shows that each component functor $\mathfrak{F}_{X\to Y}$ are equivalences of the corresponding Hom-categories, and hence $\mathfrak{F}:\Sigma\operatorname{Vect}[\bbZ_2]\rightarrow \Gamma$ is an equivalence of 2-categories. 

% For each $X,Y\in \Sigma\operatorname{Vect}[\bbZ_2]$, we then have
% \begin{equation*}
%     (\mathfrak{F})_{XY}(\operatorname{Hom}_{\operatorname{Hom}_{\Sigma\operatorname{Vect}[\bbZ_2]}(X,Y)}) \cong \operatorname{Hom}_{\operatorname{Hom}_{\operatorname{2Rep}_\text{wk}(D(B\bbZ_2))}(\mathfrak{F}(X),\mathfrak{F}(Y))}.
% \end{equation*}
\medskip

We now wish to lift $\mathfrak{F}$ to a {\it monoidal} 2-functor, which requires the fusion rules to be preserved (up to coherence). The computations \eqref{fus0}, \eqref{fus0*}, \eqref{magcheshire} show that $\mathfrak{F}:Z_1(\Sigma\operatorname{Vect}[\bbZ_2])\rightarrow \operatorname{2Rep}_\text{wk}(D(B\bbZ_2)^\text{trv})$ preserves the fusion rules of the simples, and is indeed monoidally essentially surjective. Next is to check that each component functor $\mathfrak{F}_{X\to Y},~X,Y\in \Sigma\operatorname{Vect}[\bbZ_2]$ is monoidal on the Hom-categories.

From the fusion rules \eqref{1morfus} for the 1-morphisms, we see that $\mathfrak{F}_{X\to X}$ with $X=I,c$ are indeed monoidal, but the issue is that $\mathfrak{F}_{I\to c}$ (or $\mathfrak{F}_{c\to I}$) is {\it not}: $\mathfrak{F}_{I\to c}(v_1)\otimes\mathfrak{F}_{I\to c}(v_1) = i[01]\otimes i[10] \simeq \mathfrak{1}$ is trivial in $\Gamma\subset\operatorname{2Rep}_\text{wk}(D(B\bbZ_2)^\text{trv})$, while $v_1\otimes v_1\not\simeq \mathring{1}$ is not in $\Sigma\operatorname{Vect}[\bbZ_2]$. This prevents $\mathfrak{F}$ from being a monoidal equivalence.
\end{proof}

The problem is in fact even worse --- we will show in the following that $\operatorname{2Rep}_\text{wk}(D(B\bbZ_2)^\text{trv})$ does not even define a gapped topological order. We elaborate in Section \ref{ferkitcharges} on how this problem can be amended by {\it twisting} the 2-algebra structure of $D(B\bbZ_2)$.

\subsection{The braiding data}
Let us for now turn to the {\it braiding} structure. From the perspective of $\mathscr{R}$, it is understood \cite{Johnson-Freyd:2020} in particular that there is the self-braiding
\begin{equation}
    \beta: m\otimes m\rightarrow m\otimes m\nonumber
\end{equation}
on the magnetic $m$ line, which can either be trivial $\mathring{1}$ or the electric $\bbZ_2$-particle $e$. An argument was given in \cite{Johnson-Freyd:2020} that states $\beta=\mathring{1}$ is in fact trivial. We will prove that this is also the case in $\operatorname{2Rep}_\text{wk}(D(B\bbZ_2)^\text{trv})$, but there is a major problem.

\begin{theorem}\label{braidtoric}
All braiding maps on $\operatorname{2Rep}(D(B\bbZ_2)^\text{trv})$ are trivial.
\end{theorem}
\begin{proof}
Recall from  \eqref{2repbraid1}, \eqref{2repbraid2} that the braiding structure of $\operatorname{2Rep}(D(B\bbZ_2)^\text{trv})$ is induced by a 2-$R$-matrix $(\cR,R)$ on $D(B\bbZ_2)^\text{trv}$. However, instead of explicitly solving the 2-Yang-Baxter equations, we are instead going to {\it induce} the 2-$R$-matrix from the structure of the 2-quantum "quadruple" $D(D(B\bbZ_2),D(B\bbZ_2))$.

This method is based on the general quantum double construction of Majid \cite{Majid:1994nw,Majid:1996kd}, in which the universal $R$-matrix $R$ on $D(\bbZ_2,\bbZ_2) = k\bbZ_2 \bar{\bowtie} k\bbZ_2$ can be reconstructed
\begin{equation}
    R =  \bar\Psi\circ\operatorname{coev},\label{braidrmatrix}
\end{equation}
from the {\it braided transposition} $\bar\Psi:k\bbZ_2\otimes k\bbZ_2\rightarrow k\bbZ_2\otimes k\bbZ_2$ of the underlying quantum double (note here $k\bbZ_2 = D(B\bbZ_2)_0$ is in degree-0), satisfying \cite{Majid:1994nw,Majid:1996kd}
\begin{equation}
    xx' = \cdot\circ \bar\Psi(x'\otimes x),\qquad x,x'\in\bbZ_2.\nonumber
\end{equation}
Here, $\operatorname{coev}$ is the coevaluation dual to the canonical pairing form on $\bbZ_2$. Now since $\bbZ_2$ is Abelian, $\bar \Psi$ is simply the identity and hence  \eqref{braidrmatrix} states that $R = \operatorname{id}$ is in fact the identity matrix. The braiding maps $b_{V,W} =1$ are thus all trivial.

\smallskip

Now consider the universal 2-$R$-matric $\cR$. The above result was categorified in \cite{Chen:2023tjf}, hence we can play the same game and reconstruct
\begin{equation}
    \cR^+ = \Psi^l_{-1}\circ\operatorname{coev}_l,\qquad \cR^- = \Psi_{-1}^r\circ\operatorname{coev}_r\label{braid2rmatrix}
\end{equation}
from the underlying braided transpositions 
\begin{gather}
    \Psi_{-1}^l: \widehat{\bbZ_2}\otimes\bbZ_2\rightarrow \bbZ_2\otimes\widehat{\bbZ_2},\qquad \Psi_{-1}^r: \bbZ_2\otimes\widehat{\bbZ_2}\rightarrow \widehat{\bbZ_2}\otimes\bbZ_2 \nonumber \\
    y\cdot f = \cdot \Psi^l_{-1}(y\otimes f),\qquad x\cdot g = \cdot \circ \Psi^r_{-1}(x\otimes g)\label{psimap}
\end{gather}
on the Drinfel'd double 2-bialgebra $D(D(B\bbZ_2),D(B\bbZ_2))$, where $x,f\in\bbZ_2$ and $y,g\in\widehat{\bbZ_2}$ and $\operatorname{coev}_{l,r}$ is the coevaluation dual to the Pontrjagyn pairing. However, in the case of $D(B\bbZ_2)^\text{trv}$, the braided transposition is merely the Pontrjagyn duality,
\begin{equation}
    \Psi_{-1}^l(y\otimes f) = \hat y\otimes \hat f,\qquad \Psi_{-1}^r(x\otimes g) = \hat x\otimes \hat g,\nonumber
\end{equation}
whence  \eqref{braid2rmatrix} states that $\cR^\pm = p\circ \operatorname{coev} = \operatorname{id}$ is the identity matrix. The mixed braiding maps $b_{i,W},b_{W,i}$ are thus all trivial.

\end{proof}
\noindent The fact that all the braiding maps are trivial on $\operatorname{2Rep}_\text{wk}(D(B\bbZ_2)^\text{trv})$ can also be seen from the corresponding topological NLSM $Z_\text{Kit}^0$, which has no terms in its action that encode any non-trivial statistics of the charges in the theory \cite{Zhu:2019,Wen:2019}. Of course, we already know from {\bf Proposition \ref{4dtoric}} that $\operatorname{2Rep}_\text{wk}(D(B\bbZ_2)^\text{trv})$ not (braided) monoidally equivalent to the toric code $\mathscr{R}$, and hence calling $Z_\text{Kit}^0$ the "4d toric code" is incorrect.

\begin{remark}\label{invisibletoric}
$\operatorname{2Rep}_\text{wk}(D(B\bbZ_2)^\text{trv})$ is "too trivial" to even describe a gapped topological phase, since it violates the principle of \textbf{remote detectability} \cite{Kong:2020,KongTianZhou:2020,Johnson-Freyd:2020}. This principle states that all non-trivial excitations can be detected by braiding, and it is part of the definition of a topological order (such as the toric code $\mathscr{R}\simeq \Sigma\operatorname{Vect}[\bbZ_2]$). In this simple $\bbZ_2$-charged case, this principle is encoded by the presence of the term $\langle B\cup \bar e(A)\rangle$ in the Dijkgraaf-Witten 4-cocycle $\omega$ \cite{Zhu:2019,Johnson-Freyd:2020}, which is only present for $Z_\text{Kit}^s$. Nevertheless, the above computations lay the foundation for our results in the following.
\end{remark}

\begin{remark}\label{wk2algrepthy}
    It is worthwhile to pause here and comment on the fact that $\operatorname{2Rep}_\text{wk}(D(B\bbZ_2)^\text{trv})$ is a symmetric 2-category, similar to the (Kapranov-Voevodsky) 2-representations $\operatorname{2Rep}_{D(B\bbZ_2)}$ of (the 2-group underlying) $D(B\bbZ_2)$. More generally, we have shown in the appendix of \cite{Chen:2023tjf} that, for any split skeletal 2-group $\mathbb{G}$, the weak 2-representation 2-category $\operatorname{2Rep}_\text{wk}(k\mathbb{G})$ of its associated 2-group algebra $k\mathbb{G}$ (see Section \ref{2grpbialg}) hosts the same identical set of coherence data/relations as the 2-category of KV 2-representations $\operatorname{2Rep}_{\mathbb{G}}$ as studied in the literature \cite{Baez:2012,Douglas:2018,Delcamp:2023kew}. Since the grouplike coproduct \eqref{grouplike} on $k\mathbb{G}$ is cocoummutative, the 2-$R$-matrix is trivial and hence no non-trivial braiding is present.
\end{remark}

\section{Excitations in the spin-Kitaev model $Z_\text{Kit}^s$}\label{ferkitcharges}
We now turn to the spin-Kitaev model $Z_\text{Kit}^s$ given by the Drinfel'd double 2-bialgebra $D(B\bbZ_2)^\text{sgn}$. Its 2-representations have the same ingredients as those of $D(B\bbZ_2)^\text{trv}$, and hence the 2-category $\operatorname{2Rep}(D(B\bbZ_2)^\text{sgn})$ also has four objects, similar to those in  \eqref{2repdd}.

The difference here is that $D(B\bbZ_2)^\text{sgn}_0 =\bbZ_2$ now acts non-trivially on $D(B\bbZ_2)^\text{sgn}_{-1}=\widehat{\bbZ_2}$. This action was obtained by dualizing the non-trivial action $u\in \operatorname{Aut}(k\bbZ_2)$, which induces via  \eqref{ext} the class $\bar e=\frac{1}{2}u^2\in H^2(\bbZ_2,\bbZ_2)$ determining the non-trivial central extension of $\bbZ_2$ by itself. This extension is $\bbZ_4$, which we interpret as a  "semidirect product" $\bbZ_2\rtimes \bbZ_2$ where the central element $x^2\in\bbZ_2$ acts by $-1$.

As such, the component $\rho_0^0(x)^2$ "acts" non-trivially on the degree-(-1) component of the graded 2-representation spaces. In other words, provided $\rho_0^0$ is non-trivial, the component $\rho_0$ of the 2-representation $\rho$ furnishes a representation of $\bbZ_2\rtimes\bbZ_2$, satisfying
\begin{equation}
    \rho_0(x^2)(w,v) = (\bar e(x,x)\cdot(\rho_0^1(x^2))w,\rho_0^0(x^2)v) = (-w,v),\qquad x\in\bbZ_2,\nonumber
\end{equation}
where $(w,v) \in V\cong V_{-1}\oplus V_0$. We denote such representations by $\rho_0=\rho_0^1\oplus_\pm \rho_0^0 = (\bar e\cdot \rho_0^1,\rho_0^0)$. From  \eqref{2repdd}, we thus see that the magnetic vacuum line $\mathbf{1}^*$ and the Cheshire string $\mathbf{c}$ carry a $\bbZ_4$-representation, while the electric vacuum line $\mathbf{1}$ and the magnetic Cheshire $\mathbf{c}^*$ carry a $\bbZ_2\times\bbZ_2$-representation.

\medskip

Now recall from {\it Remark \ref{rho1}} that $\rho$ should preserve the identity, which was a vacuous condition as $\rho_0^{0,1}(x^2)=1$ are both trivial for $D(B\bbZ_2)^\text{trv}$. However, due to the non-trivial sign coming from $\rho_1(\bar e(x,x))=-1$ in the current case, this becomes a non-trivial relation that one must impose,
\begin{equation}
    -1\cdot \rho_1(y)=\rho_1( \bar e(x,x)\cdot y)\rho_0^1(x^2) = \rho_0^0(x^2)\rho_1(y)=\rho_1(y),\qquad y\in\widehat{\bbZ_2}.\nonumber
\end{equation}
The component $\rho_1$ is thus no longer required in general to preserve the identity. As $V_0,V_{-1}\cong k$ are both 1-dimensional vector spaces over the ground field $k$, we have 
\begin{equation}
    \rho_1(y^{-1})\rho_1(y)=\rho_1(y)^2 = \rho_1(y)^2(\rho_1(y^2))^{-1}\equiv \bar c(y,y) = -1\label{ferm2cocy}
\end{equation}
by considering $\rho_1(y)\in k^\times$ as an invertible element. This defines a 2-cocycle $\bar c\in H^2(\widehat{\bbZ_2},k^\times)$ at degree-(-1) carried by 2-representations that have $\rho_1\neq 0$. In other words, the Cheshire strings $\mathbf{c},\mathbf{c}^*$ are capable of carrying a minus sign due to $\bar c$, while the vacuum lines $\mathbf{1},\mathbf{1}^*$ do not.

By considering $D(B\bbZ_2)$ as a categorical group, $\bar c$ determines an interchanger 2-morphism $h(y_1,y_1';y_2,y_2') = \bar c(y_1,y_2)$ implementing  \eqref{grpinterchange} \cite{Wen:2019}; see {\it Remark \ref{gencatgrp}}. Correspondingly, we thus have two versions of the 2-category $\operatorname{2Rep}_{f,m}^\tau(D(B\bbZ)^\text{sgn})$, corresponding to the versions of $D(B\bbZ_2)$ that either carry the projective sign $\bar c$ or do not. 

\paragraph{Twisted Drinfel'd double 2-bialgebras.} These 2-cocycles $\bar c,\bar e$ can alternatively be interpreted as "twists" in the 2-algebra structure of the Drinfel'd double 2-bialgebra. Moreover, they can also be interpreted as contributions to the 4-cocycles $H^4(D(B\bbZ_2),k^\times)$ of the (2-group underlying the) Drinfel'd double 2-bialgebra $D(B\bbZ_2)$. This is a categorification of the 3-cocycle twists of an ordinary 1-Drinfel'd double/3d tube algebra \cite{Willerton:2008gyk}; indeed, twists of 2-group(oid) algebras by 4-cocycles have also appeared in the construction of the 4d tube algebra \cite{Bullivant:2021}. 

More precisely, the degree-4 cohomology of $D(B\bbZ_2)$ was computed in \cite{Kapustin2017} to take the form 
\begin{equation*}
    H^4(D(B\bbZ_2),k^\times) \cong H^4(B^2\widehat{\bbZ_2},k^\times)\oplus H^2(B\bbZ_2,\widehat{\bbZ_2})\oplus H^4(\bbZ_2,k^\times).
\end{equation*}
The 2-cocycle $\bar e$ fits into the second term, while the double suspension map $\widehat{\bbZ_2}\to B^2\widehat{\bbZ_2}$ sends $\bar c\mapsto \bar c[1]$ into the first term \cite{Wen:2019,Johnson_Freyd_2023,Zhu:2019}. This allows us to identify two different 2-group 4-cocycles 
\begin{equation}
    \omega_f = \bar c[1] + \bar e, \quad \omega_b = \bar e \quad \in H^4(\mathbb{G},k^\times)\label{4cocy}
\end{equation}
corresponding to twists of the Drinfel'd double 2-bialgebra $D(B\bbZ_2)$, where the notation "$[1]$" signifies a degree-shift under the double suspension map. These are the 4-cocycles that had appeared in {\bf Theorem \ref{mainthm}}.

In analogy with the 3-dimensional case \cite{Willerton:2008gyk}, we shall denote the twisted Drinfel'd double 2-bialgebras by $D^\omega(B\bbZ_2)$, where $\omega\in H^4(D(B\bbZ_2),k^\times)$. We take, now with proper naming,
\begin{gather}
    \textbf{Spin-Kitaev: } \quad \operatorname{2Rep}_f^\tau(D(B\bbZ_2)^\text{sgn}) = \operatorname{2Rep}_\text{wk}(D^{\omega_f}(B\bbZ_2)),\nonumber\\
    \textbf{Toric code: }\quad \operatorname{2Rep}_m^\tau(D(B\bbZ_2)^\text{sgn}) = \operatorname{2Rep}_\text{wk}(D ^{\omega_b}(B\bbZ_2)),\nonumber
\end{gather}
in which the first version is called {\bf fermionic} ($f$-subscript) while the second version is {\bf bosonic} ($m$-subscript). This notation is suggestive, as it corresponds to whether the degree-(-1) $\widehat{\bbZ_2}$ of the Dijkgraaf-Witten NLSM associated to $D(B\bbZ_2)^\text{sgn}$ is fermion parity $\bbZ_2^f$ or a bosonic $\pi$-flux $\bbZ_2^m$ \cite{Zhu:2019,Wen:2019}. 

\smallskip

Strictly speaking, the monster 2-BF theory  \eqref{toric} associated to $\operatorname{2Rep}_f(D(B\bbZ_2)^\text{sgn})$ should include a term $\bar c(B,B)$ given by the data of the 2-cocycle $\bar c$, whence the partition function  \eqref{kit} reads
\begin{equation}
    Z_\text{Kit}^s(X) \sim \sum_{\substack{dA=0 \\ dB = \tau}}e^{i2\pi \int_X\langle B\cup \bar e(A)\rangle + \bar c(B,B)}.\label{fermkit}
\end{equation}
Note that this term $\bar c(B,B)$, being cohomological, does not alter the EOM\footnote{Indeed, a 2-gauge theory with $F=B$ as an equation of motion would host instead a trivial 2-group $\bbZ\xrightarrow{1}\bbZ_2$ \cite{chen:2022}.} for the fields $(A,B)$. The theory $Z_{\text{Kit}}^s$  has also appeared as part of the NLSM construction in \cite{Zhu:2019}, provided we identify
\begin{equation}
    \bar e(A) = \frac{1}{2}\operatorname{Sq}^1A,\qquad \bar c(B,B) = \frac{1}{2}\operatorname{Sq}^2B\label{Z2stablecohomo}
\end{equation}
in terms of the $\bbZ_2$-cohomology operation $\operatorname{Sq}^i:H^j(X,\bbZ_2)\rightarrow H^{j+i}(X,\bbZ_2)$ called the \textbf{Steenrod square} \cite{book-charclass}.

\begin{remark}\label{particlestatistics}
In the spin-Kitaev model $Z_\text{Kit}^s$, the coefficient of $1/2$ that appeared in front of the term $\operatorname{Sq}^2B$ means that the point-like particle in the NLSM is a fermion \cite{Zhu:2019}. If this coefficient is $1/4$, then such a term $\frac{1}{4}\operatorname{Sq}^2B = \mathfrak{p}_2(B)$ gives a cohomology operation called the {\it Pontrjagyn square} $\mathfrak{p}_2:H^2(X,\bbZ_2)\rightarrow H^4(X,\bbZ_4)$ \cite{Kapustin2017}. The point particle would then be a {\it semion} \cite{Zhu:2019} in this case.
\end{remark}

% In the following, we shall denote the two versions of the 2-categories  $\operatorname{2Rep}_m(D(B\bbZ_2)^\text{sgn})$ and $\operatorname{2Rep}_f(D(B\bbZ_2)^\text{sgn})$ collectively by $\operatorname{2Rep}(D(B\bbZ_2)^\text{sgn})$, without the subscript.

% As we aim to reconstruct the topological phases defined in \cite{Johnson-Freyd:2020}, we shall mainly focus on the fermionic version $\operatorname{2Rep}_f(D(B\bbZ_2)^\text{sgn})$ in the following. With this understood, we will drop the subscript $f$ from hereon.

\subsection{Fusion structure in the twisted case}\label{fermifusion} Due to the presence of 2-cocycles $\bar e$ and $\bar c$ in $\operatorname{2Rep}_f(D(B\bbZ_2)^\text{sgn})$, the corresponding coproduct component $\Delta_0'$ governing the tensor product of 2-representations via  \eqref{tensor2rep} now satisfies a modified version of the condition  \eqref{2algcoprod},
\begin{equation}
    \Delta_0'(x^2) =(\bar e(x,x)\cdot \bar e(x,x))\otimes x^2 = \bar c(y,y)1\otimes 1,\label{fermcoprod}
\end{equation}
where we have noted $y=\bar e(x,x)$ and the twisted monoidal structure $y\cdot y = \bar c(y,y) \cdot 1$ for generators $x\in\bbZ_2,y\in\widehat{\bbZ_2}$. The presence of the sign $\bar c(y,y)=-1$ allows us to lift or trivialize certain $\bbZ_4$-representations. We demonstrate this with explicit computations.

Forming the tensor product, we see that the fusion rules in $\operatorname{2Rep}_f(D(B\bbZ_2)^\text{sgn})$ must be different than that in $\operatorname{2Rep}(D(B\bbZ_2)^\text{trv})$. To see this more explicitly, we perform a monoidal computation while keeping track of the data $\rho_1:V_0=\text{sgn}\rightarrow V_{-1}=1$,
\begin{eqnarray}
    \mathbf{c}\otimes\mathbf{c}&=& (\rho_\mathbf{c}\otimes \rho_{\mathbf{c}}) \circ\Delta_0' \nonumber\\
    &=&(\bar e\cdot 1\otimes \bar e \cdot 1\xleftarrow{\rho_1} \text{sgn}\otimes \text{sgn} (\simeq 1)) \oplus (\bar e\cdot 1 \otimes\text{sgn}~\substack{\xleftarrow{\rho_1\otimes 1} \\ \xrightarrow[1\otimes \rho_1]{}}~ \text{sgn}\otimes \bar e\cdot 1) \nonumber\\
    &\simeq& (1 \xleftarrow{\hat 1} 1) \oplus (\bar e\cdot 1\otimes \text{sgn}~\substack{\xleftarrow{\rho_1\otimes 1} \\ \xrightarrow[1\otimes \rho_1]{}} ~\text{sgn}\otimes \bar e\cdot 1),\nonumber
\end{eqnarray}
where we we have used the fact that $(\bar e\cdot 1)^{2\otimes} \simeq 1$ and $\rho_1\otimes \rho_1 \simeq \hat 1$. 

The first term is simply the trivial representation $1$, while we use $\rho_1(y)^2= \bar c(y,y)=-1$ in the second term to lift "$\text{sgn}$" to a sign representation of the subgroup $\bbZ_2\subset \bbZ_4$. However, together with the factor $\bar e(x,x)\neq 1$, this allows to degenerate $\bar e\cdot 1\otimes \text{sgn}\simeq 1$ to the trivial representation; this is the effect of the condition  \eqref{fermcoprod}. As such, we have
\begin{equation}
    \mathbf{c}\otimes\mathbf{c} \simeq (1\xleftarrow{\hat 1} 1) \oplus (1\xleftarrow{\hat 1}1)\simeq 1\oplus 1 = \mathbf{1},\label{fus1}
\end{equation}
which is indeed distinct from  \eqref{fus0}. The magnetic Cheshire $\mathbf{c}^*$, on the other hand, does not carry $\bar e$, so it furnishes a $k\bbZ_2\times k\bbZ_2$-representation. However, it does carry the 2-cocycle $\bar c$, which lifts the sign representation of $\bbZ_2$ to the trivial one. Hence we deduce that we have $\mathbf{c}^*\otimes\mathbf{c}^*\simeq \mathbf{1}$ as well.

On the other hand, the above argument can be applied to compute the fusion 
\begin{equation}
    \mathbf{c}\otimes\mathbf{c}^* \simeq (1\xleftarrow{\hat 1}\text{sgn}) \oplus (1\xleftarrow{\hat 1} \text{sgn}) \simeq \text{sgn}\oplus \text{sgn} \simeq\mathbf{1}^*, \label{fus1*}
\end{equation}
where a non-trivial sign representation is left over due to the lack of a 2-cocycle $\bar e$ carried by the magnetic Cheshire line $\mathbf{c}^*$. Similarly, we have $\mathbf{c}^*\otimes\mathbf{c}\simeq \mathbf{1}^*$.

% Alternatively, combined with the technique outlined in the computation of  \eqref{magcheshire}, we have here
% \begin{eqnarray}
% \mathbf{1}^*\otimes\mathbf{c}&=&(\text{sgn}\oplus\text{sgn})\otimes (1\oplus_\pm\text{sgn})\nonumber\\
% &\cong& \begin{tikzcd}
% \text{sgn} \arrow[rrrr, "1", bend left] & \oplus_\pm & \text{sgn}^{2\otimes}(\simeq 1) \arrow[rrrr, "1", bend right] & \oplus & \text{sgn}& \oplus_\pm & \text{sgn}^{2\otimes}(\simeq 1) 
% \end{tikzcd}\nonumber\\
% &\simeq & \text{sgn}\oplus_\pm 1 = \mathbf{c}^*,\label{spincheshire}
% \end{eqnarray}
% which implies that $\mathbf{c}
% \otimes\mathbf{c}\simeq 1$, consistent with  \eqref{fus1}.

The above computations for \eqref{fus1}, \eqref{fus1*} rely crucially on $\bar c\neq0$. Therefore, if $\bar c=0$ were trivial, then the Cheshire strings $\mathbf{c},\mathbf{c}^*\in \operatorname{2Rep}_m(D(B\bbZ_2)^\text{sgn})$ revert to having the same fusion rules  \eqref{fus0}, \eqref{fus0*} as those in $\operatorname{2Rep}(D(B\bbZ_2)^\text{trv})$. This observation corroborates with \cite{Johnson-Freyd:2020}.

\paragraph{Fusion rules for the 2-intertwiners $i[01],i[10]$.} Now in contrast to the previous case of the invisible toric code, the coproduct component $\Delta_0$ is non-trivial for the Drinfel'd double 2-bialgebra $D(B\bbZ_2)^\text{sgn}$. By  \eqref{tensor2rep}, this induces a tensor product between the 2-representations  \eqref{2repdd} and the 2-intertwiners on them. To be concrete and for brevity, we shall concentrate on the connected component $\Gamma = \operatorname{End}_{\operatorname{2Rep}_\text{wk}(D(B\bbZ_2)^\text{sgn}}(\mathbf{1})$ in the following.

The fusion rules for the self-2-intertwiners $i[00]=i[11]=\mathfrak{1},i[00]'=i[11]'=\mathfrak{e}$ remain the same as  \eqref{fus1}, hence we shall focus on the fusion rules between $i[01],i[10]$. For convenience, we relabel these 2-intertwiners as $v_\mathbf{1},v_\mathbf{c}$ by their domains, and the goal is to directly compute the tensor product $v_\mathbf{1}\otimes v_\mathbf{c}=v_\mathbf{c}\otimes v_\mathbf{1}$ through the definition \eqref{tensor2rep}. Given \eqref{tensordecomp}, it was noted in \cite{Chen:2023tjf} that, similar to what happens in {\it Gray categories} \cite{gurski2006algebraic,neuchl1997representation}, the following two decompositions of $i\otimes j$
\begin{equation}
    v_\mathbf{1}\otimes\mathbf{1} \circ \mathbf{1}\otimes v_\mathbf{c},\qquad v_{\mathbf{c}}\otimes\mathbf{c} \circ \mathbf{c}\otimes v_\mathbf{1}\nonumber
\end{equation}
differ up to an invertible modification. This 2-isomorphism was computed in \cite{Chen:2023tjf} to be given by the weak component $\varrho=\rho_1\circ\bar e$, which in this case is determined by the 2-cocycle $\bar e\in H^2(\bbZ_2,\widehat{\bbZ_2})$ (see \eqref{moduleassoc} later).

This fact is verified after a bit of a lengthy computation. We find that, for each non-trivial $x\in\bbZ_2$ (recall the counit $\epsilon$ defines the trivial 2-representation $\rho\cong 1$),
\begin{eqnarray}
    \rho_{v_\mathbf{1}\otimes\mathbf{1}}\cdot \rho_{\mathbf{1}\otimes v_\mathbf{c}}(x) &=& \epsilon_{-1}\otimes \operatorname{id} \cong \rho_{\mathfrak{1}}, \nonumber\\
    \rho_{v_\mathbf{c}\otimes \mathbf{c}}\cdot \rho_{\mathbf{c}\otimes v_\mathbf{1}}(x) &=& (\epsilon_{-1} \otimes \rho_0(x))\cdot (\epsilon_{-1}\otimes \rho_0(x)) = (\epsilon_{-1}\otimes \rho_0(x)^2).\nonumber
\end{eqnarray}
Upon using the extension class $\bar e$, the latter indeed becomes $\rho_1(\bar e(x,x))\otimes \rho_0(x^2) = \rho_1(y)\otimes \operatorname{id} \cong \rho_\mathfrak{e}$, where $y\in\widehat{\bbZ_2}$ is the non-trivial generator. These contribute as direct summands into the tensor product, whence
\begin{equation}
    v_\mathbf{1}\otimes v_\mathbf{c} (= v_\mathbf{c}\otimes v_\mathbf{1})\simeq \mathfrak{1}\oplus\mathfrak{e}.\label{tensorv's} 
\end{equation}
This is required for the following.

\begin{theorem}\label{4dspin}
There are monoidal equivalences $$\mathfrak{F}_m:Z_1(\Sigma\operatorname{Vect}[\bbZ_2])\simeq \operatorname{2Rep}_\text{wk}(D^{\omega_b}(B\bbZ_2)),\qquad \mathfrak{F}_f:Z_1(\Sigma\mathrm{sVect}) \simeq \operatorname{2Rep}_\text{wk}(D^{\omega_f}(B\bbZ_2)$$
of fusion 2-categories.
\end{theorem}
\begin{proof}
Recall from proof of {\bf Proposition \ref{4dtoric}} that the obstruction from lifting the equivalence $\mathfrak{F}: \Sigma\operatorname{Vect}[\bbZ_2]\rightarrow \Gamma\subset\operatorname{2Rep}_\text{wk}(D(B\bbZ_2)^\text{trv})$ to a monoidal one is the component functor $\mathfrak{F}_{Ic}$ (or equivalently $\mathfrak{F}_{cI}$), where $I,c\in \Sigma\operatorname{Vect}[\bbZ_2]$ are respectively the tensor unit and the (electric) Cheshire string in $\mathscr{R}$. 

Let $\Gamma_m\subset\operatorname{2Rep}_\text{wk}(D^{\omega_b}(B\bbZ_2))$ denote the identity component. By adapting $\mathfrak{F}$ to the twisted case $\mathfrak{F}_m: \Sigma\operatorname{Vect}[\bbZ_2]\rightarrow \Gamma_m\subset \operatorname{2Rep}_\text{wk}(D^{\omega_b}(B\bbZ_2))$, we see that the fusion rule  \eqref{tensorv's} makes the component functors of $\mathfrak{F}_m$ monoidal,
\begin{eqnarray*}
    (\mathfrak{F}_m)_{Ic\to cI}(v_1v_2) &=& (\mathfrak{F}_m)_{Ic\to Ic}(\mathring{1}\oplus e) = \mathfrak{1}\oplus\mathfrak{e} \nonumber\\
    &\simeq& v_{\mathbf{1}} v_{\mathbf{c}} = (\mathfrak{F}_m)_{I\to c}(v_1)(\mathfrak{F}_m)_{c\to I}(v_2),
\end{eqnarray*}
and identically for $(\mathfrak{F}_m)_{cI\to Ic}(v_2v_1) \simeq (\mathfrak{F}_m)_{c\to I}(v_2)(\mathfrak{F}_m)_{I\to c}(v_1)$ (note the work \cite{Johnson_Freyd_2023} did not distinguish between $v_1,v_2$, so the fusion rule there is $v^2 \simeq \mathring{1} + e$). Therefore, $\mathfrak{F}_m: \Sigma\operatorname{Vect}[\bbZ_2]\rightarrow\Gamma_m$ is a monoidal equivalence.
% \begin{equation}
%      \Sigma\operatorname{Vect}[\bbZ_2] \xrightarrow{\mathfrak{F}_m|_\text{conn.}=\Sigma\Omega\mathfrak{F}}  \Sigma\Omega\Gamma\simeq \Gamma.\nonumber
% \end{equation}
% which 
Since $\mathfrak{F}_m$ and its component functors preserve all units, it extends to a monoidal equivalence $\mathfrak{F}_m: Z_1(\Sigma\operatorname{Vect}[\bbZ_2])\simeq \mathscr{R}\rightarrow \operatorname{2Rep}_\text{wk}(D^{\omega_b}(B\bbZ_2))$, as desired.

\medskip 

Now consider the fermionic case. We use the description of the braided fusion 2-category $\mathscr{S}$ describing the spin-$\bbZ_2$ gauge theory given in \cite{Johnson-Freyd:2020}. The 2-category $\mathscr{S}$ is very similar to $\mathscr{R}$: it has two identical components, with the endormophism category on the identity given by $\Omega\mathscr{S}= \text{sVect}\simeq\operatorname{Vect}[\bbZ_2]$. The fusion rules of the two components are once again graded by $\bbZ_2$. The caveat, however, is that each component are monoidally equivalent to $\Sigma\operatorname{sVect}$ instead.

In the identity component $\Sigma\operatorname{sVect}$, the Cheshire string $c\in \text{sVect}$ is the superalgebra $\text{Cl}(1)$, ie. the Clifford algebra with one odd generator. It satisfies the well-known fusion rule $c\otimes c\simeq 1$ in the ambient category $\text{sVect}$. The rest of the fusion rules are then determined by the $\bbZ_2$-grading,
\begin{equation}
    c^2  \simeq m'^2 \simeq 1,\qquad c\otimes m' \simeq m'\otimes c \simeq m,\qquad m\otimes c \simeq c\otimes m \simeq m' \nonumber
\end{equation}
Let $\Gamma_f$ denote the identity component of $\operatorname{2Rep}_\text{wk}(D^{\omega_f}(\bbZ_2))$. The 2-functor $\mathfrak{F}_f:\Sigma\text{sVect}\rightarrow \Gamma_f$, defined in the same way as in {\bf Proposition \ref{4dtoric}} and the above, the computations \eqref{fus1}, \eqref{fus1*} show that $\mathfrak{F}_f$ is monoidally essentially surjective. 

Consider $\Omega\Gamma_f = \operatorname{End}_{\operatorname{2Rep}_\text{wk}(D^{\omega_f}(\bbZ_2))}(\mathbf{1})$, whose unit is $\mathfrak{1}$. Though $\Gamma_f\not\simeq\Gamma_m$ as monoidal 2-categories, we do have $\text{sVect}\simeq \Omega\Gamma_f\simeq\Omega\Gamma_m\simeq\operatorname{Vect}[\bbZ_2]$ (only monoidally) \cite{Johnson-Freyd:2020}, hence $\mathfrak{F}_f$ is monoidally fully faithful by the same argument as the above for $\Gamma_m$. Therefore, $\mathfrak{F}_f$ extends to a monoidal equivalence $\mathfrak{F}_f:\mathscr{S}\xrightarrow{\sim}\operatorname{2Rep}_f(D(B\bbZ_2)^\text{sgn})$ as desired.
\end{proof}

% We shall prove in the following that $\operatorname{2Rep}_m(D(B\bbZ_2)^\text{sgn})$ is braided equivalent to the \textit{actual}   4d toric code.

\subsection{Proof of the main theorem}\label{proof}
Let us now look at the braiding data. We recall that the braiding in the 4d toric code $\mathscr{R} \simeq Z_1(\Sigma\operatorname{Vect}[\bbZ_2])$ is known \cite{KongTianZhou:2020} to be given by
\begin{equation}
    \beta_{X,Y}(X\otimes Y) = Y\otimes \text{sgn}_{|Y|} X,\qquad X,Y\in Z_1(\Sigma\operatorname{Vect}[\bbZ_2]) \label{magbraid}
\end{equation}
where $\text{sgn}$ is the sign representation and $|Y|\in\bbZ_2$ denotes the $\bbZ_2$-grading of the object $Y$; namely, given $Y=m,m'$ is magnetic, $\text{sgn}_{|Y|}$ acts non-trivially on the electric sector. This then gives rise, by naturality, to a non-trivial full mixed braiding  \cite{Johnson-Freyd:2020,Kong:2020wmn} 
$$\beta_{e,Y}\circ \beta_{Y,e} = -1\cdot\operatorname{id},\qquad Y=m,m^*$$ between the non-trivial 1-morphism $e\in\Omega\mathscr{R}\simeq \operatorname{Vect}[\bbZ_2]$ and the magnetic objects $m,m^*$, as required by remote detectability (see {\it Remark \ref{invisibletoric}}). 

The spin-$\bbZ_2$ gauge theory $\mathscr{S}\simeq Z_1(\Sigma\operatorname{sVect})$, on the other hand, has $\Omega\Sigma\operatorname{sVect}\simeq\operatorname{sVect}$, which has a non-trivial self-braiding $\beta_e = -1 \cdot\operatorname{id}$ for the odd object $e$ (this is what distinguishes $\operatorname{sVect}$ from $\operatorname{Vect}[\bbZ_2]$). Moreover, since the Cheshire strings $c,m'$ are now invertible, either of them be self-braided. Given that the mixed braiding maps behave the same way as in $\mathscr{R}$ (namely the only non-trivial mixed braiding maps are between $e$ and the magnetic sector, with non-trivial full-braiding), then it is one of the main results in \cite{Johnson-Freyd:2020} that \textit{only} the electric Cheshire $c$ carries a non-trivial self-braiding $\beta_c = e$ --- a non-trivial self-braiding in $m,m'$ would in fact trivialize the anomaly of $\mathscr{S}$.

\medskip

% Since now the Cheshire string $c$ is invertible in $\mathscr{S}$, there is a self-braiding morphism
% \begin{equation}
%     \beta_c: c\otimes c\rightarrow c\otimes c\nonumber
% \end{equation}
% in addition to the magnetic self-braiding $\beta_m:m\otimes m\rightarrow m\otimes m$. It was shown in \cite{Johnson-Freyd:2020} that $\beta_c\simeq e$ is in fact non-trivial. We can once again provide a direct proof of this. In fact, 

We are now in a position to prove the main theorem. 
\begin{theorem}\label{fermbraid}
The 2-functors $\mathfrak{F}_{m,f}$ in {\bf Theorem \ref{4dspin}} are $\mathrm{braided}$ equivalences.
\end{theorem}
\begin{proof}
The strategy is to simply compute all of the braiding structures in $\operatorname{2Rep}_\text{wk}(D(B\bbZ_2)^\text{sgn})$, and match them to the topological orders $\mathscr{R},\mathscr{S}$. To do this, we lift the 2-functors $\mathfrak{F}_{m,f}$ of {\bf Theorem \ref{4dspin}} to braided ones. This requires:
\begin{enumerate}
    \item For each pair of simple objects $X,Y\in \mathscr{R}$, say, the 1-morphisms $\mathfrak{F}_m(\beta_{X,Y})$ and $b_{\mathfrak{F}_m(X),\mathfrak{F}_m(Y)}$ are 2-isomorphic in $\operatorname{2Rep}_\text{wk}(D^{\omega_b}(B\bbZ_2))$, and
    \item For each object $X$ and morphism $f:Y\rightarrow Y$ in $\mathscr{R}$, the component functors $(\mathfrak{F}_m)_{XY\to Y'X}$ and $(\mathfrak{F}_m)_{YX\to XY'}$ satisfy
    \begin{equation*}
        (\mathfrak{F}_m)_{XY\to Y'X}(\beta_{X,f}) = b_{\mathfrak{F}_m(X),(\mathfrak{F}_m)_{Y\to Y'}(f)},\qquad (\mathfrak{F}_m)_{YX\to XY'}(\beta_{f,X}) = b_{(\mathfrak{F}_m)_{Y\to Y'}(f),\mathfrak{F}_m(X)}.
    \end{equation*}
\end{enumerate}
Of course, the same conditions must be met for $\mathfrak{F}_f:\mathscr{S}\rightarrow \operatorname{2Rep}_\text{wk}(D^{\omega_f}(B\bbZ_2))$. 

We shall follow the proof of {\bf Theorem \ref{braidtoric}} in order to construct the 2-$R$-matrix on $D^\omega(B\bbZ_2)$, which leads to the braiding in $\operatorname{2Rep}_\text{wk}(D^\omega(B\bbZ_2))$ through  \eqref{2repbraid1}, \eqref{2repbraid2}. We will see how each of the non-trivial 2-cocycle twists $\bar e\in H^2(\bbZ_2,\bbZ_2)$ and $\bar c\in H^2(\widehat{\bbZ_2},k^\times)$ manifest in the braiding data.

Recall the 2-$R$-matrix $(\cR,R)$ is determined by the braided transposition $\Psi$ by  \eqref{braid2rmatrix}, \eqref{braidrmatrix}. Due to the "semidirect product" structure $\widehat{\bbZ_2}\rtimes\bbZ_2$ induced by the 2-cocycle $\bar e$, the degree-0 $\bbZ_2$ acts non-trivially on the degree-(-1) $\widehat{\bbZ_2}$ by a sign $-1$. The defining relations  \eqref{psimap} then implies that 2-R-matrix $\cR$ is non-trivial:
\begin{equation}
    \cR =  (-1)^{x}\cdot y\otimes x+ x\otimes (-1)^{x}y,\qquad R = (-1)^xx\otimes x.\nonumber
\end{equation}
By \eqref{2repbraid1}, \eqref{2repbraid2}, the off-diagonal nature of these $R$-matrices witness non-trivial braiding between the electric and magnetic sectors. Indeed, $R$ acts non-trivially on 2-representations $V,W\in \operatorname{2Rep}_\text{wk}(D^\omega(B\bbZ_2))$ that differ in \textit{both} of their graded $\bbZ_2$-representations, which is only possible if $V,W$ lie in distinct sectors by \eqref{2repdd}. The sign then indicates that this braiding is non-trivial, consistent with \eqref{magbraid}.

% We  will allow us to develop non-trivial braiding maps on $D(B\bbZ_2)^\text{sgn}$. 
\begin{lemma}\label{mixeddetectability}
The 2-cocycle $\bar e\in H^2(\bbZ_2,\bbZ_2)$ leads to non-trivial \textrm{full} braiding maps between $\mathfrak{e}$ and objects $W$ in the magnetic sector. 
\end{lemma}
\begin{proof}
Recall $\bar e\in H^2(\bbZ_2,\bbZ_2)$ determines the non-trivial central extension $\bbZ_4$ of $\bbZ_2$ by itself. Provided that the component $\rho_0^0$ is non-trivial, then $\rho_0 = (\bar e\cdot \rho_0^1,\rho_0)$ furnishes a $k\bbZ_4$-representation. 

In addition, this 2-cocycle also dualizes to $\bar e\in H^2(\bbZ_2,\widehat{\bbZ_2})$, which "twists" the algebra structure in $D(B\bbZ_2)^\text{sgn}$ in the sense that 
\begin{equation}
    x\cdot (x\cdot y) = \bar e(x,x)y \neq x^2\cdot y =y, \nonumber
\end{equation}
where $x\in \bbZ_2$ and $y\in k\widehat{\bbZ_2}$. In the 2-representation 2-category $\operatorname{2Rep}_\text{wk}(D(B\bbZ_2))$, this manifests as the presence of the 2-morphism 
\begin{equation}
    \varrho(x_1,x_2) = \rho_1(\bar e(x_1,x_2)): \rho_0(x_1)\circ \rho_0(x_2) \Rightarrow \rho_0(x_1x_2),\qquad x_1,x_2\in k\bbZ_2\label{moduleassoc}
\end{equation} 
mentioned in {\bf Definition \ref{2repcat}}. This demonstrates why we {must} use the weak 2-representation theory based on $\operatorname{2Vect}^{BC}_\text{wk}$, as the strict version does contain the component $\varrho$, and hence cannot detect any twists in the 2-bialgebra $D(B\bbZ_2)$.

% Now on the other hand, the universal 2-$R$-matrix $\cR$ is reconstructed from the braided transposition $\Psi$ via  \eqref{braid2rmatrix}. We see that the defining relation  \eqref{psimap} is modified in the presence of $\bar e$,
% \begin{equation}
%     \bar e(x,x) = x\cdot (x\cdot 1) = \cdot\circ \Psi_{-1}^r(x^2\otimes 1) = \cdot\circ \Psi_{-1}^r(1\otimes 1),\nonumber
% \end{equation}
% and similarly for $\Psi_{-1}^l$. As such, $\Psi_{-1}^{l,r}$ are equipped with a factor of $\bar e$ in addition to just the Pontrjagyn duality. This factor $\bar e$ carries over to the 2-$R$-matrix $\cR^\pm$ via  \eqref{braid2rmatrix}. 

Recall  \eqref{swaptwiners} that $\mathfrak{e}$ swaps the grading of the 2-representation spaces, and hence $\bar e$ will occur only in the full mixed braiding $B_{W,\mathfrak{e}}=B_{\mathfrak{e},W}=b_{W,\mathfrak{e}}\cdot b_{\mathfrak{e},W}$ between $\mathfrak{e}$ and those 2-representations $W$ that carry a non-trivial sign representation \textit{in degree-(-1)} --- namely the magnetic sector in  \eqref{2repdd}. The other full mixed braiding maps being trivial. A simple computation then gives
\begin{equation}
    B_{W,\mathfrak{e}}: \rho_0^0(\cR^+_{(2)}) \rho_0^0(\cR^-_{(1)}) \Rightarrow \rho_0^0(\cR^+_{(2)}\cR^-_{(1)}) = 1,
\end{equation}
which is precisely the map $\varrho_1(x,x) = \rho_1(\bar e(x,x))\simeq -1$ from  \eqref{moduleassoc}. In other words, the $\bbZ_2$-particle $\mathfrak{e}$ braids non-trivially with the magnetic sector $\mathbf{1}^*,\mathbf{c}^*$, as required.

\end{proof}

% Now what of the braiding maps on the objects  \eqref{2repdd}? 

\begin{lemma}\label{fermiparticle}
The 2-cocycle $\bar c\in H^2(\widehat{\bbZ_2},k^\times)$ gives the non-trivial self-braiding $b_\mathfrak{e}=-1$. Moreover, the self-braiding $b_\mathbf{c}$ is non-trivial in $\operatorname{2Rep}_\text{wk}(D^{\omega_f}(B\bbZ_2))$, but $b_{\mathbf{c}^*},b_{\mathbf{1}^*}$ are trivial.
\end{lemma}
\begin{proof}
Consider the first statement. By naturality, the braiding maps $b_{i,j}$ on 1-morphisms $i,j$ can be decomposed into mixed braiding maps,
\begin{equation}
    b_{i,j} = b_{i,W}b_{V,j},\qquad \begin{cases} i: V\rightarrow U \\ j: W\rightarrow T\end{cases}.\nonumber
\end{equation}
Taking $i=j=\mathfrak{e}$ and the identity endomorphism $\mathfrak{1}^*:W\rightarrow W$ on a magnetic line, we see that
\begin{eqnarray}
    b_\mathfrak{e} &=& b_{\mathfrak{e},\mathfrak{1}^*}b_{\mathfrak{1}^*,\mathfrak{e}} = (b_{\mathfrak{e},W}b_{W,\mathfrak{e}})(b_{W,\mathfrak{e}}b_{\mathfrak{e},W})\nonumber\\
    &=& B_{\mathfrak{e},W}B_{W,\mathfrak{e}} = (\rho_1(\bar e(x,x)))^2 = \bar c(y,y) \cdot \operatorname{id}= -1\cdot\operatorname{id}\nonumber
\end{eqnarray}
from the definition of $\bar c$ in \eqref{ferm2cocy} and the fact that $B_{\mathfrak{e},W} = \bar e$ from the above lemma. Here, note the extension cocycle $\bar e$ satisfies $\bar e(x,x)=y$ for the non-trivial generators $x\in \bbZ_2,y\in \widehat{\bbZ_2}$. This is consistent with the observation that $\bar c$ implements the fermionic statistics of the $\bbZ_2$-charged particle in \cite{Wen:2019,Zhu:2019,Johnson-Freyd:2020}.

\smallskip
% This would have had trivial consequence if it were not for the fact that $\bar e$ also dualizes to a class $\bar e\in H^2(\bbZ_2,\widehat{\bbZ_2})$. From the definition  \eqref{braid} of the component $\Psi_{-1}^l$ of the braiding map $\Psi$, we have
% \begin{equation}
%     \bar e(x,x)y = (x^2)\rhd y = (\cdot \rhd\cdot)\Psi_{-1}^l(x^2\otimes y) = (\cdot \rhd \cdot)\Psi_{-1}^l(1\otimes y).\nonumber
% \end{equation}
% If $\Psi_{-1}^l$ does not carry the dual $\hat{\bar e} = p\circ \bar e\circ p\in H^2(\widehat{\bbZ_2},\bbZ_2)$ of $\bar e$ (namely if $\Psi_{-1}^l$ is just given by the self-duality map $\sigma$), then the right-hand side would have been trivial $1\rhd p(y)= 1\rhd \hat x=1$, which is in conflict with the left-hand side. This is amended by taking $\Psi_{-1}^l = \hat{\bar e}\cdot \sigma$, such that
% \begin{equation}
%     (\cdot \rhd\cdot)\Psi_{-1}^l(x^2\otimes y) = \hat{\bar e}(\hat y,\hat y)\hat x.\nonumber
% \end{equation}
% Therefore, a non-trivial mixed braiding map $b_{\mathfrak{e}W}$ is induced which implements the relation
% \begin{equation}
%     (\hat{\bar e}\cdot \sigma)\otimes \bar e\cdot\sigma=\Psi_{-1}^l \otimes \bar e\cdot \Psi_{-1}^r \Rightarrow \Psi_{-1}^l\cdot \bar e\otimes\Psi_{-1}^r=(\hat{\bar e}\cdot \sigma)\cdot \bar e \otimes \sigma;\nonumber
% \end{equation}
% note we have neglected the inconsequential prefator of $\rho_V\otimes\rho_W$.

Consider the second statement. Since $\bar e$ also determines a central extension of $D(B\bbZ_2)_0 =\bbZ_2$ by itself, an analogous argument as the previous lemma shows that, provided the 2-representation $\rho_0$ has the non-trivial sign representation at degree-0 (ie. the Cheshire string $\mathbf{c}$ or the magnetic vacuum line $\mathbf{1}^*$), then the self-braiding 
\begin{equation}
    b_{V}: \rho_0^0(R_{(1)})\rho_0^0(R_{(2)}) \Rightarrow \rho_0^0(R_{(1)}R_{(2)})=1 \nonumber
\end{equation}
can carry the non-trivial 1-morphism $\rho_0(\bar e(x,x))\simeq\mathfrak{e}$. In particular, this establishes that $b_{\mathbf{c}^*}\simeq \mathfrak{1}$ is trivial while $b_\mathbf{c}\simeq\mathfrak{e}$ is not.

% The only However, if $\bar c=0$ then  \eqref{fus0} states that $\mathbf{c},\mathbf{c}^*$ are non-invertible, and hence cannot be self-braided. It is therefore only in the fermionic case $\operatorname{2Rep}_f^\tau(D(B\bbZ)^\text{sgn})$ where the electric Cheshire string $\mathbf{c}$, say, has a non-trivial self-braiding $b_{\mathbf{c}}\simeq \mathfrak{e}$.

But what about the magnetic vacuum $\mathbf{1}^*$? The above argument does not force $b_{\mathbf{1}^*}$ to be trivial, but the fusion rule  \eqref{magcheshire} (in the form $\mathbf{c}\otimes\mathbf{c}^*\simeq \mathbf{1}^*$) and \eqref{ribboneq} do. Since the magnetic Cheshire $\mathbf{c}^*$ is bosonic, the \textit{full} braiding $B_{\mathbf{c}^*,\mathbf{c}}\simeq b_{\mathbf{c}} \simeq \mathfrak{e}$ must be non-trivial. Using this along with  \eqref{magcheshire} and the previous result then gives
\begin{eqnarray}
    b_{\mathbf{1}^*} &=& b_{\mathbf{c}^*\otimes\mathbf{c}} \nonumber\\
    &\cong& b_{\mathbf{c}^*,\mathbf{c}} \circ (b_{\mathbf{c}^*}\otimes b_{\mathbf{c}}) \circ b_{\mathbf{c},\mathbf{c}}\nonumber\\
    &\simeq& (b_{\mathbf{c}^*}\otimes b_{\mathbf{c}}) \circ B_{\mathbf{c}^*,\mathbf{c}}\nonumber\\
    &\simeq& \mathfrak{1}\otimes\mathfrak{e}\otimes\mathfrak{e}\simeq \mathfrak{1},\nonumber
\end{eqnarray}
hence the magnetic vacuum $\mathbf{1}^*$ must have trivial self-braiding $b_{\mathbf{1}^*}=\mathfrak{1}$.

Of course, in the absence of $\bar c$, the braiding maps considered above are all trivial. 
\end{proof}

These lemmas demonstrate that the non-trivial braiding data in $\mathscr{R}$ (resp. $\mathscr{S}$) also appear in $\operatorname{2Rep}_\text{wk}(D^{\omega_b}(B\bbZ_2))$ (resp. $\operatorname{2Rep}_\text{wk}(D^{\omega_f}(B\bbZ_2))$), and identifies them from the 2-cocycle twists $\bar e,\bar c$ present in $D^\omega(B\bbZ_2)$. Moreover, there are no more non-trivial braiding maps in the latter than in the former: for instance, we have noted in {\bf Lemma \ref{mixeddetectability}} how the full mixed braiding $B_{\mathfrak{e},V}$ with {\it electric} lines $V$ is trivial in $\operatorname{2Rep}_\text{wk}(D^\omega(B\bbZ_2))$. 

In other words, given the presence of the appropriate 2-cocycles $\bar e,\bar c$, the 2-functors $\mathfrak{F}_m,\mathfrak{F}_f$ respects the braiding data in the sense that we have described in the beginning of the proof. This establishes the desired braided equivalences.

% self-braiding data $b_{\mathbf{c}}$

% It now suffices to check that it also preserves the trivial ones. To do this, we must ensure that no other braiding maps --- aside from the above ones --- occur non-trivially in $\operatorname{2Rep}_\text{wk}(D^\omega(B\bbZ_2))$. We will without loss of generality focus on braiding data that does not involve the tensor units, such as $\mathbf{1}\in\Gamma$ or $\mathfrak{1}\in\Omega\Gamma$.

% since the $\bbZ_2$-representation in degree-0 of the tensor product $\mathfrak{e}\otimes V$ is trivial due to a grading swap in $\mathfrak{e}$, which leads to the absence of the 2-cocycle $\bar e$ in $B_{\mathfrak{e}V}$. Now 

\end{proof}

To further drive home the point of the main result {\bf Theorem \ref{fermbraid}}, we shall recover the 5-dimensional cobordism invariant associated to the spin $\bbZ_2$-gauge theory $\mathscr{S}$ from the spin-Kitaev model. Recall the expressions of $\bar e(A) = \frac{1}{2}\operatorname{Sq}^1A$ and $\bar c(B,B)=\frac{1}{2}\operatorname{Sq}^2B$ in terms of the Steenrod square. Starting from the partition function  \eqref{fermkit}, 
\begin{equation}
    Z_{\text{Kit}}^s(X)\sim \sum_{\substack{dA=0 \\ dB=\tau}} e^{i2\pi\int_X  B\cup \frac{1}{2}\operatorname{Sq}^1A + \frac{1}{2}\operatorname{Sq}^2B},\nonumber
\end{equation}
we deduce that, given $W$ is a 5-dimensional manifold with boundary $X=\partial W$, the bulk partition function takes the form \cite{Zhu:2019}
\begin{equation}
    Z_{\text{Kit}}^s(X)\sim \exp\left[i\pi \int_W  \tau(A)\cup \operatorname{Sq}^1A + \operatorname{Sq}^2\tau(A)\right]\nonumber
\end{equation}
on-shell of the EOM $dA=0,dB=\tau(A)$.

By interpreting the on-shell gauge fields $(A,B)$ (ie. satisfying $dA=0,dB=\tau(A)$) as a classifying map $f = (A,B):W\rightarrow BD(B\bbZ_2)$ \cite{Zhu:2019,Kapustin2017}, we can introduce group cohomology classes 
\begin{equation}
    E\in H^3(\bbZ_2,\widehat{\bbZ_2}),\qquad M\in H^2(\bbZ_2,\bbZ_2)\nonumber
\end{equation}
such that $f^*E=\tau(A)$ and $f^*M = \frac{1}{2}\operatorname{Sq}^1A=\bar e(A)$. Then, the spin-Kitaev partition function can be written as
\begin{equation}
     Z_{\text{Kit}}^s(X) \sim \sum_{f\in [W,B\cK]}([W],f^*\alpha),\nonumber
\end{equation}
where $[W]\in H_5(W,\bbC^\times)$ is the fundamental homology class and $\alpha$ is a degree-5 group cohomology class given by
\begin{eqnarray}
    \alpha = (-1)^{\operatorname{Sq}^2E+ E\cup M}\in H^5(\bbZ_2[3]\times\bbZ_2[2],\bbC^\times).\label{fermclass}
\end{eqnarray}
This is precisely the anomaly of the fermionic phase $\mathscr{S}$ \cite{Johnson-Freyd:2020}.

\paragraph{The $w_2w_3$ gravitational anomaly.} The reader may notice that we have conveniently left out the study of the anomalous version $\mathscr{T}$ of the fermionic order $\mathscr{S}$. This is because $\mathscr{T}$ is does not admit a description in terms of a Drinfel'd centre \cite{Johnson_Freyd_2023,decoppet2023drinfeld}, hence it may not be straightforward to construct a corresponding 2-bialgebra description. We will leave this for a future work, and say more about this in the conclusion.

\section{Conclusion}
Following the construction of Drinfel'd double 2-bialgebras, we have applied the structural results proven in \cite{Chen:2023tjf} to the case of the   4d Kitaev model based on the 2-group associated to $\bbZ_2$. We explicitly computed the associated 2-representation 2-category $\operatorname{2Rep}_\text{wk}(D(B\bbZ_2))$ and shown that it satisfies formally the formula
\begin{equation}
    Z_1(\Sigma\operatorname{Vect}[\bbZ_2]) \simeq \operatorname{2Rep}_\text{wk}(D(B\bbZ_2)),\label{2dd2cat}
\end{equation}
where $Z_1$ is the Drinfel'd centre; cf. {\bf Theorem \ref{4dspin}, \ref{fermbraid}}. This directly categorifies the characteristic relation 
\begin{equation}
    Z_1(\operatorname{Rep}(G)) \simeq \operatorname{Rep}(D(G)) \nonumber
\end{equation}
of the quantum double of Drinfel'd \cite{Drinfeld:1986in} (and, more generally, of Majid \cite{Majid:1996kd}), at least in the case $G=\bbZ_2$. Our results can be concisely summarized in the following table.

\begin{table}[h]
    \centering
    \begin{tabular}{|c|c|c|c|}
    \hline
    Gapped phase & N/A &   4d toric code & spin-$\bbZ_2$ gauge theory \\
    \hline 
    2-representations & $\operatorname{2Rep}_\text{wk}(D(B\bbZ)^\text{trv})$ & $\operatorname{2Rep}_\text{wk}(D^{\omega_b}(B\bbZ_2))$ & $\operatorname{2Rep}_\text{wk}(D^{\omega_f}(B\bbZ_2))$ \\
    2-category in \cite{Johnson-Freyd:2020} & N/A & $\mathscr{R}$ & $\mathscr{S}$ \\
    DW cocycle in \cite{Zhu:2019} & $\omega(A,B)=0$ & $\omega(A,B) = \frac{1}{2}BA^2$ & $\omega(A,B) = \frac{1}{2}BA^2 + \frac{1}{2}\operatorname{Sq}^2B$ \\
    \hline
    \end{tabular}
\end{table}
\noindent We have also been able to concretely identify the 4d 2-group Dijkgraaf-Witten NLSMs constructed in \cite{Zhu:2019} that host the 2-category $\operatorname{2Rep}_\text{wk}(D(B\bbZ_2))$ as charges. These explicit equivalences provide an explicit and rigorous identification between the 2-categorical and field theoretic descriptions of the associated gapped 4d topological phases \cite{Wen:2019}.

This hints towards a categorified notion of "Tannaka-Krein reconstruction", in which certain braided fusion 2-categories $\mathcal{C}$ can be equivalently described 
\begin{equation}
    \cC\simeq \operatorname{2Rep}_\text{wk}(\cA^\omega), \nonumber
\end{equation}
as the 2-representation 2-category of a (possibly twisted) quasitriangular 2-Hopf algebra/2-bialgebra $\cH$. Such a Tannakian duality is worthwhile to have in this context, as it (i) condenses the data of a (braided) monoidal 2-category $\cC$ to those of the underlying 2-bialgebra $\cA$, and (ii) it allows one to directly construct the action of the underlying TFT given the 2-categorical data $\cC$ of a gapped 4d topological phase.

Aside from applications to condensed matter theory, we expect that these higher bialgebraic structures would see fruitful applications in high-energy physics, higher-dimensional conformal field theory and integrable systems. 

\begin{remark}
    In fact, we can make an even bolder and refined conjecture. Suppose $\cC = Z_1(\Sigma\cB)$ is the Drinfel'd centre of $\Sigma\cB$, where $\Sigma$ is the condensation functor \cite{Gaiotto:2019xmp,Johnson_Freyd_2023} and $\cB$ is a braided fusion category satisfying the Tannakian reconstruction formula
    \begin{equation}
        \cB\simeq \operatorname{Rep}(H^\tau).\nonumber
    \end{equation}
    Here, $\tau$ is a algebraic 3-cocycle on a certain (not necessarily grouplike) quasitriangular semisimple (1-)Hopf algebra $H$, and $H^\tau$ is the corresponding twisted version \cite{Majid:1996kd} (ie. a braided quasi-bialgebra of Drinfel'd \cite{book-quasihopf}). The conjecture is then that $\cC = Z_1(\Sigma\cB)$ is braided equivalent to $\operatorname{2Rep}_\text{wk}(D^\omega(BH))$ where $\tau$ now becomes the Postnikov datum of the 2-bialgebra $BH$, which is constructed in a manner similar to Section \ref{2grpbialg}, and $\omega$ is a 4-cocycle twist of the skeletal Drinfel'd double $D(BH)$. Such a correspondence would provide a purely algebraic interpretation of the condensation completion functor $\Sigma$.
\end{remark}

% It is still not easy to find these data, however, even in the grouplike case $\cA = kG$, since one must employ spectral sequences methods to compute $H^4(G,U(1))$ \cite{Kapustin2017,Zhu:2019}.

It is interesting to note that both of the graded components of $D(B\bbZ_2)$ contribute to determine,  at each level, the structures of the 2-representation 2-category --- the objects (point excitations) are not solely determined by the degree-0 piece $D(G)_0$ of the Drinfel'd double 2-bialgebra, for instance. The underlying graded duality structure of the Drinfel'd double 2-bialgebra play a significant role in the categorified Tannakian duality. Indeed, the self-duality of the Drinfel'd double seem to be very closely related to the Morita self-duality of the Drinfel'd centre \cite{decoppet2023drinfeld} in general, hence it may be possible to study Morita theory for 2-categories \cite{decoppet2022morita} by studying (higher-)module theory of 2-bialgebras.

Aside from purely categorical endeavours, our results also pave the way towards the construction of a more complete and refined 4-dimensional topological invariant \cite{Mackaay:ek,Mackaay:hc,Douglas:2018}. It may guide us in the exploration of more intricate higher structures in tensor 2-categories, such as pivotality \cite{Douglas:2018}, ribbon or even modularity. 

\medskip

We discuss some open questions which we find very interesting and important to tackle in the future.

\paragraph{Higher-ribbon structures and modular tensor 2-categories.} As mentioned at the end of Section \ref{proof}, the fermionic order $\mathscr{T}$ still currently eludes us in terms of the above 2-Hopf algebraic treatment. This order is very closely related to the $w_2w_3$ gravitational anomaly \cite{Johnson-Freyd:2020,JuvenWang}; this gravitational anomaly has a proposed topological field theory description by the so-called "fermionic quasistrings order" \cite{Thorngren2015,Chen:2022hct}.

In any case, the order $\mathscr{T}$ is known to be distinct from $\mathscr{S}$ as fusion 2-supercategories \cite{Johnson-Freyd:2020,Johnson_Freyd_2023}. As mentioned in {\it Remark 3.4} of \cite{Johnson_Freyd_2023}, this can be understood as the difference between the self-duality datum they host for the magnetic line $m$, which concerns their spherical structures. Since the spherical data for braided fusion structures are described equivalently by ribbon twists via the Drinfle'd isomorphism (see eg. \cite{Delaney2020BraidedZA} for 1-categories), it would be interesting to understand the distinction between $\mathscr{S},\mathscr{T}$ through {\it ribbon} 2-Hopf algebras and their 2-representations.

More precisely, given a (possibly infinite-dimensional) quasitriangular 2-Hopf algebra $(\cA,\Delta,S,\cT,\cR)$, one may seek to equip it with a central {\it ribbon element} $\nu = (\nu_{-1},\nu_0) \in \cA$ in analogy with the 1-Hopf algebra case \cite{Majid:1996kd,etingof2016tensor}. Its (possibly non-finite non-semsimple) braided tensor 2-category $\operatorname{2Rep}_\text{wk}(\cA)$ of 2-representation would then be equipped with a ribbon structure, given by the twist 1- and 2-morphisms $$\theta_V: V\mapsto \nu_0\cdot V,\qquad \theta_i: i\mapsto \nu_{-1}\cdot i,$$ satisfying some appropriate higher-analogue of the ribbon equations, as well as additional coherences. Moreover, the modular data of (possibly {\it non} finite semisimple) ribbon tensor 2-categories --- such as $\operatorname{2Rep}_\text{wk}(\cA)$ for an infinite-dimensional ribbon 2-Hopf algebra $\cA$ --- could be used to construct a 4d version of the Reshetikhin-Turaev TQFT.

% It is known that if $(H,\Delta,S,R,\nu)$ is a {\it ribbon Hopf algebra}, in the sense that the central element $\mathfrak{v}=\mathfrak{u}S(\mathfrak{u})\in H$ given by $\mathfrak{u} = \cdot(S\otimes1)(R^T)\in H$ admits a square-root $\nu\in H$, then its representation category $\operatorname{Rep}(H)$ is a ribbon category. The ribbon twist $\theta_V(V) = \nu\cdot V$ on $V\in\operatorname{Rep}(H)$ is given by the action of $\nu$, and satisfies the following ribbon equation
% \begin{equation*}
%     \theta_{VW} = b_{V,W}b_{W,V}(\theta_V\otimes \theta_W),\qquad \theta_{\bar V} = \bar{\theta}_V
% \end{equation*}
% for all $V,W\in\operatorname{Rep}(H)$. Now 

Regardless if such 4d TQFTs can produce novel invariants of 4-manifolds (see eg. \cite{Reutter:2020bav} for an answer in the negative for semsimple TQFTs), however, such higher ribbon algebras are nevertheless interesting to investigate due to the higher modular data they host. In analogy with the 1-category case, it would be interesting to see if such modular tensor 2-categories model, in an appropriate sense, representations of 3-dimensional vertex operator algebras (VOAs) such as the "Raviolo VOA" defined in \cite{Garner:2023zqn} based on topological-holomorphic field theories \cite{Gwilliam:2021zkv}.

% which satisfy certain coherence conditions against the braiding and mixed-braiding maps $b$ on $\operatorname{2Rep}(\cA)$ endowed by the 2-$R$-matrix. Given that $\operatorname{2Rep}(\cA)$ forms a braided monoidal 2-category, we may take the 2-representation 2-category $\operatorname{2Rep}(\cA)$ of a ribbon 2-Hopf algebra $\cA$ to model and define in general a {\it ribbon 2-category}. Such structures would then be proposed to describe the vertex operator algebra of {\it 3-dimensional} conformal field theories.

\paragraph{Lattice realizations; 2-groupoid algebras and the 4d tube algebra.} The main result of this paper properly pins down the continuum field theory description of the given   4d gapped topological phase. We have said nothing, however, about how one may UV complete and construct a lattice realization (or a class thereof). To do so, one must construct the lattice Hamiltonian and cast the excitations encoded in the field theory as (extended) local symmetric operators on the lattice. 

We know how this is done in 3-dimensions. The 3d lattice theory is given by {\it string-net condensation} \cite{Wen:2010gda}, labeled by an underlying structure group $G_0$. This Hilbert space is equipped with a "gluing operation", which combines two local string-nets along the lattice edges. This forms the {\bf Ocneanu tube algebra}, which can be modeled by representations of the groupoid algebra $\bbC[\Lambda G_0]$ of the {\it inertia groupoid} $\Lambda G_0$ of $G_0$. However, it is known also that for finite groups $G_0$, the (twisted) Drinfel'd double $D(G_0)$ is isomorphic\footnote{For the twists, the group 3-cocycle on $D(G_0)$ is sent to a groupoid 2-cocycle on $\Lambda G_0$ via transgression.} to the (twisted) groupoid algebra $\bbC[\Lambda G_0]$ \cite{Willerton:2008gyk}, and as such the  3d tube algebra is given precisely by the representation theory of $D(G_0)$ \cite{Delcamp:2016yix}.

As always, we wish to categorify the above results. The construction of a   4d analogue of the tube algebra has in fact already appeared in \cite{Bullivant:2021,Bartsch:2023wvv}, where \textit{2-groupoid algebras} and the representation theory thereof was studied. This allows one to construct the lattice Hamiltonian, as well as the local symmetric operators, of a given   4d topological phase; see also \cite{Delcamp:2023kew} from the categorical perspective. The missing link here is of course the correspondence between the Drinfel'd double 2-bialgebra $D(G)$ and the 2-groupoid algebra $\bbC[\Lambda G]$ of the inertia 2-groupoid $\Lambda G$. Evidence to suggest such a correspondence stems from the fact that both 2-representations of 2-groupoid algebras and those of the Drinfel'd double 2-bialgebra describe gapped topological phases in   4d. Moreover, 2-groupoid algebras admit twists by 2-group 4-cocycles \cite{Bullivant:2021} just as the Drinfel'd double 2-bialgebra does, as we have demonstrated in Section \ref{ferkitcharges}. If this correspondence can be made explicit, then we can construct lattice realizations of   4d field theories directly from the underlying Drinfel'd double 2-bialgebra symmetry and its 2-representations.

% \paragraph{Categorified Witt equivalence.} With the 

\newpage

\appendix

\section{The Drinfel'd double $D(B\bbZ_n)$ and the $\bbZ_n$-gauge theory}\label{Zp}
In this section, we apply the computation techniques outlined in the main text to the case of the $\bbZ_n$-Drinfel'd double 2-bialgebra  
\begin{equation*}
    D(B\bbZ_n) = \widehat{\bbZ_n}\rightarrow \bbZ_n.
\end{equation*}
Similar to the $n=2$ case, the 2-Hopf algebra structure is given in \eqref{grouplike}, \eqref{dualaction}, and $D(B\bbZ_n)$ is once again self-dual. The argument in {\it Remark \ref{2alg-2grp}} and Section \ref{weak2bialgstructure} can be applied to determine the associator morphisms in the (weak) 2-representation 2-category $\operatorname{2Rep}_\text{wk}(D(B\bbZ_n))$ through the Postnikov class $\tau\in H^3(\bbZ_n,\widehat{\bbZ_n})$.

In analogy with \eqref{toric}, \eqref{kit}, the associated 4d topological action and its partition function are given by
\begin{equation}
    S[A,B] = \frac{1}{n} \int_X \langle B\cup dA\rangle,\qquad Z_{\text{Kit}^0}(X) = \sum_{A,B} e^{i2\pi S[A,B]}.\nonumber
\end{equation}
Note that, due to the coefficient of $\frac{1}{n}$ in front of the topological action, the flatness EOM $dA=0$ implies that $A$ is a mod-$n$ cocycle. The excitations in this theory are labeled by the 2-representations of $D(B\bbZ_n)$; we can once again form the table of irreducible 2-representations in the following,
\begin{table}[h]
    \centering
    \hspace*{-0.35cm}
    \begin{tabular}{|c|c|c|c|c|c|}
    \hline
         & Vacuum & $1$-Cheshire & $2$-Cheshire & \dots & ${n-1}$-Cheshire \\
         \hline
    Electric: & $\mathbf{1}=1\xrightarrow{1} 1$ & $\mathbf{c}_1 = 1 \xrightarrow{0}\zeta$  & $\mathbf{c}_2 = 1\xrightarrow{0}\zeta^2$ & \dots & $\mathbf{c}_{n-1}=1\xrightarrow{0}\zeta^{n-1}$ \\
    $1$-Magnetic: & $\mathbf{1}^1=\zeta\xrightarrow{1}\zeta $& $\mathbf{c}^1_1=\zeta \xrightarrow{0} \zeta^2$ & $\mathbf{c}^1_2=\zeta\xrightarrow{0}\zeta^3 $&\dots & $\mathbf{c}^{1}_{n-1}=\zeta \xrightarrow{0}1$\\
    $2$-Magnetic: & $\mathbf{1}^2=\zeta^2\xrightarrow{1}\zeta^3 $& $\mathbf{c}^2_1=\zeta^2 \xrightarrow{0} \zeta^4$ & $\mathbf{c}^2_2=\zeta^2\xrightarrow{0}\zeta^5 $&\dots & $\mathbf{c}^{2}_{n-1}=\zeta^2 \xrightarrow{0}\zeta$\\
    \vdots &\vdots &\vdots &\vdots &\vdots &\vdots  \\
    ${n-1}$-Magnetic:  & $\mathbf{1}^{n-1}=\zeta^{n-1}\xrightarrow{1} \zeta^{n-1}$ & $\mathbf{c}^{n-1}_1 = \zeta^{n-1} \xrightarrow{0} 1$  & $\mathbf{c}^{n-1}_2 = \zeta^{n-1}\xrightarrow{0}\zeta$ & \dots & $\mathbf{c}^{n-1}_{n-1}=\zeta^{n-1}\xrightarrow{0}\zeta^{-2}$ \\
    \hline 
    \end{tabular}
    % \caption{Caption}
    % \label{tab:my_label}
\end{table}

\noindent where $\zeta \in k^\times$ is a $n$-th root of unity labeling the irreducible representations of $\bbZ_n$ over $k = \bbC$. Here, As before, the component $\rho_1$ is understood to be trivial for the vacuum lines, and the map $\rho_1=\hat 1$ for the Cheshire lines. Of course, we recover the list \eqref{2repdd} for $n=2$. 

\medskip

The hope one has is that the above table should give a complete labeling of the \textit{distinct} simple objects of $\operatorname{2Rep}_\text{wk}(D(B\bbZ_n))$. This is not the case, however --- it contains too many things! In general, one can find invertible 2-intertwiners between them. We shall explicitly construct these 1-isomorphisms in the following, and give a correct characterization of the isomorphism classes of simple objects in $\operatorname{2Rep}_\text{wk}(D(B\bbZ_n))$.

% \begin{remark}
%     The hope is that this table enumerates the isomorphism classes simple objects in $\operatorname{2Rep}_\text{wk}(D(B\bbZ_n))$. However, there are in some sense "too many" Cheshire lines in that table.
%     Then there  such that the tensor power $(\mathbf{c}^a_b)^{b^{-1}\otimes}$ contains the fusion identity $\mathbf{1}$, hence $\mathbf{c}^a_b$ is not simple. 
% \end{remark}

\subsection{2-categorical structures of $\operatorname{2Rep}_\text{wk}(D(B\bbZ_n))$}
The structure of the 1- and 2-morphisms can be analyzed analogously by employing the technique used in Section \ref{invtoricfusion}, namely by repeatedly checking the commutativity of the diagrams in \eqref{2int}. In particular, there are no non-trivial 1-morphisms of the following forms 
\begin{equation*}
    \mathbf{1}^a\rightarrow\mathbf{1}^b,\qquad \mathbf{1}^a\rightarrow\mathbf{c}^b_c
\end{equation*}
for $a\neq b\in\bbZ_n$. On the other hand, since the differential $\partial=0$ is trivial for the Cheshire 2-representations, this argument does not apply to 1-morphisms among the $\mathbf{c}$'s themselves.

For $n>2$, one can write down potentially non-trivial intertwining chain maps of the form $$\bar i^{ac}[bd] = \bar i^{ac}[bd]_1 \oplus \bar i^{ac}[bd]_0 =  0\oplus 1:\mathbf{c}^a_b \rightarrow\mathbf{c}^c_d,$$ where $a+b  = c+d \mod n$ and $a\neq c$. If these were to form genuine 1-morphisms in $\operatorname{2Rep}_\text{wk}(D(B\bbZ_n))$, they must satisfy the compatibility condition
\begin{equation*}
     0= \bar i^{ac}[bd]_{1}\circ \rho_1(y)=\rho_1(y)\circ \bar i^{ac}[bd]_0 = \rho_1(y),\qquad \forall~y\in\widehat{\bbZ_n}
\end{equation*}
against the chain homotopies $\rho_1:\widehat{\bbZ_n}\rightarrow \operatorname{Hom}(V_0,V_{-1})$. This clearly forces $\rho_1=0$, which is inconsistent with the fact that the Cheshire lines must carry $\rho_1=\hat 1\neq 0$. Notice that the presence of the non-trivial group $\widehat{\bbZ_n}$ in degree-(-1) in the 2-bialgebra $D(B\bbZ_n)$ is key in ruling out non-trivial 1-morphisms between the distinct sectors.

This allows us to conclude that no non-trivial 1-morphisms exist across the different sectors, similar to the $n=2$ case studied above. In other words, $\operatorname{2Rep}_\text{wk}(D(B\bbZ_n))$ splits into a direct sum of $n$ sectors, each of which "looks the same": there is one vacuum and $n-1$ (possibly isomorphic; see later) Cheshire lines. This result is consistent with the known structure of Drinfel'd centre 2-categories \cite{Kong:2020wmn,KongTianZhou:2020},
\begin{equation}
    Z_1(\operatorname{2Vect}^{KV}[\bbZ_n]) \simeq \bigoplus_{k=1}^n\operatorname{2Rep}(\bbZ_n).\label{2repZn}
\end{equation}
We have of course not proven that each of the sectors we have are equivalent to $\operatorname{2Vect}^{KV}[\bbZ_n]$. This question, as well as the related question of the difference between $\operatorname{2Vect}^{hBC}$ and $\operatorname{2Vect}^{KV}$, shall be addressed in an upcoming work by the author.

% \begin{remark}
%     Here we comment  Indeed,   
% \end{remark}

\medskip

As distinct sectors differ by merely a global multiplication by some power of $\zeta$, we can without loss of generality focus on the electric sector $\Gamma\subset\operatorname{2Rep}_\text{wk}(D(B\bbZ_n))$ (namely the connected component of the identity). By a similar argument as in the main text, there are non-trivial 1-morphisms between vacuuum line and each of the Cheshire lines given by
\begin{equation*}
    i[0b] = 1 \oplus 0: \mathbf{1}\rightarrow\mathbf{c}_b,\qquad i[b0] = 0 \oplus 1: \mathbf{c}_b\rightarrow \mathbf{1},
\end{equation*}
for $b=1,2,\dots,n-1$. The endomorphism category $\Omega\Gamma = \operatorname{End}(\mathbf{1})$ on the tensor unit $\mathbf{1}$ is characterized as in \eqref{looping} to be equivalent to $\operatorname{Rep}(\bbZ_n)\simeq\operatorname{Vect}[\bbZ_n]$, as one expects. 

When $n>2$, computing the 1-morphisms between the Cheshire lines themselves, on the other hand, involves some deeper number-theoretic facts and deserves a more detailed study. In particular, the degree-swapping 1-morphisms \eqref{swaptwiners} are in general no longer endomorphisms. Indeed, such maps $i'[0b]:(w,v)\mapsto  \zeta^b\cdot (v,w)$ lands the $-b$-Cheshire line $\mathbf{c}_{-b}\cong 1\oplus\zeta^{-b}\ni (w,v)$ in $1\oplus\zeta^{b}\cong \mathbf{c}_{b}$ for $b\in\bbZ_n$. Moreover, if $b$ is coprime to $n$, then $b\in\bbZ_n$ admits a multiplicative inverse $b^{-1}$ and the degree-swap 1-morphism $i'[0b]$ is invertible $i'[b^{-1}0] = i'[0b]^{-1}$. The set of such numbers forms the multiplicative group $\bbZ_n^\times$, and we have the following. 

\begin{lemma}
The isomorphism classes of the Cheshire line $\mathbf{c}_d,~1\leq d\leq n-1$ in $\Gamma$ (and similarly in any other $a$-magnetic sector) are labeled by its orbit under $\bbZ_n^\times$, and each of the degree-swapping 1-morphisms descends to an endomorphism on each orbit.    
\end{lemma}

In particular, if $n=p$ were a prime, then $\bbZ_p^\times\cong\bbZ_p$ forms a field and we only have one unique Cheshire line, denoted $\mathbf{c}_{p-1}\in\Gamma$, and its endomorphism category $\operatorname{End}(\mathbf{c}_{p-1})\simeq\operatorname{Vect}[\bbZ_p]$ is generated by $\bbZ_p$ itself. Hence, the structure of $\Gamma$ in this case is very similar to \eqref{gammadiag}:
\begin{equation}
   \Gamma = \begin{tikzcd}
\mathbf{1} \arrow[rr, "{\operatorname{Vect}}", bend left] \arrow["{\operatorname{Vect}[\bbZ_p]}"', loop, distance=2em, in=215, out=145] &  & \mathbf{c} \arrow[ll, "{\operatorname{Vect}}", bend left] \arrow["{\operatorname{Vect}[\bbZ_p]}"', loop, distance=2em, in=35, out=325]
\end{tikzcd}.\label{gammadiag-p}
\end{equation}
We shall denote the simple objects in the endomorphism categories $\operatorname{Vect}[\bbZ_p]$ by $\{\mathfrak{1},\mathfrak{e}_1,\dots,\mathfrak{e}_{p-1}\}$. 

This form of $\Gamma$ is precisely the structure of the known (KV) 2-representation 2-category $\operatorname{2Rep}_{\bbZ_p}$ of $\bbZ_p$ \cite{Douglas:2018,Delcamp:2021szr,Bartsch:2022mpm}, in line with {\it Remark \ref{wk2algrepthy}}. And since we have $p$ copies of $\Gamma$ labeled by the magnetic sectors, each of which we recall have no non-trivial 1-morphisms between them, we achieve the following generalization of {\bf Proposition \ref{4dtoric}}.
\begin{proposition}
    There is a $\mathrm{non}$-$\mathrm{monoidal}$ equivalence between $\operatorname{2Rep}_{\text{wk}}(D(B\bbZ_p))$ and $Z_1(\operatorname{2Vect}^{KV}(\bbZ_p))$.
\end{proposition}
\begin{proof}
    This follows from the above lemma, as well as \eqref{2repZn} from \cite{Kong:2020wmn,KongTianZhou:2020}. The reason this equivalence is non-monoidal also stems from the fact that the tensor product of the 1-morphism $v$ is trivial in $\operatorname{2Rep}_{\text{wk}}(D(B\bbZ_p))$, while it is not in $Z_1(\Sigma\operatorname{Vect}[\bbZ_p])$.
\end{proof}

As explained in {\bf Theorem \ref{braidtoric}}, the {\it untwisted} 2-representation 2-category $\operatorname{2Rep}_\text{wk}(D(B\bbZ_p)^\text{trv})$ is symmetric and hence does not describe a gapped topological phase. In the following, we shall focus our attention on the twisted case, involving a 2-cocycle $\bar e\in H^2(\bbZ_p,\widehat{\bbZ_p})$. 

\subsection{Twists by $\bbZ_p$-cocycles}
As explained in Section \ref{ferkitcharges}, the 2-cocycles $\bar e,\bar c$ studied in the main text twist the 2-bialgebra structure of the underlying Drinfel'd double 2-bialgebra, and hence manifests in the 2-representation theory. In general, these 2-cocycles can be identified as contributions to certain 2-group 4-cocycles \cite{Kapustin2017} by \eqref{4cocy}. In the current case, these 2-cocycles are defined on $\bbZ_p$; in particular, the analogue of $\bar c\in H^2(\widehat{\bbZ_p},k^\times)$ would, intuitively from {\bf Lemma \ref{fermiparticle}}, determine a $p$-anyonic statistics of the particles in the theory. 

However, from the general theory of classification of 3+1d gapped topological phases \cite{Wen:2019,Wen:2010gda}, point particles should only have bosoinc or fermionic statistics, and for odd primes $p$ there are no $\bbZ_2$ subgroups of $\bbZ_p$ that can give rise to an "emergent" fermionic statistics through the projection $H^2(\widehat{\bbZ_2},k^\times)\rightarrow H^2(\widehat{\bbZ_p},k^\times)$. For this reason, we shall focus our attention on the 4-cocycle twist $\omega\in H^4(D(B\bbZ_p),k^\times)$ defined only with $\bar e\in H^2(\bbZ_p,\widehat{\bbZ_p})$, and denote the corresponding 2-representation 2-category by $\operatorname{2Rep}_\text{wk}(D^\omega(B\bbZ_p))$. 

\medskip

%  In order to examine this, we need an explicit description of $\bar e$.

In contrast to the $p=2$ case, here we have multiple choices for a non-trivial 2-cocycle $\bar e$. Standard computations demonstrate that there is a bijection $H^2(\bbZ_p,\bbZ_p)\cong \bbZ_p$ \cite{book-groupcohomology}, whence 2-cocycles $\bar e=\bar e_m$ are labeled by an element $m\in\bbZ_p$. From the standard cyclic resolution of $\bbZ_p$, one finds an explicit formula 
\begin{equation*}
    \bar e_m(a,b) = \begin{cases} 0 &; a+b < p \\ m &; a+b\geq p\end{cases},\qquad a,b\in\bbZ_p
\end{equation*}
for the group 2-cocycles. If we let $\Theta_m$ denote the central extension of $\bbZ_p$ by $\bbZ_p$ determined by $\bar e_m$, then we see that the Cheshire strings $\mathbf{c}^a_{p-1}$ carry a $\Theta_m$-extension instead of a $\bbZ_p\times\widehat{\bbZ_p}$-representation, in analogy with the $p=2$ case studied in Section \ref{ferkitcharges}. Further, through the Pontrjagyn isomorphism $\bbZ_p\cong \widehat{\bbZ_p}$, the 2-cocycle $\bar e_m$ (with the same label $m\in\widehat{\bbZ_p}$) also appears in the component $\rho_1$ of these 2-representations.

As with typical central extensions, there is a sense in which $\Theta_m$ can be viewed as a "semidirect product" $\bbZ_p\rtimes_m \bbZ_p$, with the action $$x^a\rhd_m (x^b\rhd_m y) = \begin{cases} y &; a+b < p \\ \zeta^m\cdot y &; a+b\geq p\end{cases},\qquad \bbZ_p=\langle x\rangle,\quad \widehat{\bbZ_p}=\langle y\rangle$$ given by multiplying a $p$-th root of unity $\zeta^m\in\bbZ_p\subset U(1)$. This view then allows us to interpret the term $\bar e_m(A)$ that appears in the topological action $S$ as a quadratic term $A\cup_m A$, where $\cup_m = \rhd_m \circ\cup$ is a composition of this action and the usual cup product. This then introduces a cohomological term
\begin{equation*}
    \int_X \langle B\cup \bar e_m(A)\rangle =\int_X \langle B\cup (A\cup_m A)\rangle 
\end{equation*}
into the topological action. The corresponding partition function of this {\it 4d $\bbZ_p$-gauge theory} then reads
\begin{equation*}
    Z_{\text{Kit}^p}(X) \sim \sum_{\substack{dA=0 \\ dB=\tau(A)}} e^{i2\pi \int_X \langle B\cup \bar e_m(A)\rangle},
\end{equation*}
which we call the {\bf 4d $\bbZ_p$-toric code}. We briefly remark here that it may be possible to express $\bar e_m(A)=\beta(A)$ in terms of the $\bbZ_p$-Bockstein homomorphism $\beta: H^1(X,\bbZ_p)\rightarrow H^2(X,\bbZ_p)$, which specializes to the Steenrod square $\operatorname{Sq}^1$ in \eqref{Z2stablecohomo} for the $p=2$ case. We shall not say more about this here.

\subsection{Fusion and braiding in the $\bbZ_p$-toric code}
As outlined in the general theory of Section \ref{tensor2repthy}, the fusion and braiding data encoded in $\operatorname{2Rep}_\text{wk}(D^{\omega_b}(B\bbZ_p))$ are determined by the structure of the graded coproduct and 2-$R$-matrices \cite{Chen:2023tjf}. We shall denote $\omega=\bar e$ throughout the following.

Since we still only have a single Cheshire line (up to equivalence) in the current case, the same computations in \eqref{fus0}, \eqref{fus0*} implies the following $\bbZ_p$-graded fusion rules 
\begin{equation*}
    \mathbf{1}^a \otimes\mathbf{1}^b \simeq\mathbf{1}^{a+b},\qquad \mathbf{1}^a \otimes \mathbf{c}_{p-1}^b \simeq \mathbf{c}_{p-1}^a\otimes\mathbf{1}^b\simeq \mathbf{c}_{p-1}^{a+b},\qquad \mathbf{c}_{p-1}^a\otimes \mathbf{c}_{p-1}^b\simeq \mathbf{c}_{p-1}^{a+b}\oplus \mathbf{c}_{p-1}^{a+b}    
\end{equation*}
where $a,b\in\bbZ_p$ with $a,b=0$ denoting the electric sector. This justifies the name "Cheshire lines" for the simple objects $\mathbf{c}_{p-1}$.

Now let us consider the tensor product between 1-morphisms. From the structure of $\Gamma$ \eqref{gammadiag-p}, we wish to compute $v\otimes v$ in terms of the 2-cocycle $\bar e_m$, where $v$ is the simple object in $\operatorname{Hom}(\mathbf{1},\mathbf{c}_{p-1})\simeq\operatorname{Vect}$. Recall from \ref{tensor2rep} that $v\otimes v$ is determined by the square of the coproduct $\Delta_0(x)$, in which the 2-cocycle $\bar e_m$ appears as the interchange 2-morphism \eqref{interchange}. Due to the self-duality of $D(B\bbZ_p)$, the coproduct $\Delta_0$ is dual to precisely the action $\rhd_m$ of $\bbZ_p$ on $\widehat{\bbZ_p}$, which is given by multiplication by a cyclic element of order $m\in\bbZ_p$. Therefore, if $m\in\bbZ_p$ is such that $2m \geq p$, then a non-trivial factor $\bar e_m(m,m)=m\neq 0$ appears in the square of the coproduct $\Delta_0(x)$, and hence the 1-morphisms $v:\mathbf{1}\rightarrow\mathbf{c}_{p-1}$ will satisfy the fusion rule
\begin{equation}
    v\otimes v\simeq  \mathfrak{1} \oplus \mathfrak{e}_m.\nonumber
\end{equation}
Recall that $\mathfrak{e}_m$ denotes the "$m$-th" simple object in $\Omega\Gamma\simeq\operatorname{Vect}[\bbZ_p]$, where $m=1,\dots,p-1$. This result reduces to \eqref{tensorv's} when $p=2$ and $m=1$. However, for $p>2$ it is curious to note that there exist $1\leq m\leq p-1$ for which $2m<p$, whence $v\otimes v\simeq \mathfrak{1}$ remains trivial in $\operatorname{2Rep}_\text{wk}(D^\omega(B\bbZ_p))$.

\medskip

Now let us consider the braiding. We know from \eqref{braidrmatrix}, \eqref{braid2rmatrix} that the $R$-matrices governing the braiding on $\operatorname{2Rep}_\text{wk}(D^\omega(B\bbZ_p))$ can be deduced from the braided transposition map $\Psi$. For the current case, the braided transpositions must carry the "semidirect product" action $\rhd_m$, which translates to the following forms of the $R$-matrices
\begin{equation*}
    R = \zeta^m \cdot x\otimes x,\qquad \cR = \zeta^m\cdot y\otimes x^m+ x\otimes \zeta^m\cdot y,\qquad m=1,\dots,p-1,
\end{equation*}
where $x,y$ are the generators. This is really an abuse of notation, since the phase $\zeta^m$ appears only when the $R$-matrices here act on 2-representations with sufficient degree. For instance, we have
\begin{equation*}
    b_{V,W} = \text{flip}\circ R\cdot (V\otimes W) = \begin{cases} W\otimes V &; |V_0| + |W_{-1}| < p \\ \zeta^m\cdot (W\otimes V) &; |V_0|+ |W_{-1}|\geq p\end{cases},
\end{equation*}
where $|V_0|,|V_{-1}|$ denotes the degree of the graded $\bbZ_p$-representations in $V$. As an explicit example, if $V=\mathbf{c}_{p-1}$ is the electric Cheshire, then $|V_{-1}|=0$ and $|V_0| = |\zeta^{p-1}| = p-1$; in this case, provided $|W_{-1}|\neq 0$, then a phase of $\zeta^m$ will appear in the braiding $b_{V,W}$. In other words, the electric Cheshire will braid non-trivially with any of the non-trivially graded object $W^a$,
\begin{equation*}
    b_{\mathbf{c}_{p-1},W^a}(\mathbf{c}_{p-1}\otimes W^a) = W^a\otimes \zeta^m\cdot\mathbf{c}_{p-1}= W^a\otimes (a\rhd_m \mathbf{c}_{p-1}),
\end{equation*}
where we have rewritten a multiplication by the phase $\zeta^m$ in terms of the action $a\rhd_m$. This braiding formula can be seen as a $\bbZ_p$-graded generalization of \eqref{magbraid}.

A similar argument as in the proof of {\bf Theorem \ref{mixeddetectability}} then leads to the following result.
\begin{theorem}
    In the 4d $\bbZ_p$-toric code, the electric $\bbZ_p$-flavoured bosons $\mathfrak{e}_k$ have non-trivial full braidings with any of the $a$-magnetic objects $W^a$,
    \begin{equation*}
        B_{\mathfrak{e}_k,W^a} = \zeta^m\cdot \id,\qquad \forall a,k=1,\dots,p-1,
    \end{equation*}
    where $\id:\mathfrak{e}_k\otimes W^a\Rightarrow\mathfrak{e}_k\otimes W^a$ denotes the identity 2-morphism.
\end{theorem}
\noindent This can be understood as the statement of remote detectability in the $\bbZ_p$-toric code.

\medskip

We conclude this paper by briefly noting that the 2-category $\operatorname{2Rep}_\text{wk}(D^\omega(B\bbZ_p))$ is distinct from the 2-category $\operatorname{2Rep}_{\bbZ_p[1]\rtimes\bbZ_p}$ of 2-representations of the split 2-group $\mathbb{G} = (\bbZ_p,\bbZ_p)$ \cite{Bartsch:2022mpm}. This is due to the fact that the self-duality and the braiding properties of $D(B\bbZ_p)$, as a Drinfel'd double 2-bialgebra, play non-trivial roles in determining the structure of $\operatorname{2Rep}_\text{wk}(D^\omega(B\bbZ_p))$. 

In any case, one is left wondering if there is a more efficient method to determine the 2-categorical structure of 2-representations in general, as opposed to a brute-force direct computation. Investigation along this line, through the approach of categorical characters \cite{Ganter:2014}, have been initiated by the author and some collaborators.

\newpage
\bibliographystyle{Biblio}
\bibliography{biblio}
\end{document}